\documentclass[11pt]{article} 
\usepackage{fullpage}
\usepackage{amssymb,latexsym,mathptmx,amsmath}
\usepackage{epsfig,graphicx}
\usepackage{float}
\usepackage{leqno}
\usepackage[scaled]{helvet}
\usepackage{mathrsfs}
\usepackage{url}
\usepackage{hyperref}
\usepackage{textcomp} 
\usepackage{parskip}
\usepackage{alltt}
\usepackage{verbdef}
\usepackage{xspace}
\usepackage{verbatim}
\usepackage{enumitem}
\usepackage{lipsum}
\usepackage{wrapfig}
\usepackage{pgf,pgfrcs}
\usepackage{alg}

\usepackage{adjustbox} 
\usepackage{array}
\usepackage{booktabs}
\usepackage{multirow}
\usepackage{pifont}

\setlist[itemize]{leftmargin=*}
\newcounter{nr}

{\begin{list}{}{\setlength{\rightmargin}{0em}}\item[]}{\end{list}}
\newcommand{\qed}{\Box}
\newenvironment{proof}{{\it Proof:}}{\hfill $\qed$ \\[2mm]}
\newenvironment{prooff}{{\it Proof:}}{}

\newtheorem{thm}{Theorem}
\newtheorem{example}{Example}[section]
\newtheorem{definition}[example]{Definition}
\newtheorem{lemma}[example]{Lemma}
\newtheorem{proposition}[example]{Proposition}

\newcommand{\true}{\textit{true}}
\newcommand{\false}{\textit{false}}

\newcommand{\mkset}[1]{\{{#1}\}}

\newcommand{\upd}[3]{{#1}[{#2} \mapsto {#3}]}
\newcommand{\stoone}[2]{\{{#1} \mapsto {#2}\}}
\newcommand{\stotwo}[4]{\{{#1} \mapsto {#2},\ {#3} \mapsto {#4}\}}
\newcommand{\fv}[1]{\mathit{fv}(#1)}
\newcommand{\dom}[1]{\mathit{dom}(#1)}
\newcommand{\Zz}{\mathbb{Z}}
\newcommand{\dlt}{\Delta}

\newcommand{\ie}{{\it i.e.}}
\newcommand{\eg}{{\it e.g.}}
\newcommand{\etc}{{\it etc.}}

\newcommand{\cf}{{\it cf.}}
\newcommand{\etal}{{\it et al.}}

\newcommand{\labv}[1]{\mathit{Lab}(#1)}
\newcommand{\usev}[1]{\mathit{Use}(#1)}
\newcommand{\defv}[1]{\mathit{Def}(#1)}
\newcommand{\skipv}{\mathbf{Skip}}
\newcommand{\assignv}[2]{{#1} := {#2}}
\newcommand{\rassignv}[2]{{#1} := \mathbf{Random}(#2)}
\newcommand{\rassignvs}[2]{{#1} := \mathtt{R}(#2)}
\newcommand{\observev}[1]{\mathbf{Observe}(#1)}
\newcommand{\observevs}[1]{\mathtt{Obs}(#1)}

\newcommand{\retv}[1]{\mathbf{Return}(#1)}
\newcommand{\retvs}[1]{\mathtt{Ret}(#1)}

\newcommand{\startv}{\mathtt{Start}}
\newcommand{\observevv}{\mathbf{Observe}}
\newcommand{\finalv}{\mathtt{End}}

\newcommand{\unodes}{\mathcal{V}}
\newcommand{\PD}{\ensuremath{\mathsf{PD}}}
\newcommand{\LAP}[2]{\ensuremath{\mathsf{LAP}({#1},{#2})}}
\newcommand{\LAPf}{\ensuremath{\mathsf{LAP}}}

\newcommand{\chain}[2]{\{{#1} \mid {#2}\}}
\newcommand{\limit}[2]{\mathit{lim}_{#1 \rightarrow \infty}\,{#2}}
\newcommand{\contarrow}{\rightarrow_{c}}
\newcommand{\fixed}[1]{\mathit{fix}({#1})}

\newcommand{\uvar}{{\cal U}}
\newcommand{\adds}[2]{{#1} \oplus {#2}}
\newcommand{\restrs}[2]{{#1}\!\mid_{{#2}}}
\newcommand{\sagree}[3]{{#1} \stackrel{{#3}}{=} {#2}}
\newcommand{\sto}[1]{\mathcal{S}(#1)}
\newcommand{\fulls}{\mathcal{F}}

\newcommand{\Dist}{\ensuremath{{\cal D}}}
\newcommand{\sumd}[1]{\sum {#1}}
\newcommand{\dagree}[3]{{#1} \stackrel{{#3}}{=} {#2}}
\newcommand{\Dijr}[3]{\ensuremath{D_{{#1},{#2}}^{{#3}}}}
\newcommand{\cnst}[4]{\ensuremath{c_{{#1},{#4}}^{{#2,#3}}}}

\newcommand{\selectf}[1]{\ensuremath{\mathsf{select}_{{#1}}}}
\newcommand{\assignf}[2]{\ensuremath{\mathsf{assign}_{{#1}:={#2}}}}
\newcommand{\rassignf}[2]{\ensuremath{\mathsf{rassign}_{{#1}:={#2}}}}
\newcommand{\semee}{[\![\;]\!]}
\newcommand{\seme}[1]{[\![{#1}]\!]}

\newcommand{\HH}[1]{{\cal H}_{{#1}}}
\newcommand{\hpd}[3]{{#1}^{({#2},{#3})}}
\newcommand{\hh}[2]{\hpd{h}{{#1}}{{#2}}}

\newcommand{\ddsym}{\overset{dd}{\rightarrow}}
\newcommand{\ddssym}{\overset{dd}{\rightarrow}^{*}}
\newcommand{\dd}[2]{#1 \ddsym #2}
\newcommand{\dds}[2]{#1 \ddssym #2}
\newcommand{\DD}{\mathtt{DD}}
\newcommand{\DDS}{\mathtt{DD}^{*}}
\newcommand{\DDclose}{\mathtt{DD}^{\mathrm{close}}}

\newcommand{\fppd}[1]{1PPD(#1)}
\newcommand{\nextq}[2]{\mathit{next}_{{#1}}(#2)}
\newcommand{\rv}[2]{\mathit{rv}_{{#1}}(#2)}
\newcommand{\BSP}{\mathtt{BSP}}
\newcommand{\LWS}{\mathtt{LWS}}
\newcommand{\PN}{\mathtt{PNV?}}
\newcommand{\ESS}{\mathtt{ESS}}

\newcommand{\dt}[1]{\index{#1}\textbf{\emph{#1}}} % defined term

\def\today{\ifcase\month\or
  January\or February\or March\or April\or May\or June\or
  July\or August\or September\or October\or November\or December\fi
  \space\number\day, \number\year}

\begin{document}

\title{A Theory of Slicing for Probabilistic Control-Flow Graphs\thanks{Expanded and revised version of a paper originally appearing in \emph{Foundations of Software Science and Computation Structures - 19th International Conference, FOSSACS 2016}.}
}
% Control Flow Graphs

 \author{
\begin{tabular}{cc}
Torben Amtoft & Anindya Banerjee \\
Kansas State University & IMDEA Software Institute \\
Manhattan, KS, USA & Madrid, Spain \\
\texttt{tamtoft@ksu.edu} & \texttt{anindya.banerjee@imdea.org}
\end{tabular}
}

%% \thanks{Research supported by the US National Science Foundation (NSF). Any opinion, findings, and conclusions or recommendations expressed in this material are those of the authors and do not necessarily reflect the views of NSF.}

\maketitle

\begin{abstract}
We present a theory for slicing probabilistic imperative programs ---containing random assignments, and ``observe'' statements (for conditioning) --- represented as probabilistic control-flow graphs (pCFGs) whose nodes modify probability distributions.
%% ; this is consistent (as we show in a companion article) with the standard probabilistic semantics for structured languages.
We show that such a representation allows direct adaptation of standard machinery such as data and control dependence, postdominators, relevant variables, \etc~to the probabilistic setting. We separate the specification of slicing from its implementation: first we develop syntactic conditions that a slice must satisfy; next we prove that any such slice is semantically correct; finally we give an algorithm to compute the least slice. To generate smaller slices, we may in addition take advantage of knowledge that certain loops will terminate (almost) always.
A key feature of our syntactic conditions is that they involve two disjoint slices such that the variables of one slice are \emph{probabilistically independent} of the variables of the other. This leads directly to a proof of correctness of probabilistic slicing. In a companion article we show adequacy of the semantics of pCFGs with respect to the standard semantics of structured probabilistic programs.

\end{abstract}

\section{Introduction}
The task of program slicing \cite{Weiser:TSE-1984,Tip:SlicingSurvey-1995}
is to remove the parts of a program that
are irrelevant in a given context. This paper addresses slicing of
probabilistic imperative programs which, in addition to the usual control
structures, contain ``random assignment'' and ``observe'' (or conditioning)
statements. The former assign random values from a given
distribution to variables. 
The latter remove undesirable combinations of values,
a feature which can be used to bias (or condition) the variables
according to real world observations. 
The excellent survey by Gordon~\etal~\cite{Gor+etal:ICSE-2014} 
depicts how probabilistic programs can be used in a variety of contexts, such as: encoding applications
from machine learning, biology, security; representing probabilistic models (Bayesian Networks, Markov Chains); estimating probability distributions through probabilistic inference algorithms (like the Metropolis-Hastings algorithm for Markov Chain Monte Carlo sampling); \etc

Program slicing of deterministic imperative programs is increasingly
well understood \cite{Pod+Cla:TSE-1990,Bal+Hor:AAD-1993,Ran+Amt+Ban+Dwy+Hat:TOPLAS-2007,Amtoft:IPL-2007,Danicic+etal:TCS-2011}.
A basic notion is that if the slice contains a
program point which depends on some other program points then 
these also should be included in the slice; here ``depends'' typically
encompasses data dependence and control dependence.  
However, Hur~\etal~\cite{Hur+etal:PLDI-2014}
recently demonstrated that in the presence of random assignments and
observations, standard notions of data and control dependence no
longer suffice for semantically correct (backward) slicing.  They 
develop a denotational framework in which they prove correct an
algorithm for program slicing. In contrast, this paper shows how \emph{classical
notions of dependence can be extended to give a semantic
foundation for the (backward) slicing of probabilistic programs}.  
The paper's key contributions are:
\begin{itemize}
\item
A formulation of probabilistic slicing in terms of probabilistic control-flow graphs (pCFGs) (Section~\ref{sec:CFG}) that allows direct adaptation of standard machinery such as data and control dependence,  postdominators, relevant variables, \etc~to the probabilistic setting.
We also provide a novel operational semantics of pCFGs 
%% probabilistic control flow graphs
(Section~\ref{sec:sem}): the semantic function $\hpd{\omega}{v}{v'}$ transforms
a probability distribution at node $v$ into a probability distribution at node $v'$,
so as to model what happens when
``control'' moves from $v$ to $v'$ in the control-flow graph.
\item
Syntactic conditions for correctness (Section~\ref{sec:conditions}) that in a non-trivial way extend classical work on program slicing~\cite{Danicic+etal:TCS-2011} and whose key feature is that they involve \emph{two} disjoint slices; in order for the first to be a correct final result of slicing, the other must contain any ``observe'' nodes sliced away and all nodes on which they depend. We show that the variables of one slice are \emph{probabilistically independent} of the variables of the other,
and this leads directly to the correctness of probabilistic slicing 
(Theorem~\ref{thm:slicing-correct} in Section~\ref{sec:correct}).
\item
An algorithm (Section~\ref{sec:alg-least}), 
with running time at most cubic in the size
of the program, that (given an approximation of which loops terminate with probability 1) computes the best possible slice 
in that it is contained in any other (syntactic) slice of the program. 
\end{itemize}
Our approach separates the specification of
slicing from algorithms to compute the best possible slice. The former is concerned with defining what is a correct syntactic slice, such that the behavior of the sliced program is equivalent to that of the original.
The latter is concerned with how to compute the best possible syntactic slice; this slice is automatically a semantically correct slice ---no separate proof is necessary.

A program's
behavior is its final probability distribution; we demand equality
modulo a constant factor so as to allow the removal of
``observe'' statements that do not introduce any bias in the final
distribution. This will be the case if the variables tested
by ``observe'' statements are independent, in the sense of
probability theory, of the variables relevant for the final value.

Compared to the conference version of this paper~\cite{Amt+Ban:FoSSaCS-2016},
the additional contributions are:
\begin{itemize}
\item
We allow to slice away certain loops
if they are known (through some analysis, or an oracle)
to terminate with probability 1.
\item
In a companion article ~\cite{Amt+Ban:ProbSemantics-2017} we establish  adequacy of the operational semantics of pCFGs 
%% probabilistic control flow graphs
%% is adequately related
with respect to the ``classical'' semantics~\cite{Gor+etal:ICSE-2014,Hur+etal:PLDI-2014}
for structured probabilistic programs. 

\end{itemize}

%% \paragraph{Outline.}
%% In Section~\ref{sec:examples}, we shall motivate our approach by means of
%% several examples. 
%% In Section~\ref{sec:CFG} we shall
%% introduce control flow graphs (CFGs for short)
%% that are probabilistic 
%% and in Section~\ref{sec:sem} we shall define their semantics; 
%% in Section~\ref{sec:consistent}
%% we shall show that for  structured programs this semantics
%% (applied after converting into CFGs) 
%% is consistent with
%% their usual semantics

%% In Section~\ref{sec:conditions} we shall develop (mostly) syntactic
%% conditions for what
%% constitutes a correct slice, and in 
%% Section~\ref{sec:correct} we shall show (Theorem~\ref{thm:slicing-correct})
%% that these conditions are sufficient to ensure a suitable version
%% of semantic correctness. 
%% Section~\ref{sec:alg-least} presents an algorithm
%% for computing the least (syntactic) slice.
%% Section~\ref{sec:future} discusses extensions and future work,
%% while Section~\ref{sec:conclude} concludes.

We prove all non-trivial results;
some proofs are in the main text
but most are relegated to Appendix~\ref{app:proofs}.
Our development is based on domain theory whose basic concepts 
we recall in Appendix~\ref{app:domain}.

\section{Motivating Examples}
\label{sec:examples}
% The goal of this section is to illustrate when it is correct to slice a probabilistic program into another probabilistic program.

\subsection{Probabilistic programs}
Whereas in deterministic languages, a variable has only one value at a given
time,
we consider a language where a variable may 
have many different values at a given time, each with a certain probability.
(Determinism is a special case where
one value has probability one, and all others have probability zero.)
We assume, to keep our development simple,
that each possible value is an integer.
A more general development, somewhat orthogonal to the aims of this article,
would allow real numbers and would 
employ measure theory (as explained in~\cite{Panangaden-2009});
we conjecture that much will extend naturally 
(with summations becoming integrals).
% On p.13 in Prakash's book it says: ``In situations with a countable set of possible states we can indeed take all sets to be measurable and much of the subtleties of measure theory can be dispensed with''. 

Similarly to \cite{Gor+etal:ICSE-2014},
probabilities are introduced by the construct
$\rassignv{x}{\psi}$ which assigns to variable $x$ a 
value with probability given
by the random distribution $\psi$ which in our setting is
a mapping from $\Zz$ (the set of integers) to $[0,1]$ such that
$\displaystyle \sum_{z \in \Zz} \psi(z) = 1$.
A program phrase modifies
a distribution into another distribution, where a distribution
assigns a probability to each possible store.
This was first formalized by Kozen~\cite{Kozen:JCSS-81} 
in a denotational setting.
As also in \cite{Gor+etal:ICSE-2014},
we shall use the construct $\observev{B}$ to ``filter out''
values which do not satisfy the boolean expression $B$. 
That is, the resulting distribution
assigns zero probability to all stores not satisfying $B$,
while stores satisfying $B$ keep their probability.

\subsection{The examples} 
Slicing amounts to picking a set $Q$ of ``program points''
(satisfying certain conditions as we shall soon discuss)
and then removing the program points not in $Q$
(as we shall formalize in Section~\ref{subsec:sem-slice}).
The examples all use a random distribution $\psi_4$ 
over $\mkset{0,1,2,3}$ where
$\psi_4(0) = \psi_4(1) = \psi_4(2) = \psi_4(3) = 
\frac{1}{4}$
whereas $\psi_4(i) = 0$ for $i \notin \{0,1,2,3\}$. 
The examples all consider whether it is correct to
let $Q$ contain exactly $\rassignv{x}{\psi_4}$ and $\retv{x}$,
and thus slice into a program $P_x$ with
straightforward semantics:  after execution, the probability
of each possible store is given by the distribution $\dlt'$ defined as
$\dlt'(\stoone{x}{i}) = \mathbf{\frac{1}{4}}$ if $i \in \{0,1,2,3\}$;
otherwise $\dlt'(\stoone{x}{i}) = 0$.

%\begin{exmp}
\begin{example}
\label{ex1}
Consider the program $P_1$ given by
%% \begin{tabbing}
%% kkk\=kkk\kill
%% \>1: \=$\rassignv{x}{\psi_4}$ \\
%% \>2: \>$\rassignv{y}{\psi_4}$ \\
%% \>3: \>$\observev{y \geq 2}$ \\
%% \>4: \>$\retv{x}$ 
%% \end{tabbing}
\[
\begin{array}{ll}
1: & \rassignv{x}{\psi_4} \\
2: & \rassignv{y}{\psi_4} \\
3: & \observev{y \geq 2} \\
4: & \retv{x} 
\end{array}
\]
%\noindent
The distribution produced by the first two assignments
will assign probability 
$\displaystyle \frac{1}{4} \cdot \frac{1}{4} = \frac{1}{16}$
to each possible store $\stotwo{x}{i}{y}{j}$ with $i,j \in \{0,1,2,3\}$.
In the final distribution $D_1$, a store
$\stotwo{x}{i}{y}{j}$ with $j < 2$ is impossible,
and for each $i \in \{0,1,2,3\}$ there are thus only two possible stores 
that associate $x$ with $i$: the store $\stotwo{x}{i}{y}{2}$, and
the store $\stotwo{x}{i}{y}{3}$.
Restricting to the variable $x$ that is ultimately returned,
\begin{displaymath}
D_1(\stoone{x}{i}) = 
\sum_{j=2}^3 D_1(\stotwo{x}{i}{y}{j}) =  
\frac{1}{16} + \frac{1}{16} = \mathbf{\frac{1}{8}} 
\end{displaymath}
if $i \in \{0,1,2,3\}$ (otherwise, $D_1(\stoone{x}{i}) = 0$).
We see that the probabilities in $D_1$ do \emph{not} add up to 1
which reflects that the purpose of an $\observevv$ statement is to
cause undesired parts of the current distribution at that node to ``disappear'' 
(which may give certain branches more relative weight than other branches). 
We also see that $D_1$ equals $\dlt'$ except for a constant factor:
$D_1 = 0.5 \cdot \dlt'$. 
That is, $\dlt'$ gives the same \emph{relative} distribution
over the values of $x$ as $D_1$ does. 
%% (An alternative way of phrasing this
%% is that the ``normalized'' semantics, 
%% cf.~what is done in \cite{Gor+etal:ICSE-2014},
%% gives the same result for $P_x$ as for $P_1$.)
We shall therefore say that $P_x$ is a \emph{correct} slice of $P_1$.
%\end{exmp}
\end{example}

Thus the $\observevv$ statement is irrelevant to the final relative
distribution of $x$.
This is because $y$ and $x$ are \emph{independent} in $D_1$,
as formalized in Definition~\ref{def:indpd}.

%\begin{exmp}
\begin{example}
\label{ex2}
Consider the program $P_2$ given by

%% \begin{tabbing}
%% kkk\=kkk\kill
%% \>1: \=$\rassignv{x}{\psi_4}$ \\
%% \>2: \>$\rassignv{y}{\psi_4}$ \\ 
%% \>3: \>$\observev{x+y \geq 5}$ \\ 
%% \>4: \>$\retv{x}$
%% \end{tabbing}
\[
\begin{array}{ll}
1: & \rassignv{x}{\psi_4} \\
2: & \rassignv{y}{\psi_4} \\ 
3: & \observev{x+y \geq 5} \\ 
4: & \retv{x}
\end{array}
\]
Here the final distribution $D_2$ allows only 3 stores:
$\stotwo{x}{2}{y}{3})$,
$\stotwo{x}{3}{y}{2})$ and
$\stotwo{x}{3}{y}{3})$,
all with probability $\displaystyle \frac{1}{16}$,
and hence
$\displaystyle D_2(\stoone{x}{2}) = \frac{1}{16}$ and 
$\displaystyle D_2(\stoone{x}{3}) = \frac{1}{16} + \frac{1}{16} = \frac{1}{8}$.
Thus the program is biased towards $x$ having value $2$ or $3$;
in particular we cannot write $D_2$ in the form $c \dlt'$.
Hence it is \emph{incorrect} to slice $P_2$ into $P_x$.
%\end{exmp}
\end{example}
%\noindent

In this example, $x$ and $y$ are \emph{not} independent in $D_2$;
this is as expected since the $\observevv$ statement in $P_2$ depends on
something (the assignment to $x$) on which the returned variable $x$
also depends.
%\begin{exmp}
\begin{example}
\label{ex3}
Consider the program $P_3$ given by 

%% \begin{tabbing}
%% kkk\=\kill
%% \>1: \=$\rassignv{x}{\psi_4}$ \\
%% \>2: \>{\bf if}$\; x$\=$ \geq 2$ \\
%% \>3: \>        \>$\rassignv{z}{\psi_4}$ \\
%% \>4: \>        \>$\observev{z \geq 3}$ \\
%% \>5: \>$\retv{x}$
%% \end{tabbing}
\[
\begin{array}{ll}
1: & \rassignv{x}{\psi_4} \\
2: & \mathbf{if}\; x \geq 2 \\
3: & \quad         \rassignv{z}{\psi_4} \\
4: & \quad         \observev{z \geq 3} \\
5: & \retv{x}
\end{array}
\]
Since three quarters of the distribution disappears when $x \geq 2$,
$P_3$ is biased in that it is more likely to return $0$ or $1$ than
$2$ or $3$; in fact, the
final distribution $D_3$ is given by
$D_3(\stoone{x}{i}) 
= \mathbf{\frac{1}{4}}$ when $i \in \{0,1\}$ and
$D_3(\stoone{x}{i}) = D_3(\stotwo{x}{i}{z}{3}) = \mathbf{\frac{1}{16}}$ when $i \in \{2,3\}$. 
(And, when $i \notin \{0,1,2,3\}$, $D_3(\stoone{x}{i}) = 0$.)
Hence it is \emph{incorrect} to slice $P_3$ into $P_x$.
\end{example}

We conclude that the $\observevv$ statement cannot be removed;
this is because it is \emph{control dependent}
on the assignment to $x$, on which the returned $x$ also depends.

The discussion so far suggests the following 
tentative correctness condition
for the set $Q$ picked by slicing:
\begin{itemize}
\item
$Q$ is ``closed under dependence'', \ie, if a program point in $Q$
depends on another program point then that program point also belongs to $Q$;
\item
$Q$ is part of a ``slicing pair'':
any $\observevv$ statement that is sliced away belongs to a set $Q_0$ that
is also closed under dependence and is disjoint from $Q$.
\end{itemize}
The above condition will be made precise in
Definition~\ref{defn:slicing-pair} (Section~\ref{sec:conditions})
which contains a further requirement,
necessary since an
$\observevv$ statement may be encoded as a potentially non-terminating 
loop, as the next example illustrates.
%\begin{exmp}
\begin{example}
\label{ex4}
Consider the program $P_4$ given by 

%% \begin{tabbing}
%% kkk\=kkk\kill
%% \>1: \=$\rassignv{x}{\psi_4}$ \\
%% \>2: \>$\assignv{y}{0}$ \\
%% \>3: \>{\bf if}$\; x$\=$ \geq 2$ \\
%% \>4: \>        \>{\bf wh}\={\bf{ile}}$\; y$\= $< 3$ \\
%% \>5: \>        \>       \>$C$ \\
%% \>6: \>$\retv{x}$
%% \end{tabbing}
\[
\begin{array}{ll}
1: & \rassignv{x}{\psi_4} \\
2: & \assignv{y}{0} \\
3: & \mathbf{if}\; x \geq 2 \\
4: & \quad \mathbf{while}\; y < 3 \\
5: & \quad \qquad      C \\
6: & \retv{x}
\end{array}
\]
where $C$ is a (random) assignment. If $C$ is ``$\assignv{y}{y + 1}$''
then the loop terminates after 3 iterations,
and $y$'s final value is 3. In the resulting
distribution $D'$, for $i \in \{0,1,2,3\}$ we have
$D'(\stoone{x}{i}) = 
D'(\stotwo{x}{i}{y}{3}) = \frac{1}{4} = \dlt'(\stoone{x}{i})$. Thus
it is correct to slice $P_4$ into $P_x$.

But if $C$ is ``$\assignv{y}{1}$'' then the program will
not terminate when $x \geq 2$ (and hence the conditional encodes
$\observev{x < 2}$). Thus the resulting distribution 
$D_4$ is given by
$D_4(\stoone{x}{i}) = \frac{1}{4}$ when $i \in \{0,1\}$ and 
$D_4(\stoone{x}{i}) = 0$ when $i \notin \{0,1\}$.
Thus it is incorrect to slice $P_4$ into $P_x$.
Indeed, Definition~\ref{defn:slicing-pair} rules out such a slicing.

Now assume that $C$ is ``$\rassignv{y}{\psi_4}$''. Then the loop
may iterate arbitrarily many times, but will yet
terminate with probability 1, and $y$'s final value is 3.
Again, it is correct to slice $P_4$ into $P_x$.
%\end{exmp}
\end{example}

\section{Probabilistic Control-Flow Graphs}
\label{sec:CFG}
This section precisely defines the kind of probabilistic control-flow graphs (pCFGs) we consider,
as well as some key concepts that are mostly standard
(see, \eg, \cite{Pod+Cla:TSE-1990,Bal+Hor:AAD-1993}).
However, we also introduce a notion (Definition~\ref{def:outside})
specific to our approach. 

Figure~\ref{fig:ex34} depicts,
with the nodes numbered,
the pCFGs corresponding to the programs
$P_3$ and $P_4$ from Examples~\ref{ex3} and \ref{ex4}.
We see that a node $v \in \unodes$
can be labeled (the label is called $\labv{v}$) with an assignment
$\assignv{x}{E}$ ($x$ a program variable
and $E$ an arithmetic expression),
with a random assignment $\rassignv{x}{\psi}$
(we shall assume that the probability distribution $\psi$
contains no program variables though it would be straightforward
to allow it as in~\cite{Hur+etal:PLDI-2014}),
with $\observev{B}$ ($B$ is a boolean expression),
or (though not part of these examples) with $\skipv$;
a node of the abovementioned kinds has exactly one outgoing edge.
Also, there are branching nodes with \emph{two}
outgoing edges. (If $v$ has an outgoing edge to $v'$ we say that
$v'$ is a successor of $v$; a branching node has a $\true$-successor
and a $\false$-successor.) Finally, there is
a unique $\finalv$ node 
to which there must be a path from all other nodes
but which has no outgoing edges,
and a special node $\startv$ (which 
is numbered 1 in the examples)
from which there is a path to all other nodes.
It will often be the case that the $\finalv$ node has a label
of the form $\retv{x}$.
%% We shall say that a CFG is deterministic
%% if it has no $\observevv$ nodes or random assignments.

We let $\defv{v}$ be
the variable occurring on the left hand side if $v$ is a (random)
assignment, and let $\usev{v}$ be the variables used in $v$, that is:
if $v$ is an assignment then those occurring on the right hand side;
if $v$ is an $\observevv$ node or a branching node
then those occurring in the boolean expression; and
if $v$ is $\retv{x}$ then $\{x\}$.
We demand that all variables be defined before they are used:
for all nodes $v \in \unodes$, for all $x \in \usev{v}$,
and for all paths $\pi$ from $\startv$ to $v$,
there must exist $v_0 \in \pi$ with $v_0 \neq v$
such that $x \in \defv{v_0}$.

\begin{figure}
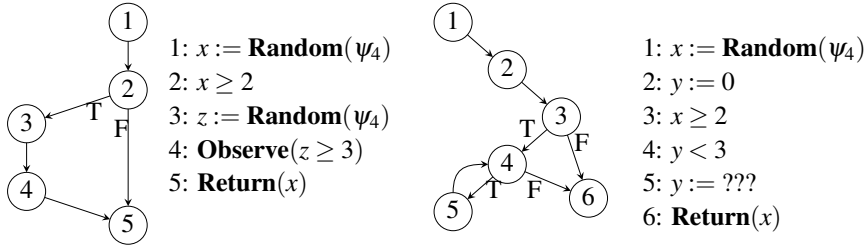

\begin{pgfpicture}{0mm}{5mm}{130mm}{35mm}
\begin{pgfmagnify}{0.9}{0.9}

\pgfnodecircle{v1}[stroke]{\pgfxy(3,3.5)}{8pt}

\pgfputat{\pgfxy(3,3.4)}{\pgfbox[center,base]{1}}

\pgfnodecircle{v2}[stroke]{\pgfxy(3,2.5)}{8pt}

\pgfputat{\pgfxy(3,2.4)}{\pgfbox[center,base]{2}}

\pgfnodecircle{v3}[stroke]{\pgfxy(1.5,2)}{8pt}

\pgfputat{\pgfxy(1.5,1.9)}{\pgfbox[center,base]{3}}

\pgfnodecircle{v4}[stroke]{\pgfxy(1.5,1)}{8pt}

\pgfputat{\pgfxy(1.5,0.9)}{\pgfbox[center,base]{4}}

\pgfnodecircle{v5}[stroke]{\pgfxy(3,0.5)}{8pt}

\pgfputat{\pgfxy(3,0.4)}{\pgfbox[center,base]{5}}

\pgfsetarrowsend{stealth}
\pgfsetendarrow{\pgfarrowtriangle{4pt}}

\pgfnodeconnline{v1}{v2}

\pgfnodeconnline{v2}{v3}

\pgfputat{\pgfxy(2.5,2.05)}{\pgfbox[center,base]{T}}

\pgfputat{\pgfxy(2.9,1.8)}{\pgfbox[center,base]{F}}

\pgfnodeconnline{v2}{v5}

\pgfnodeconnline{v3}{v4}

\pgfnodeconnline{v4}{v5}

\pgfnodecircle{w1}[stroke]{\pgfxy(7.8,3.5)}{8pt}

\pgfputat{\pgfxy(7.8,3.4)}{\pgfbox[center,base]{1}}

\pgfnodecircle{w2}[stroke]{\pgfxy(8.6,2.8)}{8pt}

\pgfputat{\pgfxy(8.6,2.7)}{\pgfbox[center,base]{2}}

\pgfnodecircle{w3}[stroke]{\pgfxy(9.4,2.1)}{8pt}

\pgfputat{\pgfxy(9.4,2.0)}{\pgfbox[center,base]{3}}

\pgfnodecircle{w4}[stroke]{\pgfxy(8.6,1.4)}{8pt}

\pgfputat{\pgfxy(8.6,1.3)}{\pgfbox[center,base]{4}}

\pgfnodecircle{w5}[stroke]{\pgfxy(7.8,0.7)}{8pt}

\pgfputat{\pgfxy(7.8,0.6)}{\pgfbox[center,base]{5}}

\pgfnodecircle{w6}[stroke]{\pgfxy(9.8,0.9)}{8pt}

\pgfputat{\pgfxy(9.8,0.8)}{\pgfbox[center,base]{6}}

\pgfsetarrowsend{stealth}
\pgfsetendarrow{\pgfarrowtriangle{4pt}}

\pgfnodeconnline{w1}{w2}

\pgfnodeconnline{w2}{w3}

\pgfnodeconnline{w3}{w4}

\pgfputat{\pgfxy(8.9,1.8)}{\pgfbox[center,base]{T}}

\pgfputat{\pgfxy(9.7,1.6)}{\pgfbox[center,base]{F}}

\pgfnodeconnline{w3}{w6}

\pgfnodeconnline{w4}{w5}

\pgfputat{\pgfxy(8.4,0.9)}{\pgfbox[center,base]{T}}

\pgfputat{\pgfxy(9.0,0.9)}{\pgfbox[center,base]{F}}

\pgfnodeconnline{w4}{w6}

\pgfnodeconncurve{w5}{w4}{90}{180}{10pt}{10pt}

\pgfputat{\pgfxy(3.6,3.0)}{\pgfbox[left,base]{1: $\rassignv{x}{\psi_4}$}}
\pgfputat{\pgfxy(3.6,2.5)}{\pgfbox[left,base]{2: $x \geq 2$}}
\pgfputat{\pgfxy(3.6,2.0)}{\pgfbox[left,base]{3: $\rassignv{z}{\psi_4}$}}
\pgfputat{\pgfxy(3.6,1.5)}{\pgfbox[left,base]{4: $\observev{z \geq 3}$}}
\pgfputat{\pgfxy(3.6,1.0)}{\pgfbox[left,base]{5: $\retv{x}$}}

\pgfputat{\pgfxy(10.6,3.0)}{\pgfbox[left,base]{1: $\rassignv{x}{\psi_4}$}}
\pgfputat{\pgfxy(10.6,2.5)}{\pgfbox[left,base]{2: $\assignv{y}{0}$}}
\pgfputat{\pgfxy(10.6,2.0)}{\pgfbox[left,base]{3: $x \geq 2$}}
\pgfputat{\pgfxy(10.6,1.5)}{\pgfbox[left,base]{4: $y < 3$}}
\pgfputat{\pgfxy(10.6,1.0)}{\pgfbox[left,base]{5: $\assignv{y}{???}$}}
\pgfputat{\pgfxy(10.6,0.5)}{\pgfbox[left,base]{6: $\retv{x}$}}

\end{pgfmagnify}
\end{pgfpicture}

\caption{\label{fig:ex34} The pCFGs for $P_3$ (left) and $P_4$ (right)
from Examples~\ref{ex3} and \ref{ex4}.}
\end{figure}

\begin{definition}[Postdomination]
We say that $v_1$ \dt{postdominates} $v$, also written
$(v,v_1) \in \PD$, if $v_1$ occurs
on all paths from $v$ to $\finalv$;
if also $v_1 \neq v$,
$v_1$ is a \dt{proper postdominator} of $v$.
\end{definition}
It is easy to see that
the relation ``postdominates'' is reflexive, transitive, and antisymmetric.
We say that $v_1$ is the \dt{first proper postdominator} of $v$
if whenever $v_2$ is another proper postdominator of $v$
then all paths from $v$ to $v_2$ contain $v_1$.
\begin{lemma}
For any $v$ with $v \neq \finalv$,
there is a unique first proper postdominator of $v$.
\end{lemma}
\begin{proof}
This follows from Lemma~\ref{lem:prec-ordering}
%% (proved in Appendix~\ref{app:proofs}) 
since the set of proper postdominators is
non-empty (will at least contain $\finalv$).
\end{proof}
We shall use the term $\fppd{v}$ for the unique first proper
postdominator of $v$.
In Figure~\ref{fig:ex34}(right), $\fppd{1} = 2$
(while also nodes 3 and 6 are proper postdominators of 1)
and $\fppd{3} = 6$.
\begin{lemma}
\label{lem:prec-ordering}
For given $v$, let $\prec$ be an ordering among
proper postdominators of $v$,
by stipulating that $v_1 \prec v_2$ iff in all acyclic paths from $v$
to $\finalv$, $v_1$ occurs strictly before $v_2$.
Then $\prec$ is transitive, antisymmetric, and total.
Also, if $v_1 \prec v_2$ then 
for \emph{all} paths from $v$ to $\finalv$ it is the case that
the first occurrence of $v_1$ is before the first occurrence of $v_2$.
\end{lemma}

\begin{definition}[\LAPf]
For $(v,v') \in \PD$, we define $\LAP{v}{v'}$ as the maximum length
of an acyclic path from $v$ to $v'$. 
(The length of a path is the number of edges.)
\end{definition}
Thus $\LAP{v}{v} = 0$ for all nodes $v$. 
As expected, we have:
%% (as proved in Appendix~\ref{app:proofs}):
\begin{lemma}
\label{lem:LAP-add}
If $(v,v_1) \in \PD$ and $(v_1,v_2) \in \PD$ (and thus
$(v,v_2) \in \PD$) then $\LAP{v}{v_2} = \LAP{v}{v_1} + \LAP{v_1}{v_2}$.
\end{lemma}

To reason about cycles, it is useful to pinpoint the kind of nodes that cause cycles:
\begin{definition}[Cycle-inducing]
\label{def:cycle-induce}
A node $v$ is \emph{cycle-inducing} if with $v' = \fppd{v}$ there exists a successor $v_i$
of $v$ such that $\LAP{v_i}{v'} \geq \LAP{v}{v'}$.
\end{definition}
Note that if $v$ is cycle-inducing then $v$ must be a branching node
(since if $v$ has only one successor then that successor is $v'$).
%\begin{exmp}
\begin{example}
\label{ex:cycle-induce}
In Figure~\ref{fig:ex34}(right), there are two branching nodes,
3 and 4, both having node 6 as their first proper postdominator.
Node 4 is cycle-inducing, since
5 is a successor of 4 with
$\LAP{5}{6} = 2 > 1 = \LAP{4}{6}$. 
On the other hand, node 3 is not cycle-inducing, since
$\LAP{3}{6} = 2$ which is strictly greater
than $\LAP{4}{6}$ ($=1$) and $\LAP{6}{6}$ ($=0$).
\end{example}
%\end{exmp}
\begin{lemma}
\label{lem:LAPcycle}
If $v$ is cycle-inducing then there exists a cycle that contains $v$
but not $\fppd{v}$.
\end{lemma}
\begin{proof}
With $v' = \fppd{v}$, by assumption there exists 
a successor $v_i$ of $v$ such that $\LAP{v_i}{v'} \geq \LAP{v}{v'}$;
observe that $v'$ is a postdominator of $v_i$.
Let $\pi$ be an acyclic path from $v_i$ to $v'$ with length
$\LAP{v_i}{v'}$;
then the path $v \pi$ is a path from $v$ to $v'$ that is longer
than $\LAP{v_i}{v'}$, and thus also longer
than $\LAP{v}{v'}$. This shows that $v \pi$ cannot be acyclic;
hence $v \in \pi$ and thus $v \pi$ contains a cycle involving $v$ but not $v'$.
\end{proof}
\begin{lemma}
\label{lem:cycle1inducing}
All cycles will contain at least one node which is cycle-inducing.
\end{lemma}
\begin{proof}
Let a cycle $\pi$ be given. For each $v \in \pi$, define 
$f(v)$ as $\LAP{v}{\finalv}$. For a node $v$ that 
has only one successor, $v_1$,
we have $f(v_1) < f(v)$ (since by Lemma~\ref{lem:LAP-add}, 
$f(v) = \LAP{v}{v_1} + f(v_1) = 1 + f(v_1)$).
Thus $\pi$ must contain a branching node $v_0$ with a successor $v_i$
such that $f(v_i) \geq f(v_0)$.
But with $v' = \fppd{v_0}$ we then have (by Lemma~\ref{lem:LAP-add})
\[
\LAP{v_i}{v'} = f(v_i) - f(v')
\geq f(v_0) - f(v')
= \LAP{v_0}{v'}
\]
which shows that $v_0$ is cycle inducing.
\end{proof}

\begin{definition}[Data dependence]
\label{def:datadep}
We say that 
$v_2$ is \dt{data dependent} on $v_1$, written $\dd{v_1}{v_2}$,
if there exists $x \in \usev{v_2} \cap \defv{v_1}$,
and there exists a path $\pi$ (with at least one edge) from $v_1$ to $v_2$ 
such that $x \notin \defv{v}$ for all nodes $v$ that are interior in $\pi$.
\end{definition}
In Figure~\ref{fig:ex34}(right), $\dd{2}{4}$ and $\dd{5}{4}$. 
A set of nodes $Q$ is
\dt{closed under data dependence} if whenever
$v_2 \in Q$ and $\dd{v_1}{v_2}$ then also $v_1 \in Q$.

\begin{definition}[Relevant variable]
We say that $x$ is ($Q$-)relevant for $v$, 
written $x \in \rv{Q}{v}$, if
there exists $v' \in Q$ such that
$x \in \usev{v'}$, and an acyclic path $\pi$ from $v$ to $v'$ 
such that $x \notin \defv{v_1}$ for all $v_1 \in \pi \setminus \mkset{v'}$. 
\end{definition}
For example, in Figure~\ref{fig:ex34}(left), 
$\rv{\{4,5\}}{4} = \{x,z\}$ but
$\rv{\{4,5\}}{3} = \{x\}$. 
The following two lemmas follow from the above definition.
\begin{lemma}
\label{lem:rv-assign}
Assume that $v$ is an assignment,
of the form $\assignv{x}{E}$, with successor $v'$.
If $v \in Q$ then (with $\fv{E}$ the free variables in $E$)
\[
\rv{Q}{v} = (\rv{Q}{v'} \setminus\mkset{x}) \cup \fv{E}.
\]
\end{lemma}
This follows since
the variables free in $E$ are relevant before $v$ (as $v \in Q$),
and all variables relevant after $v$ are also relevant before $v$ 
except for $x$ as it is being redefined.
\begin{lemma}
\label{lem:rv-branch}
Assume that $v$ is a branching node,
with condition $B$ and with successors $v_1$
and $v_2$. If $v \in Q$ then
\[
\rv{Q}{v} = \fv{B} \cup \rv{Q}{v_1} \cup \rv{Q}{v_2}.
\]
\end{lemma}
\begin{lemma}
\label{lem:rv-cup} % part 1
\label{lem:rv-cap} % part 2, not referred to
For all nodes $v$, and all node sets $Q_1$ and $Q_2$
\begin{itemize}
\item
$\rv{Q_1 \cup Q_2}{v} = \rv{Q_1}{v} \cup \rv{Q_2}{v}$
\item
$\rv{Q_1}{v} \cap \rv{Q_2}{v} = \emptyset$
if $Q_1$ and $Q_2$ are both closed under data dependence,
and $Q_1 \cap Q_2 = \emptyset$.
\end{itemize}
\end{lemma}
\begin{proof}
The first claim is obvious. For the second claim, assume
that $Q_1$ and $Q_2$ are closed under data dependence, that
$Q_1 \cap Q_2 = \emptyset$;
further assume, to get a contradiction, that $x \in \rv{Q_1}{v} \cap \rv{Q_2}{v}$.
That is, there exists $v_1 \in Q_1$ with $x \in \usev{v_1}$ and
a path from $v$ to $v_1$ that does not define $x$ until possibly $v_1$,
and there exists $v_2 \in Q_2$ with $x \in \usev{v_2}$ and
a path from $v$ to $v_2$ that does not define $x$ until possibly $v_2$.
As we have demanded that $x$ is defined before it is used,
we infer that with $\pi$ a path from $\startv$ to $v$, 
at least one of the nodes in $\pi$ defines $x$; let $v_x$ be the last
such node. As $Q_1$ and $Q_2$ are closed under data dependence,
we infer that $v_x \in Q_1$ and $v_x \in Q_2$, yielding the desired contradiction
since $Q_1 \cap Q_2 = \emptyset$.
\end{proof}
Next, a concept we have discovered useful
for the subsequent development:
\begin{definition}[Staying outside until]
\label{def:outside}
With $v'$ a postdominator of $v$, and $Q$ a set of nodes,
we say that \dt{$v$ stays outside $Q$ until $v'$} 
iff whenever $\pi$ is a path from $v$ to $v'$ 
where $v'$ occurs only at the end, 
$\pi$ will contain no node in $Q$ except possibly $v'$. 
\end{definition}
In Figure~\ref{fig:ex34}(right), 
node 4 stays outside $\mkset{1,6}$ until 6
but does not stay outside $\mkset{1,5,6}$ until 6.
Trivially, $v$ stays outside $Q$ until $v$ for
all $Q$ and $v$.

\begin{lemma}
\label{lem:outside-same-rv}
If $v$ stays outside $Q$ until $v'$
and $Q$ is closed under data dependence
then $\rv{Q}{v} = \rv{Q}{v'}$.
\end{lemma}
%% Moreover, if $Q$ satisfies certain additional properties,
%% the distribution at $v'$ (of the relevant variables)
%% will equal the distribution at $v$.
\begin{proof}
The claim is trivial if $v = v'$, so assume $v \neq v'$.
We shall show inclusions each way.

First assume that $x \in \rv{Q}{v}$. Thus
there exists $v_0 \in Q$ with $x \in \usev{v_0}$
and a path $\pi$ from $v$ to $v_0$ such that 
$x \notin \defv{v_1}$ for all $v_1 \in \pi \setminus \mkset{v_0}$.
Since $v'$ postdominates $v$, and $v$ stays outside $Q$ until $v'$,
we infer that $v'$ belongs to $\pi$
and thus a suffix of $\pi$ is a path from $v'$ to $v_0$ which
shows that $x \in \rv{Q}{v'}$.

Conversely, assume that $x \in \rv{Q}{v'}$. Thus
there exists $v_0 \in Q$ with $x \in \usev{v_0}$
and a path $\pi'$ from $v'$ to $v_0$ such that
$x \notin \defv{v_1}$ for all $v_1 \in \pi' \setminus \mkset{v_0}$.
With $\pi$ an acyclic path from $v$ to $v'$,
the concatenation of $\pi$ and $\pi'$ is a path
from $v$ to $v_0$ which will show the desired
$x \in \rv{Q}{v}$, 
provided that $\pi$ does not contain a node $v_1 \neq v'$
with $x \in \defv{v_1}$.
Towards a contradiction, assume that such a node does exist; 
with $v_1$ the last such node we would have
$\dd{v_1}{v_0}$ so from $v_0 \in Q$ and $Q$ closed under data dependence
we could infer $v_1 \in Q$ which contradicts the assumption
that $v$ stays outside $Q$ until $v'$.
\end{proof}

\section{Semantics}
\label{sec:sem}
In this section we shall define the meaning of the pCFGs introduced
in the previous section, in terms of an operational semantics 
that manipulates distributions which
assign probabilities to stores (Section~\ref{subsec:sem-dist}).
Section~\ref{subsec:sem-indpd} defines what it means for 
sets of variables to be independent wrt.~a given distribution.

The semantics of pCFGs is defined in a number of steps:
first (Section~\ref{subsec:sem-one}) we define
transfer functions for traversing one edge of the pCFG,
and next (Section~\ref{subsec:toplevel}) we present a functional, 
the fixed point of which provides the meaning of a pCFG. The 
semantics also applies to sliced programs and hence
provides the meaning of slicing.

In the companion article
\cite{Amt+Ban:ProbSemantics-2017}, we show that for a pCFG
that is the translation of a ``structured'' program, the semantics
given in this section 
is adequately related to the semantics given in
\cite{Gor+etal:ICSE-2014,Hur+etal:PLDI-2014}
for structured probabilistic programs.

\subsection{Stores and Distributions}
\label{subsec:sem-dist}
Let $\uvar$ be the universe of variables.
A store $s$ is a partial mapping from $\uvar$ to $\Zz$.
We write $\upd{s}{x}{z}$ for the store $s'$ that is like $s$
except $s'(x) = z$, and
write $\dom{s}$ for the domain of $s$.
We write $\sto{R}$ for the set of stores with domain $R$,
and also write $\fulls$ for $\sto{\uvar}$.
If $s_1 \in \sto{R_1}$ and $s_2 \in \sto{R_2}$
with $R_1 \cap R_2 = \emptyset$,
we may define $\adds{s_1}{s_2}$ with domain $R_1 \cup R_2$ the natural way.
If $s \in \sto{R'}$ and $R \subseteq R'$ we define
$\restrs{s}{R}$ as the restriction of $s$ to $R$.
With $R$ a subset of $\uvar$, we say that
$s_1$ agrees with $s_2$ on $R$, written $\sagree{s_1}{s_2}{R}$,
iff $R \subseteq \dom{s_1} \cap \dom{s_2}$ and for all $x \in R$, $s_1(x) = s_2(x)$. We assume that there is a function $\semee$
such that $\seme{E}s$ is the integer result of evaluating $E$ in store $s$
and $\seme{B}s$ is the boolean result of evaluating $B$ in store $s$
(the free variables of $E$ and $B$ must be in $\dom{s}$).

A distribution $D \in \Dist$ 
(we shall later also use the letter $\dlt$)
is a mapping from $\fulls$ to non-negative reals.
We shall often expect that $D$ is \emph{bounded}, that is
$\sumd{D} \leq 1$
where $\sumd{D}$ is a shorthand for $\sum_{s \in \fulls}D(s)$.
Thanks to our assumption that values are integers, and since $\uvar$
can be assumed finite,
$\fulls$ is a countable set and thus $\sumd{D}$ is well-defined
even without measure theory.
We define $D_1 + D_2$ by stipulating
$(D_1 + D_2)(s) = D_1(s) + D_2(s)$
(if $\sumd{D_1} + \sumd{D_2} \leq 1$ then $D_1 + D_2$ is bounded),
and for $c \geq 0$ we define $cD$ by stipulating 
$(cD)(s) = cD(s)$ (if $D$ is bounded and $c \leq 1$ then $cD$ is bounded);
we write $D = 0$ when $D(s) = 0$ for all $s$.
%% We write $D_1 \leq D_2$ iff
%% $D_1(s) \leq D_2(s)$ for all $s$,
%% We assume there is a designated initial distribution, $\initd$,
%% such that $\sumd{\initd} = 1$
%% ($\initd$ may be arbitrary as all variables must be defined
%% before they are used).
We say that $D$ is \emph{concentrated} if there exists $s_0 \in \fulls$
such that $D(s) = 0$ for all $s \in \fulls$ with $s \neq s_0$;
for that $s_0$, we say that $D$ is concentrated on $s_0$.
(Thus the distribution $0$ is concentrated on everything.)

Note that the set of real numbers in $[0..1]$ form a {\em pointed cpo}
with the usual ordering, as 0 is the bottom element and
the supremum operator yields the least upper bound of a chain.
(We refer to Appendix~\ref{app:domain} for the basics of domain theory.)
Hence also the set $\Dist$ of distributions form a pointed cpo,
with ordering defined pointwise ($D_1 \leq D_2$ iff
$D_1(s) \leq D_2(s)$ for all stores $s$), with 0 the bottom element,
and the least upper bound defined pointwise. 

The following result %% (proved in Appendix~\ref{app:proofs})
is often convenient;
in particular, it shows that
if each distribution in a chain $\chain{D_k}{k}$ is bounded
then also the least upper bound is a bounded distribution.
\begin{lemma}
\label{lem:sumlim-limsum}
Assume that $\chain{D_k}{k}$ 
is a chain of distributions (not necessarily bounded).
With $S$ a (countable) set of stores, we have
\[
\sum_{s \in S}{(\limit{k}{D_k})(s)} = \limit{k}{\sum_{s \in S}{D_k(s)}}.
\]
\end{lemma}
As suggested by the calculation in Example~\ref{ex1}, we have 
\begin{definition}
\label{def:Dpartial}
\rm
For $s \in \sto{R}$, let
$D(s) = \sum_{s_0 \in \fulls\ \mid\ \sagree{s_0}{s}{R}}D(s_0)$.
\end{definition}
Observe that $D(\emptyset) = \sumd{D}$. 
\begin{lemma}
\label{lem:sto-partial}
If $R \subseteq R'$ then
for $s \in \sto{R}$ we have
\begin{displaymath}
D(s) = \sum_{s' \in \sto{R'} \ \mid\ \sagree{s'}{s}{R}} D(s'). 
\end{displaymath}
\end{lemma}
Lemma~\ref{lem:sto-partial}
amounts to Definition~\ref{def:Dpartial}
if $R' = \uvar$,
and is trivial if $R' = R$ (as the right hand side is then
the sum of the singleton set $\mkset{D(s)}$).

By letting $R = \emptyset$ in Lemma~\ref{lem:sto-partial} we get
$D(\emptyset) = \sum_{s' \in \sto{R'}}D(s')$ and thus (by renaming)
\begin{lemma}
\label{lem:sumd}
For all distributions $D$, and all $R$,
$\sumd{D} = \sum_{s \in \sto{R}}D(s)$.
\end{lemma}
\begin{definition}[agrees with]
We say that $D_1$ agrees with $D_2$ on $R$, written $\dagree{D_1}{D_2}{R}$, 
if $D_1(s) = D_2(s)$ for all $s \in \sto{R}$.
\end{definition}
Agreement on a set implies agreement on a subset:
\begin{lemma}
\label{lem:dist-partial}
If $\dagree{D_1}{D_2}{R'}$ and $R \subseteq R'$ then $\dagree{D_1}{D_2}{R}$.
\end{lemma}
\begin{prooff}
For $s \in \sto{R}$ we infer from
Lemma~\ref{lem:sto-partial} that
\begin{displaymath}
D_1(s) 
= \sum_{s' \in \sto{R'}\ \mid\ \sagree{s'}{s}{R}}D_1(s')
= \sum_{s' \in \sto{R'}\ \mid\ \sagree{s'}{s}{R}}D_2(s')
= D_2(s). 
\end{displaymath}
\end{prooff}

\subsection{Probabilistic Independence}
\label{subsec:sem-indpd} 
Some variables of a distribution $D$ may be \emph{independent}
of other variables. That is, knowing the values of the former gives no
extra information about the values of the latter, or vice versa.
Formally:
\begin{definition}[independence]
\label{def:indpd}
\rm
Let $R_1$ and $R_2$ be disjoint sets of variables.
We say that $R_1$ and $R_2$ are \dt{independent} in $D$ iff
for all $s_1 \in \sto{R_1}$ and $s_2 \in \sto{R_2}$, we have 
$D(\adds{s_1}{s_2})\sumd{D} = D(s_1)D(s_2)$.
\end{definition}
To motivate the definition, first observe that
if $\sumd{D} = 1$ it amounts to the well-known
definition of probabilistic independence; 
next observe that if $0 < \sumd{D}$,
it is equivalent to
the well-known definition of
``normalized'' probabilities:
\begin{displaymath}
   \frac{D(\adds{s_1}{s_2})}{\sumd{D}} = 
\frac{D(s_1)}{\sumd{D}} \cdot \frac{D(s_2)}{\sumd{D}}
\end{displaymath}
Trivially,
$R_1$ and $R_2$ are independent in $D$ if
$D = 0$ or $R_1 = \emptyset$ or $R_2 = \emptyset$.
%\begin{exmp}
\begin{example}
In Example~\ref{ex1}, $\mkset{x}$ and $\mkset{y}$ are independent in $D_1$.
To see this, first note that $D_1$ produces a non-zero value for only 8 stores:
for $i=0\dots 3$, these are the 
stores $\{x \mapsto i, y \mapsto 2\}$, $\{x \mapsto i, y \mapsto 3\}$.

For $i \in \{0,1,2,3\}$ and $j \in \{2,3\}$ we have
$D_1(\stotwo{x}{i}{y}{j}) = \displaystyle \frac{1}{16}$ and thus
$D_1(\stoone{x}{i})  = \displaystyle \frac{1}{8}$ and
$D_1(\stoone{y}{j})  =  \displaystyle \frac{1}{4}$.
As $\displaystyle \sumd{D_1}  =  \sum_{s\in \fulls}D_1(s) = 8 \cdot 
\frac{1}{16} = \frac{1}{2}$,
we have the desired equality
$\displaystyle
D_1(\stotwo{x}{i}{y}{j})\sumd{D_1} = \frac{1}{32} =
D_1(\stoone{x}{i}) \cdot D_1(\stoone{y}{j})$. 

The equality holds trivially 
if $i \notin \{0,1,2,3\}$ or $j \notin \{2,3\}$
since then $D_1(\stotwo{x}{i}{y}{j}) = 0$
and either $D_1(\stoone{x}{i}) = 0$ or $D_1(\stoone{y}{j}) = 0$.
\end{example}
%\end{exmp}
\begin{example}
%\begin{exmp}
In Example~\ref{ex2}, $\mkset{x}$ and $\mkset{y}$ are 
\emph{not} independent in $D_2$:
$D_2(\stotwo{x}{3}{y}{3}) \sumd{D_2}  = \frac{3}{256}$, while
$D_2(\stoone{x}{3}) D_2(\stoone{y}{3})  = \frac{4}{256}$.
%\end{exmp}
\end{example}

\subsection{Transfer Functions}
\label{subsec:sem-one}
To deal with traversing a single edge in the pCFG, 
we shall define a number of functions with functionality
$\Dist \rightarrow \Dist$. Each such \dt{transfer function} $f$ will be 
\begin{description}
\item[additive] if $D_1 + D_2$ is a distribution then
$f(D_1 + D_2) = f(D_1) + f(D_2)$ 
(this reflects that a distribution is not more than the sum of its components);
\item[multiplicative] $f(cD) = cf(D)$ for all distributions
$D$ and all real $c \geq 0$;
\item [continuous] $f(\limit{k}{D_k}) = \limit{k}{f(D_k)}$
when $\chain{D_k}{k}$ is a chain of distributions
(this is is a key property for functions
on cpos, cf.~Appendix~\ref{app:domain});
\item [non-increasing] $\sumd{f(D)} \leq \sumd{D}$ for all
distributions $D$ (this reflects that distribution may disappear,
as we have seen in our examples, but
cannot be created ex nihilo).
\end{description}
Some functions will even be
\begin{description}
\item[sum-preserving] $\sumd{f(D)} = \sumd{D}$ for all
distributions $D$ 
(if $D$ is such that this equation holds
we say that $f$ is sum-preserving for $D$).
%% \item[deterministic] $f(D)$ is concentrated if $D$ is concentrated.
\end{description}
%% Transfer functions for non-trivial random assignments are not deterministic.
Transfer functions for $\observevv$ nodes are 
not sum-preserving (unless the condition is always true), and neither is the
semantic function for a loop that has a non-zero probability of
non-termination;
a primary contribution of this article is to show that
if a loop is sum-preserving it may be safe to slice it away,
even if it occurs in a branch (cf.~Example~\ref{ex4}).

To show that a function is sum-preserving, it suffices
to consider concentrated distributions:
%% (as proved in Appendix~\ref{app:proofs}):
\begin{lemma}
\label{lem:sum-pres-conc}
Let $f \in \Dist \rightarrow \Dist$
be continuous and additive.
Assume that for all $D$ that are concentrated, $f$ is sum-preserving for $D$.
Then $f$ is sum-preserving.
\end{lemma}
For a boolean expression $B$, we define $\selectf{B}$ by letting
$\selectf{B}(D) = D'$ where
\[
\begin{array}{rcll}
D'(s) & = & D(s) & \mbox{if } \seme{B}s \\[1mm]
D'(s) & = & 0 & \mbox{otherwise}
\end{array}
\]
\begin{lemma}
\label{lem:selectB}
For all $B$,
$\selectf{B}$ is continuous, additive, multiplicative, and
non-increasing; also, for all $D$ we have
\[
\selectf{B}(D) + \selectf{\neg B}(D) = D.
\]
\end{lemma}

\paragraph{Assignments}
For a variable $x$ and an expression $E$,
we define $\assignf{x}{E}$ by letting
$\assignf{x}{E}(D)$ 
be a distribution $D'$ such that
for each $s' \in \fulls$, 
\[
D'(s') =
\sum_{s \in \fulls\ \mid\ s' = \upd{s}{x}{\seme{E}s}} D(s)
\]
That is, the ``new'' probability of a store $s'$ is the sum of the 
``old'' probabilities
of the stores that become like $s'$ after the assignment
(this will happen for a store $s$ if $s' = \upd{s}{x}{\seme{E}s}$).
\begin{lemma}
\label{lem:assignRirrel}
Assume that $\assignf{x}{E}(D) = D'$
and $x \notin R$. Then
$\dagree{D}{D'}{R}$.
\end{lemma}
%% The proof is in Appendix~\ref{app:proofs}.
\begin{lemma}
\label{lem:assign-misc}
Each $\assignf{x}{E}$ is additive, multiplicative, non-increasing, and sum-preserving.
\end{lemma}
\begin{proof}
That $\assignf{x}{E}$ is sum-preserving, 
and hence non-increasing, follows from
Lemma~\ref{lem:assignRirrel}, with $R = \emptyset$. 
Additivity and multiplicativity are trivial.
\end{proof}
%% In Appendix~\ref{app:proofs}, we prove:
\begin{lemma}
\label{lem:assign-cont}
$\assignf{x}{E}$ is continuous.
\end{lemma}

\paragraph{Random Assignments}
For a variable $x$ and a random distribution $\psi$,
we define $\rassignf{x}{\psi}$ by letting
$\rassignf{x}{\psi}(D)$
be a distribution $D'$ such that
for each $s' \in \fulls$,
\[
D'(s') =
\sum_{s \in \fulls\ \mid\ \sagree{s'}{s}{\uvar \setminus \mkset{x}}} \psi(s'(x))D(s)
\]
\begin{lemma}
\label{lem:rassignRirrel}
Assume that $\rassignf{x}{E}(D) = D'$
and $x \notin R$. Then
$\dagree{D}{D'}{R}$.
\end{lemma}
%% The proof is in Appendix~\ref{app:proofs}.
\begin{lemma}
\label{lem:rassign-misc}
Each $\rassignf{x}{E}$ is additive, multiplicative, non-increasing, and sum-preserving.
\end{lemma}
\begin{proof}
That $\rassignf{x}{E}$ is sum-preserving, 
and hence non-increasing, follows from
Lemma~\ref{lem:rassignRirrel}, with $R = \emptyset$. 
Additivity and multiplicativity is trivial.
\end{proof}
%% In Appendix~\ref{app:proofs}, we prove:
\begin{lemma}
\label{lem:rassign-cont}
$\rassignf{x}{E}$ is continuous.
\end{lemma}

\subsection{Fixed-point Semantics}
\label{subsec:toplevel}
\label{subsec:sem-slice}
Having expressed the semantics of a single edge,
we shall now express the semantics of a full pCFG.
Our goal is to compute ``modification functions''
to express how a distribution is modified
as ``control'' moves from $\startv$ to $\finalv$.
To accomplish this, we shall solve a more general problem:
for each $(v,v') \in \PD$,
state how a given distribution
is modified as ``control'' moves from $v$ to $v'$
along paths that may contain
multiple branches and even loops
but which do \emph{not} contain $v'$ until the end.

We would have liked to have a definition of the modification
function that is inductive in $\LAP{v}{v'}$,
but this is not possible due to cycle-inducing nodes
(cf.~Definition~\ref{def:cycle-induce}).
For such nodes, the semantics cannot be expressed by recursive calls
on the successors, but the semantics of (at least) one of
the successors will have to be provided as an {\em argument}.
This motivates that our main semantic function be a \emph{functional}
that transforms a modification function into another modification function,
with the desired meaning being the \emph{fixed point}
(cf.~Lemma~\ref{lem:cont-cpo-fix} in Appendix~\ref{app:domain}) of this functional.

We shall now specify a functional $\HH{X}$
which is parametrized on a set $X$ of nodes;
the idea is that only the nodes in $X$ are taken into account.
To get a semantics for the original program, we must let $X$ be the set
$\unodes$ of all nodes; to get a semantics for a sliced program,
we must let $X$ be the set $Q$ of nodes included in the slice.

$\HH{X}$ operates on $\PD \rightarrow \Dist \rightarrow \Dist$
and we shall show (Lemma~\ref{lem:hcontH}) that it even operates
on $\PD \rightarrow \Dist \contarrow \Dist$ 
(we let $\contarrow$ denote the set of continuous functions)
which, as stated in
Lemma~\ref{lem:cont-cpo}
in Appendix~\ref{app:domain},
is a pointed cpo:
\begin{lemma}
\label{lem:is-cpo}
$\PD \rightarrow (\Dist \contarrow \Dist)$ is a pointed cpo,
with the ordering given pointwise, and
with least element $0$ given as $\lambda (v_1,v_2).\lambda D.0$.
\end{lemma}
\begin{definition}[$\HH{X}$]
\label{def:HHX}
The functionality of $\HH{X}$ is given by 
\[
\HH{X}: (\PD \rightarrow \Dist \rightarrow \Dist)
 \rightarrow
   (\PD \rightarrow \Dist \rightarrow \Dist)
\]
where, given 
\[
h_0: \PD \rightarrow \Dist \rightarrow \Dist
\]
we define 
\[
h = \HH{X}(h_0): \PD \rightarrow \Dist \rightarrow \Dist
\]
by letting $h (v,v')$, written $\hh{v}{v'}$, be stipulated by the
following rules that are inductive in $\LAP{v}{v'}$:
\begin{enumerate}
\item
if $v' = v$ then $\hh{v}{v'}(D) = D$;
\item
\label{evalk-transitive}
otherwise, if $v' \neq v''$ with $v'' = \fppd{v}$
then 
\[
\hh{v}{v'}(D) = \hh{v''}{v'}(\hh{v}{v''}(D))
\]
(this is well-defined by Lemma~\ref{lem:LAP-add});
\item
otherwise, that is if $v' = \fppd{v}$:
\begin{enumerate}
\item
\label{defn:evalk:notinX}
if $v \notin X$ 
or $\labv{v} = \skipv$
then $\hh{v}{v'}(D) = D$;
\item
if $v \in X$ with $\labv{v}$ of the form $\assignv{x}{E}$
then $\hh{v}{v'}(D) = \assignf{x}{E}(D)$;
\item
if $v \in X$ with $\labv{v}$ of the form $\rassignv{x}{\psi}$
then $\hh{v}{v'}(D) = \rassignf{x}{\psi}(D)$;
\item
if $v \in X$ with $\labv{v}$ of the form $\observev{B}$
then $\hh{v}{v'}(D) = \selectf{B}(D)$;
\item
\label{def:evalk:cond}
otherwise, that is if $v$ is a branching node with condition $B$,
we compute $\hh{v}{v'}$ as follows:
with $v_1$ the $\true$-successor of $v$ and $v_2$ the $\false$-successor
of $v$, let $D_1 = \selectf{B}(D)$ and $D_2 = \selectf{\neg B}(D)$;
we then let $\hh{v}{v'}(D)$ be $D'_1 + D'_2$ where for each
$i \in \{1,2\}$, $D'_i$ is computed as
\begin{itemize}
\item
if $\LAP{v_i}{v'} < \LAP{v}{v'}$ then
$D'_i =\hh{v_i}{v'}(D_i)$;
\item
if $\LAP{v_i}{v'} \geq \LAP{v}{v'}$ (and thus $v$ is cycle-inducing)
then 
$D'_i = \hpd{h_0}{v_i}{v'}(D_i)$.
\end{itemize}
\end{enumerate}
\end{enumerate}
\end{definition}
\begin{lemma}
\label{lem:hcontH}
Assume that $\hpd{h_0}{v}{v'}$ is continuous for all $(v,v') \in \PD$ and
let $h = \HH{X}(h_0)$. Then $\hh{v}{v'}$ is continuous for all $(v,v') \in \PD$.
\end{lemma}
\begin{proof}
This follows by an easy induction in $\LAP{v}{v'}$,
using Lemmas~\ref{lem:selectB}, \ref{lem:assign-cont} and
\ref{lem:rassign-cont}, and the fact that the composition of two continuous
functions is continuous.
\end{proof}
Thus $\HH{X}$ is a mapping from
$\PD \rightarrow (\Dist \contarrow \Dist)$ to itself.
\begin{lemma}
\label{lem:HHcont}
The functional $\HH{X}$ is continuous
on $\PD \rightarrow (\Dist \contarrow \Dist)$.
\end{lemma}

Lemmas~\ref{lem:is-cpo} and
\ref{lem:HHcont}, together with
Lemma~\ref{lem:cont-cpo-fix} in Appendix~\ref{app:domain}, 
% (where also the notation $\limit{k}{}$ is defined)
give
\begin{proposition}
The functional $\HH{X}$ has a least fixed point
(belonging to $\PD \rightarrow (\Dist \contarrow \Dist)$),
called $\fixed{\HH{X}}$, and given as
$\limit{k}{\HH{X}^k(0)}$, that is
the limit of the chain $\chain{\HH{X}^k(0)}{k}$
(where $\HH{X}^k(0)$ denotes $k$ applications of
$\HH{X}$ to the modification function that maps all distributions to 0). 
\end{proposition}
We can now define the meaning of the original program:
\begin{definition}[Meaning of Original Program]
Given a pCFG, with $\unodes$ the set of its nodes, 
we define its meaning $\omega$ as $\omega = \fixed{\HH{\unodes}}$.
Thus $\omega = \limit{k}{\omega_k}$ 
where $\omega_k = \HH{\unodes}^k(0)$ (thus $\omega_0 = 0$).
\end{definition}
Thus for all $k > 0$ we have $\omega_k = \HH{\unodes}(\omega_{k-1})$.
Intuitively speaking, $\omega_k$ is the meaning of the program assuming
that control is allowed to loop, that is move ``backwards'',
at most $k-1$ times.

A \emph{slice set} is a set $Q$ of nodes 
 (which must satisfy certain
conditions, \cf~Definition~\ref{defn:slicing-pair})
to be included in the slice:
\begin{definition}[Meaning of Sliced Program]
\label{def:slice-meaning}
Given a pCFG, and given a slice set $Q$,
we define the meaning of the sliced program as
$\phi = \fixed{\HH{Q}}$.
Thus $\phi = \limit{k}{\phi_k}$ 
where $\phi_k = \HH{Q}^k(0)$.
\end{definition}
\begin{example}
%\begin{exmp}
Consider Example~\ref{ex1}, with $Q$ containing nodes 1 and 4.
Thus nodes 2 and 3 are treated like $\skipv$ nodes, and
for $D$ we thus have
$\hpd{\phi}{1}{4}(D) = \rassignf{x}{\psi_4}(D)$.
%\end{exmp}
\end{example}
The following result is applicable to the original program
as well as to the sliced program:
%% are proved in Appendix~\ref{app:proofs}:
\begin{lemma}
\label{lem:fixed-mult-nonincr-determ}
Let $h = \fixed{\HH{X}}$. Then for each $(v,v') \in \PD$,
$\hpd{h}{v}{v'}$ is additive, multiplicative and non-increasing
(as is also $\hpd{\omega_k}{v}{v'}$ for each $k \geq 0$).
%% Also, for a deterministic CFG, $\hpd{h}{v}{v'}$ is deterministic,
%% that is: if $D$ is concentrated then $\hpd{h}{v}{v'}(D)$ is concentrated.
\end{lemma}
On the other hand, $\hpd{h}{v}{v'}$ may fail to be sum-preserving, due to
$\observevv$ nodes (\cf~Examples~\ref{ex1}--\ref{ex3}),
or due to infinite loops (\cf~Example~\ref{ex4}).

\begin{figure}
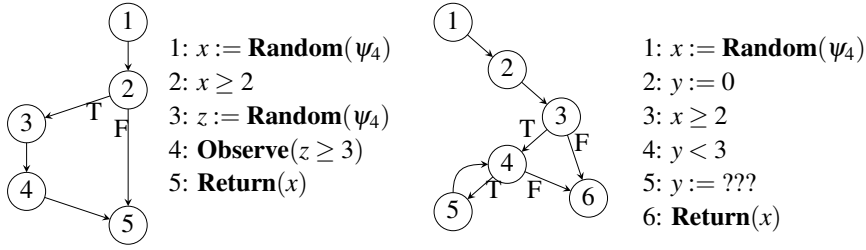

\begin{pgfpicture}{0mm}{5mm}{130mm}{35mm}
\begin{pgfmagnify}{0.9}{0.9}

\pgfnodecircle{v1}[stroke]{\pgfxy(3,3.5)}{8pt}

\pgfputat{\pgfxy(3,3.4)}{\pgfbox[center,base]{1}}

\pgfnodecircle{v2}[stroke]{\pgfxy(3,2.5)}{8pt}

\pgfputat{\pgfxy(3,2.4)}{\pgfbox[center,base]{2}}

\pgfnodecircle{v3}[stroke]{\pgfxy(1.5,2)}{8pt}

\pgfputat{\pgfxy(1.5,1.9)}{\pgfbox[center,base]{3}}

\pgfnodecircle{v4}[stroke]{\pgfxy(1.5,1)}{8pt}

\pgfputat{\pgfxy(1.5,0.9)}{\pgfbox[center,base]{4}}

\pgfnodecircle{v5}[stroke]{\pgfxy(3,0.5)}{8pt}

\pgfputat{\pgfxy(3,0.4)}{\pgfbox[center,base]{5}}

\pgfsetarrowsend{stealth}
\pgfsetendarrow{\pgfarrowtriangle{4pt}}

\pgfnodeconnline{v1}{v2}

\pgfnodeconnline{v2}{v3}

\pgfputat{\pgfxy(2.5,2.05)}{\pgfbox[center,base]{T}}

\pgfputat{\pgfxy(2.9,1.8)}{\pgfbox[center,base]{F}}

\pgfnodeconnline{v2}{v5}

\pgfnodeconnline{v3}{v4}

\pgfnodeconnline{v4}{v5}

\pgfnodecircle{w1}[stroke]{\pgfxy(7.8,3.5)}{8pt}

\pgfputat{\pgfxy(7.8,3.4)}{\pgfbox[center,base]{1}}

\pgfnodecircle{w2}[stroke]{\pgfxy(8.6,2.8)}{8pt}

\pgfputat{\pgfxy(8.6,2.7)}{\pgfbox[center,base]{2}}

\pgfnodecircle{w3}[stroke]{\pgfxy(9.4,2.1)}{8pt}

\pgfputat{\pgfxy(9.4,2.0)}{\pgfbox[center,base]{3}}

\pgfnodecircle{w4}[stroke]{\pgfxy(8.6,1.4)}{8pt}

\pgfputat{\pgfxy(8.6,1.3)}{\pgfbox[center,base]{4}}

\pgfnodecircle{w5}[stroke]{\pgfxy(7.8,0.7)}{8pt}

\pgfputat{\pgfxy(7.8,0.6)}{\pgfbox[center,base]{5}}

\pgfnodecircle{w6}[stroke]{\pgfxy(9.8,0.9)}{8pt}

\pgfputat{\pgfxy(9.8,0.8)}{\pgfbox[center,base]{6}}

\pgfsetarrowsend{stealth}
\pgfsetendarrow{\pgfarrowtriangle{4pt}}

\pgfnodeconnline{w1}{w2}

\pgfnodeconnline{w2}{w3}

\pgfnodeconnline{w3}{w4}

\pgfputat{\pgfxy(8.9,1.8)}{\pgfbox[center,base]{T}}

\pgfputat{\pgfxy(9.7,1.6)}{\pgfbox[center,base]{F}}

\pgfnodeconnline{w3}{w6}

\pgfnodeconnline{w4}{w5}

\pgfputat{\pgfxy(8.4,0.9)}{\pgfbox[center,base]{T}}

\pgfputat{\pgfxy(9.0,0.9)}{\pgfbox[center,base]{F}}

\pgfnodeconnline{w4}{w6}

\pgfnodeconncurve{w5}{w4}{90}{180}{10pt}{10pt}

\pgfputat{\pgfxy(3.6,3.0)}{\pgfbox[left,base]{1: $\rassignv{x}{\psi_4}$}}
\pgfputat{\pgfxy(3.6,2.5)}{\pgfbox[left,base]{2: $x \geq 2$}}
\pgfputat{\pgfxy(3.6,2.0)}{\pgfbox[left,base]{3: $\rassignv{z}{\psi_4}$}}
\pgfputat{\pgfxy(3.6,1.5)}{\pgfbox[left,base]{4: $\observev{z \geq 3}$}}
\pgfputat{\pgfxy(3.6,1.0)}{\pgfbox[left,base]{5: $\retv{x}$}}

\pgfputat{\pgfxy(10.6,3.0)}{\pgfbox[left,base]{1: $\rassignv{x}{\psi_4}$}}
\pgfputat{\pgfxy(10.6,2.5)}{\pgfbox[left,base]{2: $\assignv{y}{0}$}}
\pgfputat{\pgfxy(10.6,2.0)}{\pgfbox[left,base]{3: $x \geq 2$}}
\pgfputat{\pgfxy(10.6,1.5)}{\pgfbox[left,base]{4: $y < 3$}}
\pgfputat{\pgfxy(10.6,1.0)}{\pgfbox[left,base]{5: $\assignv{y}{???}$}}
\pgfputat{\pgfxy(10.6,0.5)}{\pgfbox[left,base]{6: $\retv{x}$}}

\end{pgfmagnify}
\end{pgfpicture}

\caption{\label{fig:ex34r} The pCFGs for $P_3$ (left) and $P_4$ (right)
from Examples~\ref{ex3} and \ref{ex4} (copied from Figure~\ref{fig:ex34}).}
\end{figure}

\begin{example}
%\begin{exmp}
\label{ex:sem-ex}
Consider Example~\ref{ex4}, with pCFG depicted in the right of 
Figure~\ref{fig:ex34r} (which for the reader's convenience we
copy from Figure~\ref{fig:ex34}).
We shall consider various possibilities
for the assignment at node 5;
we shall prove that $\hpd{\omega}{4}{6}$ is sum-preserving
when $\labv{5}$ is an assignment $\assignv{y}{y+1}$
or a random assignment $\rassignv{y}{\psi_4}$
(but not when it is $\assignv{y}{1}$).

In all cases, it is convenient to define,
for integers $i,j$ and for real $r \geq 0$,
the concentrated distribution $\Dijr{i}{j}{r}$ by stipulating that
\begin{eqnarray*}
\Dijr{i}{j}{r}(s) & = & r \mbox{ if } s = \stotwo{x}{i}{y}{j} \\
\Dijr{i}{j}{r}(s) & = & 0 \mbox{ otherwise}
\end{eqnarray*}
When control first reaches node 4, the distribution is
$\displaystyle \sum_{i = 2,3} \Dijr{i}{0}{0.25}$
since $y$ is zero, and for $x$, only the values 2 and 3 lead to
node $4$ while the values 0 and 1 do not.

To show that $\hpd{\omega}{4}{6}$ is sum-preserving,
by Lemma~\ref{lem:sum-pres-conc}
it is enough to show that 
$\hpd{\omega}{4}{6}$ is sum-preserving for each $\Dijr{i}{j}{r}$.
For that purpose, we shall consider a given $i$, and 
a given $r \geq 0$.
Recall (Example~\ref{ex:cycle-induce})
that $\LAP{5}{6} > \LAP{4}{6}$, 

First consider $j \geq 3$.
Then $\selectf{y < 3}(\Dijr{i}{j}{r}) = 0$ and 
$\selectf{\neg(y < 3)}(\Dijr{i}{j}{r}) = \Dijr{i}{j}{r}$;
thus we see from clause~\ref{def:evalk:cond}
in Definition~\ref{def:HHX} 
(substituting $\omega_{k-1}$ for $h_0$)
that for $k \geq 1$ we have
\[
\hpd{\omega_k}{4}{6}(\Dijr{i}{j}{r}) =
\hpd{\omega_{k-1}}{5}{6}(0) +
\hpd{\omega_k}{6}{6}(\Dijr{i}{j}{r}) = \Dijr{i}{j}{r}
\]
and we thus infer that
\begin{equation}
\label{eq:Dijr-false} 
\forall j \geq 3:\
\hpd{\omega}{4}{6}(\Dijr{i}{j}{r}) = \Dijr{i}{j}{r}.
\end{equation}
Similarly,
clause~\ref{def:evalk:cond} in Definition~\ref{def:HHX}
also gives us
\begin{equation}
\label{eq:Dijr-true}
\forall j < 3, \forall k \geq 1:\
\hpd{\omega_k}{4}{6}(\Dijr{i}{j}{r}) = 
\hpd{\omega_{k-1}}{5}{6}(\Dijr{i}{j}{r}).
\end{equation}
Also, clause~\ref{evalk-transitive} 
in Definition~\ref{def:HHX} (and the definition of $\omega_0$) gives us
\begin{equation}
\label{eq:Dijr-trans}
\forall k \geq 0, \forall D \in Dist:\
\hpd{\omega_k}{5}{6}(D) = \hpd{\omega_k}{4}{6}(\hpd{\omega_k}{5}{4}(D)).
\end{equation}
We shall now look at the various cases for the assignment at node 5.
\begin{description}
\item[$\assignv{y}{1}$]
For all $k \geq 1$, and all $j < 3$, we have
$\hpd{\omega_k}{5}{4}(\Dijr{i}{j}{r}) = \Dijr{i}{1}{r}$
and by (\ref{eq:Dijr-trans}) thus
$\hpd{\omega_k}{5}{6}(\Dijr{i}{j}{r}) = \hpd{\omega_k}{4}{6}(\Dijr{i}{1}{r})$
(which also holds for $k \geq 0$).
Thus from (\ref{eq:Dijr-true}) we get that
$\hpd{\omega_k}{4}{6}(\Dijr{i}{j}{r}) = 
\hpd{\omega_{k-1}}{4}{6}(\Dijr{i}{1}{r})$ for all $k \geq 1$ and $j < 3$.
As $\omega_0 = 0$, we see by induction that
$\hpd{\omega_k}{4}{6}(\Dijr{i}{j}{r}) = 0$ 
for all $k \geq 0$ and $j < 3$, and for all $j < 3$ we thus have
\[
\hpd{\omega}{4}{6}(\Dijr{i}{j}{r}) = 0
\]
which confirms that from node 4 
the probability of termination is zero (actually termination is impossible)
and that certainly $\hpd{\omega}{4}{6}$ is not sum-preserving.
\item[$\assignv{y}{y+1}$]
For all $k \geq 1$, and all $j < 3$, we have
$\hpd{\omega_k}{5}{4}(\Dijr{i}{j}{r}) = \Dijr{i}{j+1}{r}$
and by (\ref{eq:Dijr-trans}) thus
$\hpd{\omega_k}{5}{6}(\Dijr{i}{j}{r}) = \hpd{\omega_k}{4}{6}(\Dijr{i}{j+1}{r})$
(which also holds for $k \geq 0$).
Thus from (\ref{eq:Dijr-true}) we get that
$\hpd{\omega_k}{4}{6}(\Dijr{i}{j}{r}) = 
\hpd{\omega_{k-1}}{4}{6}(\Dijr{i}{j+1}{r})$ for all $k \geq 1$ and $j < 3$.
We infer that for
all $j < 3$, and all $k > 3-j$, 
\[
\hpd{\omega_k}{4}{6}(\Dijr{i}{j}{r}) =
\hpd{\omega_{k-(3-j)}}{4}{6}(\Dijr{i}{3}{r}) = \Dijr{i}{3}{r}
\]
and thus we infer that for all $j < 3$ we have
\[
\hpd{\omega}{4}{6}(\Dijr{i}{j}{r}) = \Dijr{i}{3}{r}
\]
which together with (\ref{eq:Dijr-false}) confirms that
$\hpd{\omega}{4}{6}$ is sum-preserving as
any loop from node 4 will eventually terminate.
\item[$\rassignv{y}{\psi_4}$]
For all $k \geq 1$, and all $j < 3$, 
we have 
\[
\hpd{\omega_k}{5}{4}(\Dijr{i}{j}{r}) = 
0.25 \cdot (\Dijr{i}{0}{r} + \Dijr{i}{1}{r} + \Dijr{i}{2}{r} + \Dijr{i}{3}{r})
\]
and by (\ref{eq:Dijr-trans}), together with the fact
(Lemma~\ref{lem:fixed-mult-nonincr-determ}) that
$\hpd{\omega_k}{4}{6}$ is additive and multiplicative, thus
\[
\hpd{\omega_k}{5}{6}(\Dijr{i}{j}{r}) = 
0.25 \cdot (\hpd{\omega_k}{4}{6}(\Dijr{i}{0}{r}) + 
\hpd{\omega_k}{4}{6}(\Dijr{i}{1}{r}) + 
\hpd{\omega_k}{4}{6}(\Dijr{i}{2}{r}) + 
\hpd{\omega_k}{4}{6}(\Dijr{i}{3}{r}))
\]
(which also holds for $k \geq 0$)
so from (\ref{eq:Dijr-true}) we get that
\begin{equation}
\label{eq:Dijr-recur}
\forall k \geq 1, j < 3:
\hpd{\omega_k}{4}{6}(\Dijr{i}{j}{r}) = 
0.25 \cdot \left(\sum_{q = 0,1,2,3}\hpd{\omega_{k-1}}{4}{6}(\Dijr{i}{q}{r})\right).
\end{equation}
One can easily prove by induction in $k$ that if
$j_1 < 3$ and $j_2 < 3$ then
$\hpd{\omega_k}{4}{6}(\Dijr{i}{j_1}{r}) =
\hpd{\omega_k}{4}{6}(\Dijr{i}{j_2}{r})$
so if we define
$D_k = \hpd{\omega_k}{4}{6}(\Dijr{i}{0}{r})$
we have
$\hpd{\omega_k}{4}{6}(\Dijr{i}{j}{r}) = D_k$ for all $j < 3$.
We shall now establish 
\begin{equation}
\label{eq:Dijr-limit}
\limit{k}{D_k} = \Dijr{i}{3}{r}
\end{equation}
which together with (\ref{eq:Dijr-false}) will demonstrate
that $\hpd{\omega}{4}{6}$ is sum-preserving as
any loop from node 4 will terminate with probability 1.

To show (\ref{eq:Dijr-limit}), observe that 
(\ref{eq:Dijr-recur}) together with (\ref{eq:Dijr-false})
makes it easy to prove by induction that
$D_k(s) = 0 = \Dijr{i}{3}{r}(s)$ 
for all $k \geq 0$ when
$s  \neq \stotwo{x}{i}{y}{3}$, 
and also gives the recurrences
\begin{eqnarray*}
D_0(s_3) & = & 0 \\
D_1(s_3) & = & 0 \\
D_k(s_3) & = & 0.75 \cdot D_{k-1}(s_3) + 0.25 \cdot r \mbox{ for }
k \geq 2
\end{eqnarray*}
when $s_3 = \stotwo{x}{i}{y}{3}$.
We must prove that $\limit{k}{D_k(s_3)} = r$
(as $\Dijr{i}{3}{r}(s_3) = r$)
but this follows, with $a = 0.75$ and $b = 0.25$, from a general result: 
\begin{lemma}
If $\chain{x_i}{i}$ is a sequence of non-negative reals, satisfying 
$x_0 = x_1 = 0$ and $x_k = a\,x_{k-1} +b\,r$ for $k > 1$ 
where $a,b,r$ are non-negative reals
with $b > 0$ and $a + b = 1$, then $\limit{i}{x_i} = r$.
\end{lemma}
\begin{proof}
Observe that: \emph{(i)} $\chain{x_i}{i}$ is a chain
(as can be seen by induction since $x_0 = x_1 \leq x_2$ and
if $x_k \leq x_{k+1}$ then $x_{k+1} \leq x_{k+2}$);
\emph{(ii)} $x_i \leq r$ for all $i$ since if $x_k > r$ for some $k$
then  $x_{k+1} = ax_k + br = (1-b)x_k + br = x_k + b(r-x_k) < x_k$
which contradicts $\chain{x_i}{i}$ being a chain;
\emph{(iii)} thus $\limit{i}{x_i} < \infty$ and
since $\limit{i}{x_i} = a \cdot \limit{i}{x_i} + br$
we get $b \cdot \limit{i}{x_i} = (1-a)\limit{i}{x_i} = br$
from which we infer the desired $\limit{i}{x_i} = r$.
\end{proof}
\end{description}
%\end{exmp}
\end{example}
The final two lemmas below provide justification that Definition~\ref{def:outside}
is indeed useful:
if $v$ stays outside $Q$ until $v'$ then the sliced program behaves
as the identity from $v$ to $v'$, and so does
the original program --- at least on the relevant variables --- if 
it is sum-preserving,
\begin{lemma}
\label{lem:phi1-id}
Given a pCFG, a slice set $Q$, and $(v,v') \in \PD$ such that
$v$ stays outside $Q$ until $v'$. 
We then have
$\hpd{\HH{Q}(h)}{v}{v'}(D) = D$ for all $D \in \Dist$
and all modification functions $h$.
\end{lemma}
\begin{lemma}
\label{lem:sum-preserve-not-affect-rv}
Let $(v,v') \in \PD$. Assume that
$Q$ is closed under data dependence and that $v$ stays outside $Q$
until $v'$ (by Lemma~\ref{lem:outside-same-rv} it thus makes sense
to define $R = \rv{Q}{v} = \rv{Q}{v'}$).

%% Let $\omega = \fixed{\HH{\unodes}}$. 
For all distributions $D$,
if $\sumd{\hpd{\omega}{v}{v'}(D)} = \sumd{D}$ 
then $\dagree{\hpd{\omega}{v}{v'}(D)}{D}{R}$.
\end{lemma}

%\section{Consistency with Denotational Semantics}
%\label{sec:consistent}
%\input{consistency}

\section{Conditions for Slicing}
\label{sec:conditions}
With $Q$ the slice set,
we now develop conditions for $Q$ that ensure semantic correctness.
It is standard to require $Q$ to be closed
under data dependence, cf.~Def.~\ref{def:datadep},
and additionally also under some kind of
``control dependence'', a concept which in this section 
we shall elaborate on and then study the extra conditions
needed in our probabilistic setting. 
Eventually, Definition~\ref{defn:slicing-pair} 
gives conditions that involve not only $Q$ but also
another slice set $Q_0$ containing all $\observevv$ nodes to 
be sliced away.
As stated in Proposition~\ref{prop:indpd2indpd} (Section~\ref{subsec:prob-indpd}),
these conditions are sufficient
to establish probabilistic independence 
of $Q$ and $Q_0$. This in turn is crucial for establishing
the correctness of slicing, as stated in 
Theorem~\ref{thm:slicing-correct} (Section~\ref{subsec:correct-slicing}).

\paragraph*{Weak Slice Sets}
Danicic~\etal~\cite{Danicic+etal:TCS-2011} showed
that various kinds of control dependence can all be elegantly expressed
within a general framework whose core is the following notion:
\begin{definition}[next visible]
\rm
With $Q$ a set of nodes, $v'$ is a \dt{next visible} in $Q$ of $v$ iff
$v' \in Q \cup \mkset{\finalv}$, and
$v'$ occurs on all paths from $v$ to a node in $Q \cup \mkset{\finalv}$.
\end{definition}
A node $v$ can have at most one next visible in $Q$.
It thus makes sense to write
$v' = \nextq{Q}{v}$ if 
$v'$ is a next visible in $Q$ of $v$.
We say that $Q$ \dt{provides next visibles}
iff $\nextq{Q}{v}$ exists for all nodes $v$.
If $v' = \nextq{Q}{v}$ then $v'$ is a postdominator of $v$,
and if $v \in Q \cup \mkset{\finalv}$ then 
$\nextq{Q}{v} = v$.

In the pCFG for $P_3$ (Figure~\ref{fig:ex34r}(left)),
letting $Q = \{1,3,5\}$, node 5 is a next visible
in $Q$ of 4: all paths from 4 to a node in $Q$ will contain 5.
But no node is a next visible in $Q$ of 2:
node 3 is not since there is a path from 2 to 5 not containing 3,
and node 5 is not since there is a path from 2 to 3 not containing 5.
Therefore $Q$ cannot be the 
slice set: node 1 can have only one successor in the sliced program
but we have no reason to choose either of the nodes 3 and 5 over the other
as that successor. This motivates the following definition:
\begin{definition}[weak slice set]
\rm
We say that $Q$ is a \dt{weak slice set} iff it 
provides next visibles, and is closed under data dependence.
\end{definition}
While the importance of ``provides next visible''
was recognized already in
\cite{Ran+Amt+Ban+Dwy+Hat:TOPLAS-2007,Amtoft:IPL-2007},
Danicic \etal\ %\cite{Danicic+etal:TCS-2011}
were the first to realize that it is
\emph{the} key property (together with data dependence) to ensure
semantically correct slicing.
They call the property ``weakly committing'' 
(thus our use of ``weak'') and our definition differs slightly from theirs
in that we always consider $\finalv$ as ``visible''. 

Observe that the empty set is a weak slice set,
since it is vacuously closed under data dependence,
and since for all $v$ we have $\finalv = \nextq{\emptyset}{v}$;
also the set $\unodes$ of all nodes is a weak slice set
since it is
trivially closed under data dependence,
and since for all $v$ we have $v = \nextq{\unodes}{v}$.
The property of being a weak slice set is also 
%% (as proved in Appendix~\ref{app:proofs})
closed under union:
\begin{lemma}
\label{lem:weak-union}
If $Q_1$ and $Q_2$ are weak slice sets,
also $Q_1 \cup Q_2$ is a weak slice set.
\end{lemma}
The following result is frequently used:
\begin{lemma}
\label{lem:notQ-outside}
Assume that $Q$ and $v$ are such that
$\nextq{Q}{v}$ exists, and that
$v \notin Q \cup \mkset{\finalv}$.
Then $v$ stays outside $Q$ until $\fppd{v}$.
\end{lemma}
\begin{proof}
Let $v' = \fppd{v}$ and
assume, to get a contradiction, that $\pi$ is a path from 
$v$ to $v'$ where $v'$ occurs only at the end
and which contains a node in $Q \setminus \mkset{v'}$;
let $v_0$ be the first such node. We infer that $v_0 \neq v$ (as $v \notin Q$)
and $v_0 \neq v'$,
and that $v_0 = \nextq{Q}{v}$.
Thus $v_0$ is a proper postdominator of $v$, which entails
(since $v' = \fppd{v}$) that $v'$ occurs on all paths from
$v$ to $v_0$. For the path $\pi$ it is thus the case that
$v'$ occurs before $v_0$,
which contradicts our assumption that $v'$ occurs only at the end.
\end{proof}
A weak slice set that covers all cycles must include
all nodes that are cycle-inducing (cf.~Definition~\ref{def:cycle-induce}):
\begin{lemma}
\label{lem:cycleQinduce}
Assume that $Q$ provides next visibles,
and that each cycle contains a node in $Q$.
Then all cycle-inducing nodes belong to $Q$.
\end{lemma}
\begin{proof}
Let $v$ be a cycle-inducing node, and let $v' = \fppd{v}$
and $v'' = \nextq{Q}{v}$.
By Lemma~\ref{lem:LAPcycle} there is a cycle $\pi$ which contains $v$ but not $v'$.
By assumption, there exists $v_1$ that belongs to $\pi$ and also to $Q$.
Since there is a path within $\pi$ from $v$ to $v_1$, $v''$ will occur on that path;
hence $v''$ belongs to $\pi$, and $v'' \in Q$ (as a cycle cannot contain $\finalv$). 

We know that $v''$ is a postdominator of $v$.
Assume, to get a contradiction, that $v''$ is a proper postdominator of $v$;
then $v'$ will occur on all paths from $v$ to $v''$, and hence $v'$ belongs to $\pi$
which is a contradiction. We infer $v'' = v$ and thus the desired $v \in Q$.
\end{proof}

\subsection{Adapting to the Probabilistic Setting}
As already motivated through Examples~\ref{ex1}--\ref{ex4}, the key challenge in slicing probabilistic programs is how to handle $\observevv$ nodes.
In Section~\ref{sec:examples} we hinted at some tentative conditions
a slice set $Q$ must satisfy; we can now phrase them more precisely:
\begin{enumerate}
\item
$Q$ must be a weak slice set that contains $\finalv$, and
\item
there exists another weak slice set $Q_0$ such that
(a) $Q$ and $Q_0$ are disjoint and 
(b) all $\observevv$ nodes belong to either $Q$ or $Q_0$.
\end{enumerate}
We shall now see how these conditions work out for our example programs.

For programs $P_1, P_2$,
the control flow is linear and hence all nodes have a next
visible, no matter the choice of $Q$;
thus a node set is a weak slice set iff it is closed under data dependence.

For $P_1$ we may choose $Q = \{1,4\}$ and $Q_0 = \{2,3\}$
as they are disjoint, and both closed under data dependence.
As can be seen from Defs.~\ref{def:slice-meaning} and~\ref{def:HHX},
the resulting sliced program has the same meaning as the program that results from
$P_1$ by replacing all nodes not in $Q$ by $\skipv$, that is
\[
\begin{array}{ll}
1: & \rassignv{x}{\psi_4};\\
2: & \skipv;\\
3: & \skipv;\\
4: & \retv{x}
\end{array}
\]
which is obviously equivalent to $P_x$ as defined
in Section~\ref{sec:examples}.

Next consider the program $P_2$ 
where $Q$ should contain 4 and hence (by data dependence) also contain 1.
Now assume, in order to remove the $\observevv$ node (and produce $P_x$),
that $Q$ does not contain 3. Then $Q_0$ must contain 3, and
(as $Q_0$ is closed under data dependence) also 1.
But then $Q$ and $Q_0$ are not disjoint, which contradicts our
requirements. Thus $Q$ does contain 3, and hence also 2.
That is, $Q = \{1,2,3,4\}$. We see that the only possible slicing
is the trivial one.

Any slice for $P_3$ (Figure~\ref{fig:ex34r}) will also be trivial. 
From $5 \in Q$ we infer (by data dependence) that $1 \in Q$.
Assume, to get a contradiction, that $4 \notin Q$.
As $4$ is an $\observevv$ node
we must thus have $4 \in Q_0$, and for node 2 to have a next visible in $Q_0$
we must then also have $2 \in Q_0$ 
which by data dependence implies $1 \in Q_0$ which 
contradicts $Q$ and $Q_0$ being disjoint.
This shows $4 \in Q$ which implies $3 \in Q$ (by data dependence)
and $2 \in Q$ (as otherwise 2 has no next visible in $Q$).

For $P_4$, we need $6 \in Q$ and by data dependence thus also
$1 \in Q$; actually,
our tentative conditions can be satisfied by choosing
$Q = \{1,6\}$ and $Q_0 = \emptyset$, as for all $v \neq 1$
we would then have $6 = \nextq{Q}{v}$.
From Definitions~\ref{def:slice-meaning} and~\ref{def:HHX}
we see that the resulting sliced program has the same meaning as
\[
\begin{array}{ll}
1: & \rassignv{x}{\psi_4};\\
2: & \skipv;\\
3: & \skipv;\\
6: & \retv{x} 
\end{array}
\]
Yet, in Example~\ref{ex4} we saw that in
general (as when node $C$ is labeled $\assignv{y}{1}$)
this is \emph{not} a correct slice of $P_4$.
This reveals a problem with our tentative correctness conditions;
they do not take into account that
$\observevv$ nodes may be ``encoded'' as infinite loops.

To repair that, we shall demand
that just like all $\observevv$ nodes
must belong to either $Q$ or $Q_0$, also all cycles
must touch either $Q$ or $Q_0$,
except if the cycle is known
to terminate with probability 1.
Allowing this exception is an added contribution to 
the conference version of this article~\cite{Amt+Ban:FoSSaCS-2016}.

Observe that all cycles touch either $Q$ or $Q_0$
iff all cycle-inducing nodes belong to $Q$ or $Q_0$:
``only if'' follows from Lemma~\ref{lem:cycleQinduce}
(and~\ref{lem:weak-union}),
and ``if'' follows from Lemma~\ref{lem:cycle1inducing}.

We are now done motivating how to arrive at the following definition
which shall serve as our final stipulation of what is a correct slice:
\begin{definition}[slicing pair]
\label{defn:slicing-pair}
Let $Q, Q_0$ be sets of nodes.
$(Q,Q_0)$ is a \dt{slicing pair} iff
\begin{enumerate}
\item
$Q, Q_0$ are both weak slice sets with $\finalv \in Q$;
\item
$Q, Q_0$ are disjoint;
\item
all $\observevv$ nodes are in $Q \cup Q_0$; and
\item
for all cycle-inducing nodes $v$, either
\begin{enumerate}
\item
$v \in Q \cup Q_0$, or
\item
$\hpd{\omega}{v}{\fppd{v}}$ is sum-preserving
(in which case we shall say that $v$ is sum-preserving).
\end{enumerate}
\end{enumerate}
\end{definition}
%% Observe by Lemma~\ref{lem:weak-union} that
%% if $(Q,Q_0)$ is a slicing pair then
%% $Q \cup Q_0$ is a weak slice set.

For Example~\ref{ex4}, we saw in Example~\ref{ex:cycle-induce}
that 4 is the only cycle-inducing node,
and we must demand either $4 \in Q \cup Q_0$ or 
that $\hpd{\omega}{4}{6}$ is sum-preserving.
Recalling the findings in Example~\ref{ex:sem-ex}, we see that:
\begin{enumerate}
\item
If node 5 is labeled $\assignv{y}{y+1}$ or
$\rassignv{y}{\psi_4}$ then
we don't need $4 \in Q \cup Q_0$, and thus
$(\{1,6\},\emptyset)$ is a valid slicing pair.
\item
If node 5 is labeled $\assignv{y}{1}$ then
we must require $4 \in Q \cup Q_0$.
But $4 \in Q_0$ is impossible, as then $3 \in Q_0$ 
(since otherwise 3 has no next visible in $Q_0$)
which by data dependence implies $1 \in Q_0$ which 
contradicts $Q \cap Q_0 = \emptyset$, since $1 \in Q$.
Thus $4 \in Q$, and then $3 \in Q$ (since otherwise
3 has no next visible in $Q$) and $2,5 \in Q$ (by data dependence).
We see that $Q$ contains all nodes, giving a trivial slice.
\end{enumerate}
The concepts involved in Definition~\ref{defn:slicing-pair} help
towards establishing a result showing
%% proved in Appendix~\ref{app:proofs},
that a larger class of semantic functions are sum-preserving:
\begin{lemma}
\label{lem:sum-pres}
Assume $Q'$ is a node set which 
contains all $\observevv$ nodes, and 
that for each cycle-inducing node $v_0$,
either $v_0 \in Q'$ or $\hpd{\omega}{v_0}{\fppd{v_0}}$ is sum-preserving.
If $v$ stays outside $Q'$ until $v'$
then $\sumd{\hpd{\omega}{v}{v'}(D)} = \sumd{D}$ for all $D$.
\end{lemma}
We can now state a key result which shows that
nodes not in a slicing pair are not relevant for computing
the final result:
\begin{lemma}
\label{lem:outsideQQ0irrelevant}
Assume that $(Q,Q_0)$ is a slicing pair,
and that $v$ stays outside $Q \cup Q_0$ until $v'$.
With $R = \rv{Q \cup Q_0}{v} = \rv{Q \cup Q_0}{v'}$ 
(equality holds by Lemma~\ref{lem:outside-same-rv}),
we have $\dagree{\hpd{\omega}{v}{v'}(D)}{D}{R}$ for all $D$.
\end{lemma}
\begin{proof}
With the given assumptions, Lemma~\ref{lem:sum-pres}
is applicable (with $Q' = Q \cup Q_0$) to establish that for all $D$ we have
$\sumd{\hpd{\omega}{v}{v'}(D)} = \sumd{D}$.
The claim now follows from Lemma~\ref{lem:sum-preserve-not-affect-rv}.
\end{proof}

\section{Slicing and its Correctness}
\label{sec:correct}
In this section, we shall embark on proving the semantic correctness 
of the slicing conditions developed in Section~\ref{sec:conditions}.
This will involve reasoning about the behavior of $\hpd{\omega}{v}{v'}$
for $(v,v') \in \PD$, but which reasoning principle should we employ?
Just doing induction in $\LAP{v}{v'}$ will obviously not 
work for a cycle-inducing node;
instead, the following approach is often feasible:
\begin{enumerate}
\item
prove results about $\hpd{\omega_k}{v}{v'}$ by induction in $k$,
where for each $k$ we do an inner induction on $\LAP{v}{v'}$
(possible since $\omega_{k+1} = \HH{\unodes}(\omega_{k})$
where $\hpd{\HH{\unodes}(\omega_{k})}{v}{v'}$ is defined
inductively in $\LAP{v}{v'}$);
\item
lift the results about $\hpd{\omega_k}{v}{v'}$
to results about the limit $\hpd{\omega}{v}{v'}$.
\end{enumerate}
Unfortunately, this approach does not work for certain properties,
such as being sum-preserving as this will hold only in the limit.
We need to be more clever!

Our idea is, given a slicing pair $(Q,Q_0)$ which will serve as implicit parameters, 
to introduce a family ($k \geq 0$) of
functions
\[
  \gamma_k: \PD \rightarrow (\Dist \contarrow \Dist)
\]
such that we can prove that $\chain{\gamma_k}{k}$ is a chain
with $\limit{k}{\gamma_k} = \omega$. Of course, already $\chain{\omega_k}{k}$ is such
a chain, but we shall define $\gamma_k$ in a way such that we can
still reason by induction, but also get certain properties already from
the beginning of the fixed point iteration, not just at the limit.
This is achieved by this inductive definition:
\begin{definition}
\label{def:gamma-k}
Given a slicing pair $(Q,Q_0)$, for $k \geq 0$ define $\gamma_k$ as follows:
\begin{eqnarray*}
\hpd{\gamma_0}{v}{v'} & = & \hpd{\omega}{v}{v'} \mbox{ if $v$ stays outside $Q \cup Q_0$ until $v'$} \\
\hpd{\gamma_0}{v}{v'} & = & 0 \hspace*{4mm} \mbox{ otherwise} \\[1mm]
\gamma_{k} & = & \HH{\unodes}(\gamma_{k-1}) \mbox{ for $k > 0$}
\end{eqnarray*}
\end{definition}
%% In Appendix~\ref{app:proofs} we shall prove the following result: 
\begin{lemma}
\label{lem:gamma-as-omega}
Assume that $v$ stays outside $Q \cup Q_0$ until $v'$.
Then $\hpd{\gamma_k}{v}{v'} = \hpd{\omega}{v}{v'}$ for all $k \geq 0$.
\end{lemma}
Observe that $\hpd{\gamma_0}{v}{v'} \leq \hpd{\gamma_1}{v}{v'}$ holds for all
$(v,v') \in \PD$ since by
Lemma~\ref{lem:gamma-as-omega} we have equality when $v$ stays outside
$Q \cup Q_0$ until $v'$, and the left hand side is 0 otherwise.
Thus $\gamma_0 \leq \gamma_1$ which enables us 
(since $\HH{\unodes}$ is monotone) 
to infer inductively that $\chain{\gamma_k}{k}$ is a chain.
Moreover, since 
$0 = \omega_0 \leq \gamma_0 \leq \omega$ trivially holds,
we can inductively 
(since $\omega$ is a fixed point of $\HH{\unodes}$)
infer that
$\omega_k \leq \gamma_k \leq \omega$ for all $k \geq 0$.
This allows us to deduce
that $\limit{k}{\gamma_k} = \omega$;
we have thus proved:
\begin{proposition}
\label{prop:gamma-limit-omega}
The sequence $\chain{\gamma_k}{k}$ is a chain,
with $\limit{k}{\gamma_k} = \omega$.
\end{proposition}

\subsection{Probabilistic Independence}
\label{subsec:prob-indpd}
We shall now show a main contribution of this paper:
that we have provided (in Definition~\ref{defn:slicing-pair}) 
syntactic conditions for probabilistic independence,
in that the $Q$-relevant
variables are probabilistically independent (as defined in Definition~\ref{def:indpd})
of the $Q_0$-relevant variables:
\begin{proposition}[Independence]
\label{prop:indpd2indpd}
Let $(Q,Q_0)$ be a slicing pair.
Assume that $\hpd{\omega}{v}{v'}(D) = D'$.
If $\rv{Q}{v}$ and $\rv{Q_0}{v}$ are independent in $D$
then $\rv{Q}{v'}$ and $\rv{Q_0}{v'}$ are independent in $D'$.
\end{proposition}
This follows from a more general result:
%% proved in Appendix~\ref{app:proofs}:
\begin{lemma}
\label{lem:indpd2indpd}
Let $(Q,Q_0)$ be a slicing pair.
Assume that $\hpd{\omega}{v}{v'}(D) = D'$.
Let $R = \rv{Q}{v}$, $R' = \rv{Q}{v'}$,
$R_0 = \rv{Q_0}{v}$, and $R'_0 = \rv{Q_0}{v'}$.
If $R$ and $R_0$ are independent in $D$ then
\begin{enumerate}
\item 
%% \label{indpdlem:indp}
$R'$ and $R'_0$ are independent in $D'$
\item
%% \label{indpdlem:outq}
if $v$ stays outside $Q$ until $v'$
(and by Lemma~\ref{lem:outside-same-rv} thus $R' = R$)
then for all $s \in \sto{R}$ we have
\begin{displaymath}
D(s)\sumd{D'} = D'(s)\sumd{D}
\end{displaymath}
\item
%% \label{indpdlem:outq0}
if $v$ stays outside $Q_0$ until $v'$
(and thus $R'_0 = R_0$)
then for all $s_0 \in \sto{R_0}$ we have
\begin{displaymath}
D(s_0)\sumd{D'} = D'(s_0)\sumd{D}
\end{displaymath}
\end{enumerate}
\end{lemma}

\subsection{Correctness of Slicing}
\label{subsec:correct-slicing}

We can now precisely phrase the desired correctness result,
which (as hinted at in Section~\ref{sec:examples})
states that the sliced program produces 
the same \emph{relative} distribution
over the values of the relevant variables
as does the original program, and will be at least as ``defined'':
\begin{thm}
\label{thm:slicing-correct}
For a given pCFG,
let $(Q,Q_0)$ be a slicing pair,
and let $\phi = \fixed{\HH{Q}}$ 
(cf.~Definition~\ref{def:slice-meaning})
be the meaning of the sliced program.
For a given $(v,v') \in \PD$, and a given $D \in \Dist$
such that $\rv{Q}{v}$ and $\rv{Q_0}{v}$ are independent in $D$,
there exists a real number $c$ (depending on $v,v'$ and $D$)
with $0 \leq c \leq 1$ such that
\[
\dagree{\hpd{\omega}{v}{v'}(D)}{c \cdot \hpd{\phi}{v}{v'}(D)}{\rv{Q}{v'}}.
\]
Moreover, if $v$ stays outside $Q_0$ until $v'$ then $c = 1$.
\end{thm}
We need to assume that the $Q$-relevant and $Q_0$-relevant variables are independent, 
so as to allow observe nodes in $Q_0$ to be sliced away 
(since then such nodes will not change the relative distribution of 
the $Q$-relevant variables), and also to allow certain branching nodes
to be sliced away.

To prove Theorem~\ref{thm:slicing-correct}
(as done at the end of this section), 
we need a result that involves $\gamma_k$ 
as introduced in
Definition~\ref{def:gamma-k},
and which also (so as to facilitate a proof by induction in $\LAP{v}{v'}$)
allows the sliced program to be given a distribution that, while
agreeing on the relevant variables, may differ from the distribution
given to the original program:
\begin{lemma}
\label{lem:slicing-correct}
For a given pCFG,
let $(Q,Q_0)$ be a slicing pair.
For all $k \geq 0$, all $(v,v') \in \PD$
with $R = \rv{Q}{v}$ and $R' = \rv{Q}{v'}$ and $R_0 = \rv{Q_0}{v}$,
all $D \in \Dist$
such that $R$ and $R_0$ are independent in $D$,
and all $\dlt \in \Dist$ such
that $\dagree{D}{\dlt}{R}$, we have
\[
\dagree{\hpd{\gamma_k}{v}{v'}(D)}{\cnst{k}{v}{v'}{D} \cdot 
\hpd{\Phi_k}{v}{v'}(\dlt)}{R'}.
\]
\end{lemma}
Here the numbers $\cnst{k}{v}{v'}{D}$, and the modification functions
$\Phi_k$, are defined below (Definitions~\ref{def:slicing-const} and \ref{def:phipk}).

\begin{definition}
\label{def:slicing-const}
For $k \geq 0$, $(v,v') \in \PD$, and $D \in \Dist$, the number
$\cnst{k}{v}{v'}{D}$ is given by the following rules that are
inductive in $\LAP{v}{v'}$:
\begin{enumerate}
\item
if $v = v'$ then $\cnst{k}{v}{v'}{D} = 1$
\item
otherwise, if $v' \neq v''$ where $v'' = \fppd{v}$ then
\[
\cnst{k}{v}{v'}{D} = \cnst{k}{v}{v''}{D} \cdot 
\cnst{k}{v''}{v'}{\hpd{\gamma_k}{v}{v''}(D)}
\]
\item
otherwise, if $v$ stays outside $Q_0$ until $v'$
then $\cnst{k}{v}{v'}{D} = 1$
\item
otherwise, if $D = 0$ then
$\cnst{k}{v}{v'}{D} = 1$ else
\[
\cnst{k}{v}{v'}{D} = \frac{\sumd{\hpd{\gamma_k}{v}{v'}(D)}}{\sumd{D}}.
\]
\end{enumerate}
\end{definition}
We could have swapped the order of the first three clauses
of Definition~\ref{def:slicing-const}, since it is easy to prove
by induction in $\LAP{v}{v'}$ that
\begin{lemma}
\label{lem:cnst1}
If $v$ stays outside $Q_0$ until $v'$ then 
$\cnst{k}{v}{v'}{D} = 1$ for all $k \geq 0$ and $D \in \Dist$.
\end{lemma}
Since each $\gamma_k$ is non-increasing (cf.~Lemma~\ref{lem:gamma-mult-nonincr}),
it is easy to prove by induction in $\LAP{v}{v'}$ that
\begin{lemma}
We have $0 \leq \cnst{k}{v}{v'}{D} \leq 1$ for
all $k \geq 0$, $(v,v') \in \PD$, $D \in \Dist$.
\end{lemma}
Since we know (Proposition~\ref{prop:gamma-limit-omega})
that $\chain{\gamma_k}{k}$ is a chain, we get:
\begin{lemma}
\label{lem:cnst-chain}
$\chain{\cnst{k}{v}{v'}{D}}{k}$ is a chain
for each $(v,v') \in \PD$ and $D \in \Dist$.
\end{lemma}

\begin{definition}
\label{def:phipk}
Given a slicing pair $(Q,Q_0)$, for $k \geq 0$ define $\Phi_k$ as follows:
\begin{eqnarray*}
\hpd{\Phi_0}{v}{v'}(D) & = & D \mbox{ if $v$ stays outside $Q \cup Q_0$ until $v'$} \\
\hpd{\Phi_0}{v}{v'}(D) & = & 0 \hspace*{4mm} \mbox{ otherwise} \\[1mm]
\Phi_{k} & = & \HH{Q}(\Phi_{k-1}) \mbox{ for $k > 0$}
\end{eqnarray*}
\end{definition}
\begin{lemma}
\label{lem:Phi-limit-phi}
$\chain{\Phi_k}{k}$ is a chain,
with $\limit{k}{\Phi_k} = \limit{k}{\phi_k} = \phi$.
\end{lemma}
\begin{proof}
By Lemma~\ref{lem:phi1-id}
we get
$\phi_0 \leq \Phi_0 \leq \phi_1$
so by the monotonicity of $\HH{Q}$ we inductively get
\[
\phi_k \leq \Phi_k \leq \phi_{k+1} \mbox{ for all } k \geq 0
\]
which yields the claim.
\end{proof}

\paragraph{Proof of Theorem~\ref{thm:slicing-correct}}
%% In Appendix~\ref{app:proofs} we prove Lemma~\ref{lem:slicing-correct}
%% which will establish Theorem~\ref{thm:slicing-correct} 
We are given $(v,v') \in \PD$, and $D$ such that
$\rv{Q}{v}$ and $\rv{Q_o}{v}$ are independent in $D$;
let $R' = \rv{Q}{v'}$.
For each $s' \in \sto{R'}$ we have the calculation
\begin{eqnarray*}
& & \hpd{\omega}{v}{v'}(D)(s') \\
(Proposition~\ref{prop:gamma-limit-omega}) & = &
\limit{k}{\hpd{\gamma_k}{v}{v'}(D)(s')} 
\\ (Lemma~\ref{lem:slicing-correct}) & = & 
\limit{k}{(\cnst{k}{v}{v'}{D} \cdot \hpd{\Phi_k}{v}{v'}(D)(s'))}
\\ (Lemma~\ref{lem:Phi-limit-phi} & = &
(\limit{k}{\cnst{k}{v}{v'}{D}}) \cdot \hpd{\phi}{v}{v'}(D)(s').
\end{eqnarray*}
With $c = \limit{k}{\cnst{k}{v}{v'}{D}}$
(well-defined by Lemma~\ref{lem:cnst-chain})
we thus have
\[
\dagree{\hpd{\omega}{v}{v'}(D)}{c \cdot 
\hpd{\phi}{v}{v'}(D)}{R'}
\]
which yields the result since if
$v$ stays outside $Q_0$ until $v'$
then $c = 1$
(by Lemma~\ref{lem:cnst1}).

\section{Computing the (Least) Slice}
\label{sec:alg-least}
There always exists at least one slicing pair,
with $Q$ the set of all nodes and with $Q_0$ the empty 
set;
in that case, the sliced program is the same as the original.
Our goal, however, is to find a slicing pair $(Q,Q_0)$ where
$Q$ is as small as possible.
This section describes an algorithm for doing so. 

Looking at Definition~\ref{defn:slicing-pair},
we see a couple of potential obstacles:
\begin{enumerate}
\item
Detecting whether a node is sum-preserving is undecidable,
as it is easy to see that the halting problem can be reduced to it
(see~\cite{Kam+Kat:MFSC-2015} for more results about the decidability of termination in a probabilistic setting).
\item
Since finding the longest acyclic path is in general an NP-hard problem
(as the Hamiltonian path problem can be reduced to it),
it may not be feasible in polynomial time to
detect whether a node is cycle-inducing --- but since we only
consider graphs where each node has at most two outgoing edges,
and since we do not need to actually compute the longest acyclic paths
but only to compare their lengths, there may still exist a polynomial
algorithm for checking if a node is cycle-inducing (finding such
an algorithm is a topic for future work).
\end{enumerate}
Therefore, our approach shall be to assume that we have been provided
(perhaps by an \emph{oracle}) a list $\ESS$ that
approximates the \emph{essential} nodes:
\begin{definition}[essential nodes]
A node $v$ is essential iff
\begin{enumerate}
\item
$v$ is an $\observevv$ node, or
\item
$v$ is cycle-inducing but not sum-preserving.
\end{enumerate}
\end{definition}
We can now provide a computable version of
Definition~\ref{defn:slicing-pair}:
\begin{definition}
\label{defn:slicing-pair-ess}
Let $\ESS$ be a set of nodes that contains all essential nodes.
Then $(Q,Q_0)$ is a \dt{slicing pair wrt.~$\ESS$} iff
\begin{enumerate}
\item
$Q, Q_0$ are both weak slice sets with $\finalv \in Q$;
\item
$Q, Q_0$ are disjoint;
\item
$\ESS \subseteq Q \cup Q_0$.
\end{enumerate}
\end{definition}
If we find $(Q,Q_0)$ satisfying Definition~\ref{defn:slicing-pair-ess}
then $(Q,Q_0)$ will also satisfy Definition~\ref{defn:slicing-pair}
(and hence Theorem~\ref{thm:slicing-correct}, etc, will apply):
\begin{proposition}
If $(Q,Q_0)$ is a slicing pair wrt.~$\ESS$ then $(Q,Q_0)$ is a slicing pair.
\end{proposition}
On the other hand, the converse does not necessarily hold
as $\ESS \not\subseteq Q \cup Q_0$ may happen if non-essential nodes are included in $\ESS$. For example, in Example~\ref{ex4} with
$C$ as ``$\assignv{y}{y + 1}$'' there are no essential nodes,
and thus (cf.~the discussion after Definition~\ref{defn:slicing-pair})
$(\{1,6\},\emptyset)$ is a slicing pair. However
if we were unable to infer that 4 is sum-preserving,
we may have $\ESS = \{4\}$, in which case
$(\{1,6\},\emptyset)$ is not a slicing pair wrt.~$\ESS$.

To approximate the essential nodes (which is outside the scope of this article)
one may use techniques from
\cite{Monniaux:SAS-2001,Cha+San:CAV-2013,Fio+Her:POPL-2015}
for detecting that loops terminate with probability one,
or techniques from \cite{Kam+Kat+etal:ESOP-2016}
for detecting a stronger property:
that the expected run-time is finite.

If the pCFG in question is a translation of a structured command
%% ~Section~\ref{sec:consistent}),
(cf.~the companion article~\cite{Amt+Ban:ProbSemantics-2017}),
it will be safe to let $\ESS$ contain (in addition to the $\observevv$ nodes)
the branching nodes created when translating while loops
(but $\ESS$ does not need to contain the branching nodes created
when translating conditionals since such nodes will not be cycle-inducing). 

With the set $\ESS$ given, we can now develop our algorithm
to find the least $Q$ that for some $Q_0$ satisfies the conditions
in Definition~\ref{defn:slicing-pair-ess}.
We shall measure its running time in terms 
of $|\unodes|$, the number of nodes in the pCFG;
we shall often write $n$ instead of $|\unodes|$
(note that the number of edges is at most $2n$ and thus in $O(n)$).

Our approach has four stages:
\begin{enumerate}
\item
to compute (Section~\ref{subsec:DD}) the data dependences
(in time $O(n^3)$);
\item
to construct an algorithm $\PN$ (Section~\ref{subsec:pn})
that (in linear time) checks is a given set of nodes provides next visibles,
and if not, returns a set of nodes that 
definitely needs to be added;
\item
to construct an algorithm $\LWS$ (Section~\ref{subsec:least-weak})
that computes the least weak slice set that contains a given set of nodes
(each call to $\LWS$ takes time in $O(n^2)$);
\item
to compute (Section.~\ref{subsec:least-pair}) the best slicing
pair wrt.~the given $\ESS$.
\end{enumerate}
The resulting algorithm $\BSP$ (for \underline{b}est \underline{s}licing
\underline{p}air) has a total running time in $O(n^3)$.

\subsection{Computing Data Dependencies}
\label{subsec:DD}

Our algorithms use a boolean table $\DDS$
such that $\DDS(v,v')$ is true
iff $\dds{v}{v'}$ where $\ddssym$ is the reflexive and
transitive closure of 
$\ddsym$ 
defined in Definition~\ref{def:datadep}.
\begin{lemma}
\label{lem:DDS-n3}
There exists an algorithm that computes $\DDS$ in time $O(n^3)$.
\end{lemma}
\begin{proof}
First, for each node $v$ with $\defv{v} \neq \emptyset$,
we find the nodes $v'$ with $\dd{v}{v'}$ which
can be done in time $O(n)$
by a depth-first search which does not go past the nodes
that redefine the variable defined in $v$.
Thus in time $O(n^2)$, we can compute a boolean table $\DD$ such that
$\DD(v,v')$ is true iff $\dd{v}{v'}$.
To compute $\DDS$ we now take the reflexive and transitive closure of $\DD$
which can be done in time $O(n^3)$ (for example
using Floyd's algorithm).
\end{proof}
Given $\DDS$, it is easy to ensure that sets are closed under data dependence,
and we shall
do that in an incremental way, as stated by the following result:
%% (proved in Appendix~\ref{app:proofs}):
\begin{lemma}
\label{lem:DDclose}
There exists an algorithm $\DDclose$
which given a node set $Q$ that is closed under data dependence,
and a node set $Q_1$,
returns the least set containing $Q$ and $Q_1$
that is closed under data dependence.
Moreover, assuming $\DDS$ is given,
$\DDclose$ runs in time $O(n \cdot |Q_1|)$.
\end{lemma}

\subsection{Checking For Next Visibles}
\label{subsec:pn}
%% The algorithm is quite similar in spirit to what was presented
%% in \cite[Sect.9]{Amtoft+etal:TR-2013} though the settting there
%% is somewhat different.
A key ingredient in our approach
is the function $\PN$, presented in Figure~\ref{fig:provides-next},
that for a given $Q$ checks if it provides next visibles,
and if not, returns a non-empty set of nodes which must be part of any
set that provides next visibles and contains $Q$.
The function $\PN$ works by doing a
backward breadth-first search (with $F$ being the current ``frontier'')
from $Q \cup \mkset{\finalv}$ 
to find (using the table $N$ that approximates ``next visible'')
the first node(s), if any, from which
two nodes in $Q \cup \finalv$ are reachable without going through $Q$; 
such ``conflict'' nodes are stored in $C$ and
must be included in any superset providing next visibles.

\begin{figure}
\algname{$\PN$}{$Q$}
\begin{algtab*}
\algbegin
  $F \leftarrow Q \cup \{\finalv\}$ \\
  $C \leftarrow \emptyset$ \\
  \algforeach{$v \in \unodes \setminus F$}
    $N[v] \leftarrow \bot$ \\
  \algend
  \algforeach{$v \in F$}
    $N[v] \leftarrow v$ \\
  \algend
  \algwhile{$F \neq \emptyset\ \wedge\ C = \emptyset$}
    $F' \leftarrow \emptyset$ \\
    \algforeach{edge from $v \notin Q$ to $v' \in F$} 
        \algif{$N(v) = \bot$}
           $N(v) \leftarrow N(v')$ \\
           $F' \leftarrow F' \cup \{v\}$ \\
        \algelsif{$N(v) \neq N(v')$}
           $C \leftarrow C \cup \{v\}$ \\
    \algend
    \algend
    $F \leftarrow F'$ \\
   \algend
  \algreturn{$C$}  
\end{algtab*}

\caption{\label{fig:provides-next} 
An algorithm to check if $Q$ provides next visibles.}
\end{figure}

\setlength{\parindent}{0.15in}

\begin{example}
%\begin{exmp}
\label{ex:PN1}
Consider the program $P_1$ from Example~\ref{ex1}.
\begin{itemize}
\item
Calling $\PN$ on $\{1,4\}$ returns $\emptyset$
after a sequence of iterations
where $F$ is first $\{1,4\}$ and next $\{3\}$ and next 
$\{2\}$ and finally $\emptyset$.
\item
Calling $\PN$ on $\{2,3\}$ returns $\emptyset$
after a sequence of iterations
where $F$ is first $\{1\}$ and finally $\emptyset$.
\end{itemize}
%\end{exmp}
\end{example}

%\begin{exmp}
\begin{example}
\label{ex:PN2}
Consider the program $P_4$ from Example~\ref{ex4},
with pCFG depicted in Figure~\ref{fig:ex34r}(right).
Then
\begin{itemize}
\item
$\PN(\{1,6\})$
returns $\emptyset$, after a sequence of iterations
where $F$ is first $\{1,6\}$ and next $\{3,4\}$ and next 
$\{2,5\}$ and finally $\emptyset$.
\item
$\PN(\{2,4,5\})$ returns $\{3\}$,
as initially $F = \{2,4,5,6\}$ which causes
the first iteration of the while loop to put $3$ in $C$.
\end{itemize}
%\end{exmp}
\end{example}
%% In Appendix~\ref{app:proofs} 
The following result establishes the correctness of $\PN$:
\begin{lemma}
\label{lem:PN}
The function $\PN$ runs in time $O(n)$ and, given $Q$,
returns $C$ such that $C \cap Q = \emptyset$ and
\begin{itemize}
\item
if $C$ is empty then $Q$ provides next visibles
\item
if $C$ is non-empty then %$C \cap Q = \emptyset$
%and 
all supersets of $Q$ that provide next
visibles will contain $C$.
\end{itemize}
\end{lemma}

\subsection{Computing Least Weak Slice Set}
\label{subsec:least-weak}
We are now ready to define, in Figure~\ref{fig:lws-alg},
a function $\LWS$ 
which constructs the least weak slice set that contains a given set $\hat{Q}$;
it works by successively adding nodes to the set until it 
is closed under 
data dependence, and provides next visibles. 

\begin{figure}
\algname{$\LWS$}{$\hat{Q}$}
\begin{algtab*}
\algbegin
  $Q \leftarrow \DDclose(\emptyset,\hat{Q})$ \\
  $C \leftarrow \PN(Q)$ \\
  \algwhile{$C \neq \emptyset$}
     $Q \leftarrow \DDclose(Q,C)$ \\ 
     $C \leftarrow \PN(Q)$ \\
  \algend
  \algreturn{$Q$}
\end{algtab*}

\caption{\label{fig:lws-alg}
An algorithm that finds the least weak slice set containing $\hat{Q}$.}
\end{figure}

\setlength{\parindent}{0.15in}

\begin{example}
%\begin{exmp}
\label{ex:LWS1}
We shall continue Example~\ref{ex:PN1}
(which considers the program $P_1$ from Example~\ref{ex1}).
First observe that the
non-trivial true entries of $\DDS$ 
are $(1,4)$ (since $\dd{1}{4}$) and $(2,3)$.
\begin{itemize}
\item
When running $\LWS$ on $\{4\}$,
initially $Q = \{1,4\}$
which is also the final value of $Q$ since 
$\PN(\{1,4\})$ returns $\emptyset$.
\item
When running $\LWS$ on $\{3\}$,
initially $Q = \{2,3\}$
which is also the final value of $Q$ since 
$\PN(\{2,3\})$ returns $\emptyset$.
\end{itemize}
%\end{exmp}
\end{example}

\begin{example}
%\begin{exmp}
\label{ex:LWS2}
We shall continue Example~\ref{ex:PN2}
(which considers the program $P_4$ from Example~\ref{ex4},
with pCFG depicted in Figure~\ref{fig:ex34r}(right)).

First observe that 
$\ddsym$ is given as follows:
$\dd{1}{3}$, $\dd{1}{6}$, $\dd{2}{4}$, $\dd{2}{5}$, $\dd{5}{4}$, 
and $\dd{5}{5}$.
\begin{itemize}
\item
When running $\LWS$ on $\{6\}$, initially $Q = \{1,6\}$
which is also the final value of $Q$ since $\PN(\{1,6\})$
returns $\emptyset$.
\item
When running $\LWS$ on $\{4\}$, we initially have $Q = \{2,4,5\}$.
The first call to $\PN$ thus (Example~\ref{ex:PN2}) returns $\{3\}$.
Since $\dd{1}{3}$ holds, the next iteration of $\LWS$ will have 
$Q = \{1,2,3,4,5\}$ which is also the final value of $Q$
since $\PN$ will return $\emptyset$ on that set.
\end{itemize}
%\end{exmp}
\end{example}
%% In Appendix~\ref{app:proofs} 
The following result establishes the correctness of $\LWS$:
\begin{lemma}
\label{lem:lws-correct}
The function $\LWS$, given $\hat{Q}$,
returns $Q$ such that 
\begin{itemize}
\item
$Q$ is a weak slice set
\item
$\hat{Q} \subseteq Q$
\item
if $Q'$ is a weak slice set with $\hat{Q} \subseteq Q'$
then $Q \subseteq Q'$.
\end{itemize}
Moreover, assuming $\DDS$ is given, 
$\LWS$ runs in time $O(n^2)$.
\end{lemma}

\subsection{Computing The Best Slicing Pair}
\label{subsec:least-pair}

We are now ready to define, in Figure~\ref{fig:best-slicing-pair},
an algorithm $\BSP$ which given a set $\ESS$ that contains 
all essential nodes (for an implicitly given pCFG)
returns a slicing pair $(Q,Q_0)$ such that $Q \subseteq Q'$
for any other slicing pair $(Q',\_)$.
The idea is to build $Q$ incrementally, with $Q$ initially containing only $\finalv$; each iteration will 
process the nodes in $\ESS$ that are not already in $Q$, 
and add them to $Q$ (via $F$) if they cannot 
be placed in $Q_0$ without causing $Q$ and $Q_0$ to overlap.

\begin{figure}
\algname{$\BSP$}{$\ESS$}
\begin{algtab*}
\algbegin
  $W \leftarrow \ESS$ \\
  \algforeach{$v \in W \cup \mkset{\finalv}$}
     $Q_v \leftarrow$  \algcall{$\LWS$}{$\mkset{v}$} \\
  \algend
  $Q \leftarrow \emptyset$ \\ 
  $F \leftarrow Q_{\finalv}$ \\
  \algwhile{$F \neq \emptyset$}
    {\bf Invariants:} \\
      \hspace*{5mm} $Q$ and $F$ are both weak slice sets, with $\finalv \in Q \cup F$
          \\
      \hspace*{5mm} $W \subseteq \ESS$ and if $v \in W$ then $Q_v \cap Q = \emptyset$
       \\
      \hspace*{5mm} if $v \in \ESS$ but $v \notin W$ then $v \in Q \cup F$
      \\
     \hspace*{5mm} if $(Q',Q'_0)$ is a slicing pair wrt.~$\ESS$
then $Q \cup F \subseteq Q'$ \\
    $Q \leftarrow Q \cup F$ \\ 
    $F \leftarrow \emptyset$ \\
    \algforeach{$v \in W$}
      \algif{$Q_v \cap Q \neq \emptyset$}
         $W \leftarrow W \setminus \mkset{v}$ \\ 
         $F \leftarrow F \cup Q_v$; \\
    \algend
    \algend
    \algend
 $Q_0 \leftarrow  \bigcup_{v \in W} Q_v$ \\
 \algreturn{$(Q,Q_0)$}  
\end{algtab*}

\caption{\label{fig:best-slicing-pair} 
Finding the best slicing pair ($\BSP$).}
\end{figure}

\setlength{\parindent}{0.15in}

\begin{example}
%\begin{exmp}
We shall continue Examples~\ref{ex:PN1} and \ref{ex:LWS1}
(which consider the program $P_1$ from Example~\ref{ex1}).
Here 3 is the only essential node so we may assume that
$\ESS = \{3\}$;
$\BSP$ thus needs to run
$\LWS$ on $\{4\}$ and on $\{3\}$ 
and from Example~\ref{ex:LWS1} we see
that we get
$Q_4 = \{1,4\}$ and $Q_3 = \{2,3\}$.
When the members of $W = \{3\}$ are first examined in the $\BSP$ algorithm,
we have $Q = Q_4$ and thus $Q_3 \cap Q = \emptyset$.
Hence the while loop terminates after one iteration, with $Q = \{1,4\}$,
and subsequently we get $Q_0 = Q_3 = \{2,3\}$.
%\end{exmp}
\end{example}

\begin{example}
%\begin{exmp}
We shall continue Examples~\ref{ex:PN2} and \ref{ex:LWS2}
(which consider the program $P_4$ from Example~\ref{ex4},
with pCFG depicted in Figure~\ref{fig:ex34r}(right)).
We know from Example~\ref{ex:cycle-induce} that
node 4 is cycle-inducing but node 3 is not;
in Example~\ref{ex:sem-ex} we showed that
node 4 is essential when $\labv{5}$ is an assignment $\assignv{y}{1}$
(as then $\hpd{\omega}{4}{6}$ is not sum-preserving)
and that node 4 is not essential
when $\labv{5}$ is an assignment  $\assignv{y}{y+1}$ 
or a random assignment $\rassignv{y}{\psi_4}$
(as then $\hpd{\omega}{4}{6}$ is sum-preserving).

There are thus two natural possibilities for $\ESS$: the set $\{4\}$,
and the empty set; we shall consider both:
\begin{itemize}
\item
First assume that $\ESS = \emptyset$.
$\BSP$ thus needs to run
$\LWS$ on only $\{6\}$,
and from Example~\ref{ex:LWS2} we see
that we get $Q_6 = \{1,6\}$.
As $W = \emptyset$,
the while loop terminates after one iteration with
$Q = Q_6 = \{1,6\}$,  
and subsequently we get $Q_0 = \emptyset$.
\item
Next assume that $\ESS = \{4\}$.
$\BSP$ thus needs to run 
$\LWS$ on $\{4\}$ and $\{6\}$,
and from Example~\ref{ex:LWS2} we see
that we get $Q_4 = \{1,2,3,4,5\}$ and $Q_6 = \{1,6\}$.

When the members of $W = \{4\}$ are first examined in the $\BSP$ algorithm,
we have $Q = Q_6$ and thus $Q_4 \cap Q = \{1\} \neq \emptyset$.
Hence $W$ will become empty, and eventually the loop will terminate with
$Q = Q_6 \cup Q_4 = \{1,2,3,4,5,6\}$
(and we also get $Q_0 = \emptyset$).
\end{itemize}
%\end{exmp}
\end{example}
That $\BSP$ produces 
the \emph{best slicing pair} is captured by
the following result:
%% (proved in Appendix~\ref{app:proofs}):
\begin{thm}
\label{thm:BSP-correct}
The algorithm $\BSP$ returns,
given a pCFG and a set of nodes $\ESS$,
sets $Q$ and $Q_0$
such that
\begin{itemize}
\item
$(Q,Q_0)$ is a slicing pair wrt.~$\ESS$
\item
if $(Q',Q'_0)$ is a slicing pair wrt.~$\ESS$
then $Q \subseteq Q'$.
\end{itemize}
Moreover, $\BSP$ runs in time $O(n^3)$
(with $n$ the number of nodes in the pCFG).
\end{thm}
%% The algorithm given in~\cite{Danicic+etal:TCS-2011}
%% for computing (their version of) weak slices runs in cubic time.
We do not expect that there exists an algorithm wih
lower asymptotic complexity,
since we need to compute data dependencies which is known to involve
computing a transitive closure.

\section{Improving Precision}
\label{sec:future}
Section~\ref{sec:alg-least} presented an 
algorithm for computing the least slice 
satisfying Definition~\ref{defn:slicing-pair-ess};
such a slice will also 
satisfy Definition~\ref{defn:slicing-pair}
and hence be semantically correct
(as phrased in Theorem~\ref{thm:slicing-correct}).
Still, a smaller semantically correct slice may exist;
in this section we 
briefly discuss two approaches for finding such slices:
semantic analysis of the pCFG, and syntactic transformation of the pCFG.
(Obviously, it is undecidable to always find
the smallest semantically correct slice.)

\subsection{Improvement by Semantic Analysis}
Already in Section~\ref{sec:alg-least} we discussed
how a precise (termination) analysis may help us
to construct a set $\ESS$ that contains fewer (if any)
non-essential nodes which in turn may enable us to slice away some loops.

The size of the slice may also be reduced
if a semantic analysis can determine
that a boolean expression always evaluates to $\true$.
This is illustrated by the pCFGs in Figure~\ref{fig:redundant},
as we shall now discuss.

First consider the pCFG on the left.
As $y = 7$ holds at node 4, the $\observevv$ statement can be discarded,
and indeed, the pCFG is semantically equivalent to the pCFG
containing only nodes 1 and 5. Yet it
has no smaller syntactic slice, since
if $(Q,Q_0)$ is a slicing pair, implying 
$5 \in Q$ and thus $1 \in Q$, then
$Q = \{1,2,3,4,5\}$ as we now show.
If $4 \in Q_0$ then $3 \in Q_0$ (as $Q_0$ provides next visibles)
and thus $1 \in Q_0$ (by data dependence) 
which contradicts $Q \cap Q_0 = \emptyset$. 
As $4$ (as it is essential)
must belong to $Q \cup Q_0$, 
we see that $4 \in Q$; but then $2 \in Q$
(by data dependence) and $3 \in Q$ (as $Q$ provides next visibles).

Next consider the (generic) pCFG on the right,
where $B_1$ and $B_2$ are expressions involving $y$.
There exists no smaller syntactic slice,
since if $(Q,Q_0)$ is a slicing pair and thus $6 \in Q$
then (by data dependence) $4,5 \in Q$ and thus
(as $Q$ provides next visibles) $3 \in Q$ and thus
(by data dependence) $1 \in Q$; also $2 \in Q$ as otherwise
$2 \in Q_0$ and thus (by data dependence) $1 \in Q_0$ which
contradicts $Q \cap Q_0 = \emptyset$.
Still, if say $B_2$ is a logical consequence of $B_1$,
then it is semantically
sound to slice away nodes $3$ and $5$.
Thus, even though $(Q,Q_0) = (\{1,2,4,6\},\emptyset)$
is not a slicing pair according to Definition~\ref{defn:slicing-pair}
as 3 has no next visible in $\{1,2,4,6\}$,
it may be considered a ``semantically valid slicing pair''.

\begin{figure}
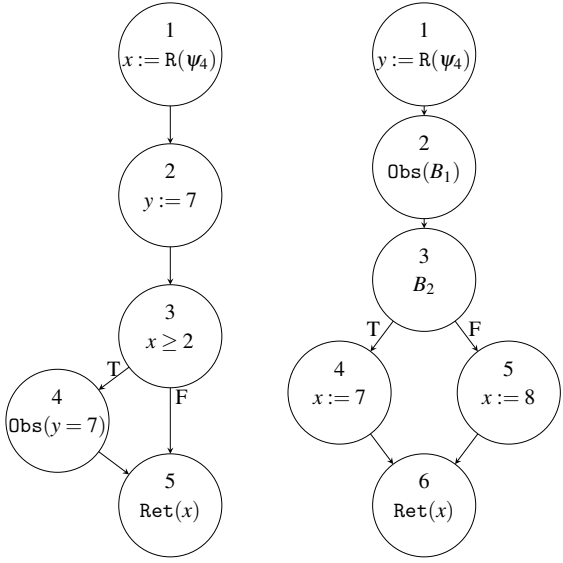


\begin{pgfpicture}{0mm}{0mm}{130mm}{75mm}
\begin{pgfmagnify}{0.75}{0.75}

\pgfnodecircle{v1}[stroke]{\pgfxy(3,9)}{26pt}

\pgfputat{\pgfxy(3,9.3)}{\pgfbox[center,base]{1}}

\pgfputat{\pgfxy(3,8.8)}{\pgfbox[center,base]{$\rassignvs{x}{\psi_4}$}}

\pgfnodecircle{v2}[stroke]{\pgfxy(3,6.5)}{26pt}

\pgfputat{\pgfxy(3,6.8)}{\pgfbox[center,base]{2}}

\pgfputat{\pgfxy(3,6.3)}{\pgfbox[center,base]{$\assignv{y}{7}$}}

\pgfnodecircle{v3}[stroke]{\pgfxy(3,4)}{26pt}

\pgfputat{\pgfxy(3,4.3)}{\pgfbox[center,base]{3}}

\pgfputat{\pgfxy(3,3.8)}{\pgfbox[center,base]{$x \geq 2$}}

\pgfnodecircle{v4}[stroke]{\pgfxy(1,2.5)}{26pt}

\pgfputat{\pgfxy(1,2.8)}{\pgfbox[center,base]{4}}

\pgfputat{\pgfxy(1,2.3)}{\pgfbox[center,base]{$\observevs{y = 7}$}}

\pgfnodecircle{v5}[stroke]{\pgfxy(3,1)}{26pt}

\pgfputat{\pgfxy(3,1.3)}{\pgfbox[center,base]{5}}

\pgfputat{\pgfxy(3,0.8)}{\pgfbox[center,base]{$\retvs{x}$}}

\pgfsetarrowsend{stealth}
%\pgfsetendarrow{\pgfarrowtriangle{4pt}}

\pgfnodeconnline{v1}{v2}

\pgfnodeconnline{v2}{v3}

\pgfputat{\pgfxy(2.0,3.3)}{\pgfbox[center,base]{T}}

\pgfputat{\pgfxy(3.2,2.8)}{\pgfbox[center,base]{F}}

\pgfnodeconnline{v3}{v5}

\pgfnodeconnline{v3}{v4}

\pgfnodeconnline{v4}{v5}

\pgfnodecircle{v1}[stroke]{\pgfxy(7.5,9)}{26pt}

\pgfputat{\pgfxy(7.5,9.3)}{\pgfbox[center,base]{1}}

\pgfputat{\pgfxy(7.5,8.8)}{\pgfbox[center,base]{$\rassignvs{y}{\psi_4}$}}

\pgfnodecircle{v2}[stroke]{\pgfxy(7.5,7)}{26pt}

\pgfputat{\pgfxy(7.5,7.3)}{\pgfbox[center,base]{2}}

\pgfputat{\pgfxy(7.5,6.8)}{\pgfbox[center,base]{$\observevs{B_1}$}}

\pgfnodecircle{v3}[stroke]{\pgfxy(7.5,5)}{26pt}

\pgfputat{\pgfxy(7.5,5.3)}{\pgfbox[center,base]{3}}

\pgfputat{\pgfxy(7.5,4.8)}{\pgfbox[center,base]{$B_2$}}

\pgfnodecircle{v4}[stroke]{\pgfxy(6,3)}{26pt}

\pgfputat{\pgfxy(6,3.3)}{\pgfbox[center,base]{4}}

\pgfputat{\pgfxy(6,2.8)}{\pgfbox[center,base]{$\assignv{x}{7}$}}

\pgfnodecircle{v5}[stroke]{\pgfxy(9,3)}{26pt}

\pgfputat{\pgfxy(9,3.3)}{\pgfbox[center,base]{5}}

\pgfputat{\pgfxy(9,2.8)}{\pgfbox[center,base]{$\assignv{x}{8}$}}

\pgfnodecircle{v6}[stroke]{\pgfxy(7.5,1)}{26pt}

\pgfputat{\pgfxy(7.5,1.3)}{\pgfbox[center,base]{6}}

\pgfputat{\pgfxy(7.5,0.8)}{\pgfbox[center,base]{$\retvs{x}$}}

\pgfsetarrowsend{stealth}
%\pgfsetendarrow{\pgfarrowtriangle{4pt}}

\pgfnodeconnline{v1}{v2}

\pgfnodeconnline{v2}{v3}

\pgfputat{\pgfxy(6.6,4)}{\pgfbox[center,base]{T}}

\pgfputat{\pgfxy(8.4,4)}{\pgfbox[center,base]{F}}

\pgfnodeconnline{v3}{v4}

\pgfnodeconnline{v3}{v5}

\pgfnodeconnline{v4}{v6}

\pgfnodeconnline{v5}{v6}

\end{pgfmagnify}
\end{pgfpicture}

\caption{\label{fig:redundant} 
A redundant $\observevv$ node (left) and a potentially
redundant branch (right).}
\end{figure}

\subsection{Improvement by Syntactic Transformation}
Simple analyses like constant propagation may
improve the precision of slicing even in a deterministic
setting, but the probabilistic setting gives an extra opportunity:
after an $\observev{B}$ node, we know that $B$ holds.
As richly exploited in \cite{Hur+etal:PLDI-2014}, 
a simple syntactic transformation often suffices to get the benefits
of that information, as we illustrate on the program
from \cite[Figure~4]{Hur+etal:PLDI-2014}
whose pCFG (in slightly modified form) is depicted in 
Figure~\ref{fig:Hur4}.
In our setting, if $(Q,Q_0)$ with $18 \in Q$ is the best slicing pair,
then $Q$ will contain everything
except nodes 12, 13, 14, as can be seen as follows:
$16,17 \in Q$ by data dependence;
$15 \in Q$ as $Q$ provides next visibles;
$6,7,8,9 \in Q$ by data dependence;
$3,4,5 \in Q$ as $Q$ provides next visibles;
$1,2 \in Q$ by data dependence;
also $10 \in Q$ as otherwise $10 \in Q_0$ and thus also $9 \in Q_0$
which contradicts $Q \cap Q_0 = \emptyset$.

Alternatively, suppose we insert a node 11 
labeled $\assignv{g}{0}$ between nodes 10 and 12.
This clearly preserves the semantics, but allows a much smaller slice:
choose $Q = \{11,15,16,17,18\}$ and 
$Q_0 = \{1,2,3,4,5,6,7,8,9,10\}$.
This is much like what is arrived at (through a more complex process) in 
\cite[Figure~15]{Hur+etal:PLDI-2014}.

Future work involves 
exploring a larger range of examples, and
(while somewhat orthogonal to the current work)
investigating useful techniques for computing slices
that are smaller than the least syntactic slice yet semantically correct.

\begin{figure}
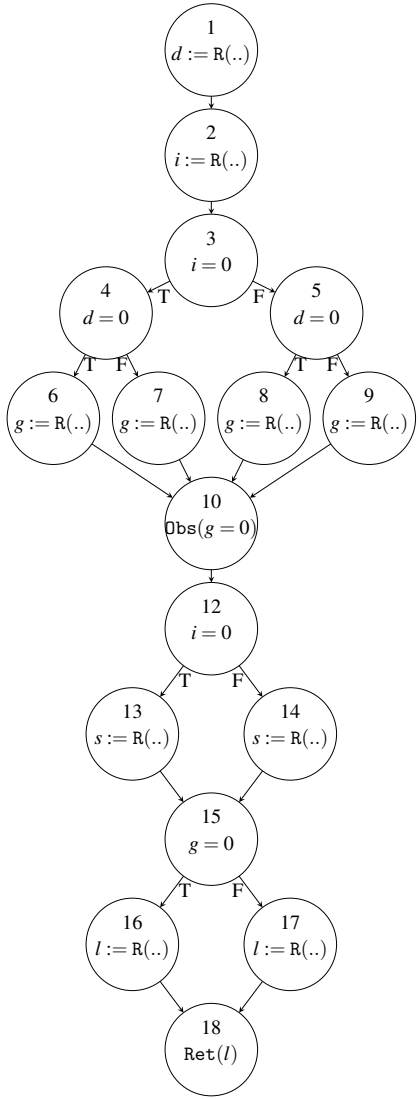

   
\begin{pgfpicture}{0mm}{0mm}{80mm}{150mm}
\begin{pgfmagnify}{0.7}{0.7}

\pgfnodecircle{v1}[stroke]{\pgfxy(5,20)}{25pt}

\pgfputat{\pgfxy(5,20.3)}{\pgfbox[center,base]{1}}
\pgfputat{\pgfxy(5,19.8)}{\pgfbox[center,base]{$\rassignvs{d}{..}$}}

\pgfnodecircle{v2}[stroke]{\pgfxy(5,18)}{25pt}

\pgfputat{\pgfxy(5,18.3)}{\pgfbox[center,base]{2}}
\pgfputat{\pgfxy(5,17.8)}{\pgfbox[center,base]{$\rassignvs{i}{..}$}}

\pgfnodecircle{v3}[stroke]{\pgfxy(5,16)}{25pt}

\pgfputat{\pgfxy(5,16.3)}{\pgfbox[center,base]{3}}
\pgfputat{\pgfxy(5,15.8)}{\pgfbox[center,base]{$i = 0$}}

\pgfputat{\pgfxy(4.1,15.2)}{\pgfbox[center,base]{T}}
\pgfputat{\pgfxy(5.9,15.2)}{\pgfbox[center,base]{F}}

\pgfnodecircle{v4}[stroke]{\pgfxy(3,15)}{25pt}

\pgfputat{\pgfxy(3,15.3)}{\pgfbox[center,base]{4}}
\pgfputat{\pgfxy(3,14.8)}{\pgfbox[center,base]{$d = 0$}}

\pgfputat{\pgfxy(2.7,13.9)}{\pgfbox[center,base]{T}}
\pgfputat{\pgfxy(3.3,13.9)}{\pgfbox[center,base]{F}}

\pgfnodecircle{v5}[stroke]{\pgfxy(7,15)}{25pt}

\pgfputat{\pgfxy(7,15.3)}{\pgfbox[center,base]{5}}
\pgfputat{\pgfxy(7,14.8)}{\pgfbox[center,base]{$d = 0$}}

\pgfputat{\pgfxy(6.7,13.9)}{\pgfbox[center,base]{T}}
\pgfputat{\pgfxy(7.3,13.9)}{\pgfbox[center,base]{F}}

\pgfnodecircle{v6}[stroke]{\pgfxy(2,13)}{25pt}

\pgfputat{\pgfxy(2,13.3)}{\pgfbox[center,base]{6}}
\pgfputat{\pgfxy(2,12.8)}{\pgfbox[center,base]{$\rassignvs{g}{..}$}}

\pgfnodecircle{v7}[stroke]{\pgfxy(4,13)}{25pt}

\pgfputat{\pgfxy(4,13.3)}{\pgfbox[center,base]{7}}
\pgfputat{\pgfxy(4,12.8)}{\pgfbox[center,base]{$\rassignvs{g}{..}$}}

\pgfnodecircle{v8}[stroke]{\pgfxy(6,13)}{25pt}

\pgfputat{\pgfxy(6,13.3)}{\pgfbox[center,base]{8}}
\pgfputat{\pgfxy(6,12.8)}{\pgfbox[center,base]{$\rassignvs{g}{..}$}}

\pgfnodecircle{v9}[stroke]{\pgfxy(8,13)}{25pt}

\pgfputat{\pgfxy(8,13.3)}{\pgfbox[center,base]{9}}
\pgfputat{\pgfxy(8,12.8)}{\pgfbox[center,base]{$\rassignvs{g}{..}$}}

\pgfnodecircle{v10}[stroke]{\pgfxy(5,11)}{25pt}

\pgfputat{\pgfxy(5,11.3)}{\pgfbox[center,base]{10}}
\pgfputat{\pgfxy(5,10.8)}{\pgfbox[center,base]{$\observevs{g = 0}$}}

\pgfnodecircle{v12}[stroke]{\pgfxy(5,9)}{25pt}

\pgfputat{\pgfxy(5,9.3)}{\pgfbox[center,base]{12}}
\pgfputat{\pgfxy(5,8.8)}{\pgfbox[center,base]{$i = 0$}}

\pgfputat{\pgfxy(4.5,7.9)}{\pgfbox[center,base]{T}}
\pgfputat{\pgfxy(5.5,7.9)}{\pgfbox[center,base]{F}}

\pgfnodecircle{v13}[stroke]{\pgfxy(3.5,7)}{25pt}

\pgfputat{\pgfxy(3.5,7.3)}{\pgfbox[center,base]{13}}
\pgfputat{\pgfxy(3.5,6.8)}{\pgfbox[center,base]{$\rassignvs{s}{..}$}}

\pgfnodecircle{v14}[stroke]{\pgfxy(6.5,7)}{25pt}

\pgfputat{\pgfxy(6.5,7.3)}{\pgfbox[center,base]{14}}
\pgfputat{\pgfxy(6.5,6.8)}{\pgfbox[center,base]{$\rassignvs{s}{..}$}}

\pgfnodecircle{v15}[stroke]{\pgfxy(5,5)}{25pt}

\pgfputat{\pgfxy(5,5.3)}{\pgfbox[center,base]{15}}
\pgfputat{\pgfxy(5,4.8)}{\pgfbox[center,base]{$g = 0$}}

\pgfputat{\pgfxy(4.5,3.9)}{\pgfbox[center,base]{T}}
\pgfputat{\pgfxy(5.5,3.9)}{\pgfbox[center,base]{F}}

\pgfnodecircle{v16}[stroke]{\pgfxy(3.5,3)}{25pt}

\pgfputat{\pgfxy(3.5,3.3)}{\pgfbox[center,base]{16}}
\pgfputat{\pgfxy(3.5,2.8)}{\pgfbox[center,base]{$\rassignvs{l}{..}$}}

\pgfnodecircle{v17}[stroke]{\pgfxy(6.5,3)}{25pt}

\pgfputat{\pgfxy(6.5,3.3)}{\pgfbox[center,base]{17}}
\pgfputat{\pgfxy(6.5,2.8)}{\pgfbox[center,base]{$\rassignvs{l}{..}$}}

\pgfnodecircle{v18}[stroke]{\pgfxy(5,1)}{25pt}

\pgfputat{\pgfxy(5,1.3)}{\pgfbox[center,base]{18}}
\pgfputat{\pgfxy(5,0.8)}{\pgfbox[center,base]{$\retvs{l}$}}

\pgfsetarrowsend{stealth}
%\pgfsetendarrow{\pgfarrowtriangle{4pt}}

\pgfnodeconnline{v1}{v2}

\pgfnodeconnline{v2}{v3}

\pgfnodeconnline{v3}{v4}

\pgfnodeconnline{v3}{v5}

\pgfnodeconnline{v4}{v6}

\pgfnodeconnline{v4}{v7}

\pgfnodeconnline{v5}{v8}

\pgfnodeconnline{v5}{v9}

\pgfnodeconnline{v6}{v10}

\pgfnodeconnline{v7}{v10}

\pgfnodeconnline{v8}{v10}

\pgfnodeconnline{v9}{v10}

\pgfnodeconnline{v10}{v12}

\pgfnodeconnline{v12}{v13}

\pgfnodeconnline{v12}{v14}

\pgfnodeconnline{v13}{v15}

\pgfnodeconnline{v14}{v15}

\pgfnodeconnline{v15}{v16}

\pgfnodeconnline{v15}{v17}

\pgfnodeconnline{v16}{v18}

\pgfnodeconnline{v17}{v18}

\end{pgfmagnify}
\end{pgfpicture}
\caption{\label{fig:Hur4} The program from Figure 4 of Hur et al. (modified).} 
\end{figure}

%\cite{Hur+etal:PLDI-2014}(modified).}

\section{Conclusion and Related Work}
\label{sec:conclude}
We have developed a theory for the slicing of probabilistic imperative
programs. 
We have used and extended techniques from the literature 
\cite{Pod+Cla:TSE-1990,Bal+Hor:AAD-1993,Ran+Amt+Ban+Dwy+Hat:TOPLAS-2007,Amtoft:IPL-2007} on the slicing of 
deterministic imperative programs.
These frameworks, some of which have been partly verified
by mechanical proof assistants 
\cite{Wasserrab:PhD-2010,Blazy+etal:CPP-2015},
were recently coalesced by Danicic~\etal~\cite{Danicic+etal:TCS-2011} who 
provide solid 
semantic foundations for the slicing of a large class of deterministic programs.
Our extension of that work is non-trivial in that we need to capture 
probabilistic independence between two sets of variables,
as done in Proposition~\ref{prop:indpd2indpd},
which requires \emph{two} slice sets rather than one. The technical
foundations of our work rest on a novel semantics of pCFGs. 
%% probabilistic control-flow graphs.
In a companion article~\cite{Amt+Ban:ProbSemantics-2017} 
we establish an adequacy result that shows that for pCFGs 
%% control-flow graphs
that are translations of programs in a structured probabilistic language, 
our semantics is suitably
related to that language's denotational semantics
as formulated first by Kozen~\cite{Kozen:JCSS-81} and later
augmented by Gordon~\etal~\cite{Gor+etal:ICSE-2014} 
(in particular to handle conditioning).

We were directly inspired by Hur~\etal~\cite{Hur+etal:PLDI-2014} 
who point out the challenges involved in the slicing of probabilistic
programs, and present an algorithm which constructs
a semantically correct slice. The paper does not state whether
it is in some sense the least possible slice; neither does it
address the complexity of the algorithm.
While Hur \etal's approach differs from ours, for example
it is for a structured language and uses the
denotational semantics presented by Gordon~\etal~\cite{Gor+etal:ICSE-2014},
it is not surprising that their correctness proof also
has probabilistic independence (termed ``decomposition'') as a key notion.
Our theory separates specification and implementation which we believe
provides for a cleaner approach. But
as mentioned in Section~\ref{sec:future}, they incorporate
powerful optimizations that we do not (yet) allow.

Future work includes investigating how our techniques can be used
to analyze which sets of variables in a given probabilistic program
are probabilistically independent of each other 
(a topic explored in, for example, \cite{Bouissou+etal:TACAS-2016}).
%% Sriram wrote us in September 2016 about \cite{Bouissou+etal:TACAS-2016} 
%% {\em Regarding tracking which variables are dependent and independent, we try to do this dynamically in our TACAS 2016 paper.
%% I think that if we can ``factor'' the code fragment into independent units, computations in in the TACAS paper become much faster. Of course, I have stayed away from observe statements since the act of conditioning causes the resulting distributions to become way more complex than is possible for my approaches to handle currently.}}

\textit{Acknowledgements.} We much appreciate the feedback we 
have received on earlier versions; 
in addition to anonymous conference reviewers, 
we would in particular like to thank Gordon Stewart.
% who already has a COI; we could also mention Sasa Misailovic and Daniel Ritchie and Sriram Sankaranarayanan but may want to have them as reviewers

\bibliographystyle{abbrv} 
\bibliography{bib1}

\appendix

\section{Domain Theory}
\label{app:domain}
This section summarizes key aspects of domain theory,
as presented in, \eg, \cite{Schmidt:DenSemantics,Winskel:semantics}.

A domain is a set $D$ equipped with a partial order 
$\sqsubseteq$, that is $\sqsubseteq$ is reflexive, transitive, and
anti-symmetric.
A chain $\chain{x_k}{k}$ is a mapping from the natural numbers into $D$
such that if $i < j$ then $x_i \sqsubseteq x_j$.
We say that $D$ is a \emph{cpo} if each chain $\chain{x_k}{k}$ has a 
\emph{least upper bound} (also called \emph{limit}), 
that is $x \in D$ such that
$x_k \sqsubseteq x$ for all $k$ and such that if also
$x_k \sqsubseteq y$ for all $k$ then $x \sqsubseteq y$;
we shall often write $\limit{k}{x_k}$ for that least upper bound.
We say that a cpo is a \emph{pointed} cpo if there exists a least
element, that is an element $\bot$ such that $\bot \sqsubseteq x$ for all $x \in D$.

We say that a domain $D$ is \emph{discrete} if $x \sqsubseteq y$ implies $x = y$;
a discrete domain is trivially a cpo (but not a pointed cpo unless a singleton).

A function $f$ from a cpo $D_1$ to a cpo $D_2$ is \emph{continuous}
if for each chain $\chain{x_k}{k}$ in $D_1$ the following holds:
$\chain{f(x_k)}{k}$ is a chain in $D_2$,
and $\limit{k}{f(x_k)} = f(\limit{k}{x_k})$.
We let $D_1 \contarrow D_2$ denote the set of 
continuous functions from $D_1$ to $D_2$.
A continuous function $f$ is also monotone, that is $f(x_1) \sqsubseteq f(x_2)$
when $x_1 \sqsubseteq x_2$
(for then $x_1,x_2,x_2,x_2....$ is a chain
and by continuity thus
$f(x_2)$ is the least upper bound of $f(x_1),f(x_2)$
implying $f(x_1) \sqsubseteq f(x_2)$).

\begin{lemma}
\label{lem:cont-cpo}
Let $D_1$ and $D_2$ be cpos. Then $D_1 \contarrow D_2$
is a cpo, with ordering defined pointwise: $f_1 \sqsubseteq f_2$
iff $f_1(x) \sqsubseteq f_2(x)$ for all $x \in D_1$.

If $D_2$ is a pointed cpo then also $D_1 \contarrow D_2$ is a pointed cpo.

If $D_1$ is discrete then $D_1 \contarrow D_2$ contains
all functions from $D_1$ to $D_2$
(and thus we may just write $D_1 \rightarrow D_2$).
\end{lemma}
\begin{proof}
Let $\chain{f_k}{k}$ be a chain of continuous functions from $D_1$ to $D_2$,
with $f$ their pointwise limit, that is:
$f(x) = \limit{k}{f_k(x)}$ for all $x \in D_1$.
We have to show that $f$ is continuous.
But if $\chain{x_k}{k}$ is a chain in $D_1$ then
\begin{eqnarray*}
f(\limit{k}{x_k}) & = & \limit{m}{f_m(\limit{k}{x_k})}
\\ & = & \limit{m}{\limit{k}{f_m(x_k)}}
\\ & = & \limit{k}{\limit{m}{f_m(x_k)}}
\\ & = & \limit{k}{f(x_k)}.
\end{eqnarray*}
If $D_2$ has a bottom element $\bot$ then
$\lambda x.\bot$ is the bottom element in
$D_1 \contarrow D_2$, and if $D_1$ is discrete 
then all functions from $D_1$ to $D_2$ are continuous since
a chain in $D_1$ can contain only one element.
\end{proof}

\begin{lemma}
\label{lem:cont-cpo-fix}
Let $f$ be a continuous function on
a pointed cpo $D$. Then\footnote{Recall that $f^k$ is defined
by letting $f^0(x) = x$, and $f^{k+1}(x) = f(f^{k}(x))$ for $k \geq 0$.}
$\chain{f^k(\bot)}{k}$ is a chain,
and $\limit{k}{f^k(\bot)}$ is the least fixed point of $f$.
\end{lemma}
\begin{proof}
From $\bot \sqsubseteq f(\bot)$ we by monotonicity of $f$ infer that
$f^k(\bot) \sqsubseteq f^{k+1}(\bot)$ for all $k$
so $\chain{f^k(\bot)}{k}$ is indeed a chain.
With $y = \limit{k}{f^k(\bot)}$ we see by continuity of $f$ 
that $y$ is indeed a fixed point of $f$:
$f(y)= \limit{k}{f^{k+1}(\bot)} = y$.
And if $z$ is also a fixed point, we have
$\bot \sqsubseteq z$ and by monotonicity of $f$ thus
$f^k(\bot) \sqsubseteq f^{k}(z) = z$ for all $k$, from which
we infer $y \sqsubseteq z$.
\end{proof}

%% For functional composition, with $f$ a function from $D_1$ to $D_2$
%% and $g$ a function from $D_2$ to $D_3$, we shall often write
%% $f ; g$ for the function $h$ such that $h(x) = g(f(x))$ for $x \in D_1$.
%% Function composition is continuous:
%% \begin{lemma}
%% \label{lem:fun-comp-cont}
%% Let $D$, $D_1$ and $D_2$ be cpos,
%% let $\chain{f_k}{k}$ be a chain in $D \contarrow D_1$,
%% and let $\chain{g_k}{k}$ be a chain in $D_1 \contarrow D_2$.
%% With $f = \limit{k}{f_k}$ and $g = \limit{k}{g_k}$,
%% we have $f ; g = \limit{k}{(f_k; g_k)}$.
%% \end{lemma}
%% \begin{proof}
%% \begin{eqnarray*}
%% (f ; g)(x) & = & g(f(x)) = \limit{k}{g(f_k(x))} = 
%% \limit{k}{\limit{j}{g_j(f_k(x))}}
%% \\ & = & \limit{k}{g_k(f_k(x))}
%% = (\limit{k}{(f_k ; g_k)})(x).
%% \end{eqnarray*}
%% \end{proof}

\section{Miscellaneous Proofs}
\label{app:proofs}
\subsection{Proofs for Section~\ref{sec:CFG}}

{\bf Lemma~\ref{lem:prec-ordering}}:
For given $v$, let $\prec$ be an ordering among
proper postdominators of $v$,
by stipulating that $v_1 \prec v_2$ iff in all acyclic paths from $v$
to $\finalv$, $v_1$ occurs strictly before $v_2$.
Then $\prec$ is transitive, antisymmetric, and total.
Also, if $v_1 \prec v_2$ then 
for \emph{all} paths from $v$ to $\finalv$ it is the case that
the first occurrence of $v_1$ is before the first occurrence of $v_2$.

\begin{proof}
The first two properties are obvious.

We next show that $\prec$ is total.
Assume, to get a contradiction, that there exists an 
acyclic path $\pi_1$ from $v$ to $\finalv$
that contains $v_1$ strictly before $v_2$,
and also an
acyclic path $\pi_2$ from $v$ to $\finalv$
that contains $v_2$ strictly before $v_1$.
But then the concatenation of the prefix of $\pi_1$ that ends with $v_1$,
and the suffix of $\pi_2$ that starts with $v_1$, is a path
from $v$ to $\finalv$
that avoids $v_2$, yielding a contradiction
as $v_2$ postdominates $v$.

Finally, assume that $v_1 \prec v_2$, and that $\pi$ is a path
from $v$ to $\finalv$; to get a contradiction, assume that
there is a prefix $\pi_1$ of $\pi$ that ends with $v_2$ but does 
not contain $v_1$. Since there exists an acyclic path
from $v$ to $\finalv$, we infer from $v_1 \prec v_2$ that
there is an acyclic path $\pi_2$ from $v_2$ that does not contain $v_1$.
But the concatenation of $\pi_1$ and $\pi_2$ is a path from
$v$ to $\finalv$ that does not contain $v_1$,
which contradicts $v_1$ being a proper postdominator of $v$.
\end{proof}

\noindent
{\bf Lemma~\ref{lem:LAP-add}}:
If $(v,v_1) \in \PD$ and $(v_1,v_2) \in \PD$ (and thus
$(v,v_2) \in \PD$) then $\LAP{v}{v_2} = \LAP{v}{v_1} + \LAP{v_1}{v_2}$.

\begin{proof}
If $v = v_1$ or $v_1 = v_2$, the claim is obvious;
we can thus assume that $v_1$ and $v_2$ are proper postdominators of $v$
and by Lemma~\ref{lem:prec-ordering} we further infer that $v_1$
will occur before $v_2$ in all paths from $v$ to $\finalv$.

First consider an acyclic path $\pi$ from $v$ to $v_2$. 
We have argued that $\pi$ will contain $v_1$,
and hence $\pi$ is the concatenation of an acyclic path from
$v$ to $v_1$, thus of length $\leq \LAP{v}{v_1}$,
and an acyclic path from $v_1$ to $v_2$, thus of length $\leq \LAP{v_1}{v_2}$.
Thus the length of $\pi$ is $\leq \LAP{v}{v_1} + \LAP{v_1}{v_2}$;
as $\pi$ was an arbitrary acyclic path from $v$ to $v_2$, 
this shows ``$\leq$''. 

To show ``$\geq$'', let $\pi_1$ be an acyclic path from
$v$ to $v_1$ of length $\LAP{v}{v_1}$, and $\pi_2$ be an acyclic path from
$v_1$ to $v_2$ of length $\LAP{v_1}{v_2}$. 
Let $\pi$ be the concatenation
of $\pi_1$ and $\pi_2$;
$\pi$ is an acyclic path from $v$ to $v_2$
since if $v' \neq v_1$ occurs in both
paths then there is a path from $v$ to $v_2$ that avoids $v_1$
which is a contradiction. 
As $\pi$ is of length 
$\LAP{v}{v_1} + \LAP{v_1}{v_2}$, this shows ``$\geq$''.
\end{proof}

\begin{lemma}
\label{lem:outside-decomp}
Assume that $v'$ is a proper postdominator of $v$, that
with $v'' = \fppd{v}$ we have $v' \neq v''$, 
and that $Q$ is a set of nodes.

If $v$ stays outside $Q$ until $v'$ then
\emph{(i)}
$v$ stays outside $Q$ until $v''$, and \emph{(ii)} $v''$
stays outside $Q$ until $v'$.
\end{lemma}
\begin{proof}
For (i), let $\pi$ be a path from $v$ to $v''$ that contains $v''$ only at the
end. Hence $v'$ cannot be in $\pi$ (as $v''$ occurs in all paths from $v$ to $v'$), 
so we can extend $\pi$ into a path $\pi'$
from $v$ to $v'$
that contains $v'$ only at the end. Since $v$ stays outside $Q$ until $v'$,
$\pi'$ contains no node in $Q$ except possibly $v'$, and hence $\pi$ contains
no node in $Q$.

For (ii), let $\pi$ be a path from $v''$ to $v'$ that contains $v'$ only at the end.
There is a path from $v$ to $v''$ that does not contain $v'$, so we can extend
$\pi$ into a path $\pi'$ from $v$ to $v'$ that contains $v'$ only at the end.
Since $v$ stays outside $Q$ until $v'$,
$\pi'$ contains no node in $Q$ except possibly $v'$, and hence $\pi$ contains
no node in $Q$ except possibly $v'$.
\end{proof}

\begin{lemma}
\label{lem:outside-branch}
Assume that $v'$ is a proper postdominator of $v$, that
$v_1$ is a successor of $v$, and that $Q$ is a set of nodes.

If $v$ stays outside $Q$ until $v'$ then
also $v_1$ stays outside $Q$ until $v'$.
\end{lemma}
\begin{proof}
Let $\pi$ be a path from $v_1$ to $v'$ that contains $v'$ only at the end.
Since $v \neq v'$, we can extend $\pi$ into a path $\pi'$ from $v$ to $v'$
that contains $v'$ only at the end.
Since $v$ stays outside $Q$ until $v'$,
$\pi'$ contains no node in $Q$ except possibly $v'$,
and hence $\pi$ contains no node in $Q$ except possibly $v'$.
\end{proof}

\subsection{Proofs for Section~\ref{sec:sem}}

\noindent
{\bf Lemma~\ref{lem:sumlim-limsum}}:
Assume that $\chain{D_k}{k}$ 
is a chain of distributions (not necessarily bounded)
with $D' = \limit{k}{D_k}$.
With $S$ a (countable) set of stores, we have
\[
\sum_{s \in S}{D'(s)} = \limit{k}{\sum_{s \in S}{D_k(s)}}
\]
\begin{proof}
From $D_k \leq D'$ we get that
$\sum_{s \in S}{D'(s)}$ is an upper bound for 
$\chain{\sum_{s \in S}{D_k(s)}}{k}$; as
$\limit{k}{\sum_{s \in S}{D_k(s)}}$ is the least upper bound,
we get
\[
\limit{k}{\sum_{s \in S}{D_k(s)}} \leq \sum_{s \in S}{D'(s)}.
\]
To establish that equality holds, we shall assume
$\limit{k}{\sum_{s \in S}{D_k(s)}} < \sum_{s \in S}{D'(s)}$
so as to get a contradiction.
Then there exists $\epsilon > 0$ such that
$\limit{k}{\sum_{s \in S}{D_k(s)}} + \epsilon < \sum_{s \in S}{D'(s)}$.
We infer that there exists a finite set $S_0$ with $S_0 \subseteq S$
such that 
$\limit{k}{\sum_{s \in S}{D_k(s)}} + \epsilon < \sum_{s \in S_0} D'(s)$.
For each $s \in S_0$ there exists $K_s$ such that
$D_k(s) > D'(s) - \epsilon/|S_0|$ for $k \geq K_s$,
and thus there exists $K$ (the maximum element of the finite set $\{K_s \mid s \in S_0\}$)
such that for each $s \in S_0$, and each $k \geq K$,
$D_k(s) + \epsilon/|S_0| > D'(s)$.
But then we get the desired contradiction:
\begin{eqnarray*}
\sum_{s \in S_0} D'(s) & < & \sum_{s \in S_0}(D_K(s) + \epsilon/|S_0|)
= \sum_{s \in S_0}D_K(s) + \epsilon
\leq \limit{k}{\sum_{s \in S}{D_k(s)}} + \epsilon
\\ & < & \sum_{s \in S_0} D'(s).
\end{eqnarray*}
\end{proof}

\noindent
{\bf Lemma~\ref{lem:sto-partial}}:
If $R \subseteq R'$ then
for $s \in \sto{R}$ we have
\begin{displaymath}
D(s) = \sum_{s' \in \sto{R'} \ \mid\ \sagree{s'}{s}{R}} D(s'). 
\end{displaymath}
\begin{proof}
We have the calculation 
\begin{eqnarray*}
\sum_{s' \in \sto{R'} \ \mid\ \sagree{s'}{s}{R}} D(s')
& = &
\sum_{s' \in \sto{R'} \ \mid\ \sagree{s'}{s}{R}}
\left(\sum_{s_0 \in \fulls\ \mid\ \sagree{s_0}{s'}{R'}} D(s_0) \right)
\\[2mm] & = & \sum_{s_0 \in \fulls,\ s' \in \sto{R'}\ \mid\ \sagree{s'}{s}{R},\ \sagree{s_0}{s'}{R'}} D(s_0)
 = 
\sum_{s_0 \in \fulls\ \mid\ \sagree{s_0}{s}{R}}D(s_0)
 = D(s) 
\end{eqnarray*}
where the third equality is justified as follows:
for a given $s_0 \in \fulls$, exactly one $s' \in \sto{R'}$ will
satisfy $\sagree{s_0}{s'}{R'}$, and for that $s'$ we will
have (since $R \subseteq R'$)
that $\sagree{s'}{s}{R}$ iff $\sagree{s_0}{s}{R}$.
\end{proof}

\noindent
{\bf Lemma~\ref{lem:sum-pres-conc}}:
Let $f \in \Dist \rightarrow \Dist$
be continuous and additive.
Assume that for all $D$ that are concentrated, $f$ is sum-preserving for $D$.
Then $f$ is sum-preserving.

\begin{proof}
Let $s_1$, $s_2$, \ldots be an enumeration of stores in $\fulls$.
For given $D \in \Dist$, and for each $k \geq 1$,
let $D_k$ be given by stipulating $D_k(s_k) = D(s_k)$
but $D_k(s) = 0$ when $s \neq s_k$. Thus each $D_k$ is concentrated,
and $D = \limit{k}{D'_k}$ where 
$D'_k = D_1 + \ldots + D_k$. 
Since $f$ is assumed continuous and additive, 
\[
f(D) = f(\limit{k}{D'_k}) = \limit{k}{f(D'_k)} =
\limit{k}{(f(D_1) + \ldots + f(D_k))}
\]
and thus the desired result follows from the calculation
(where we use Lemma~\ref{lem:sumlim-limsum} twice,
and exploit that $f$ is sum-preserving for each $D_k$)
\begin{eqnarray*}
\sumd{f(D)} & = & 
\sum_{s \in \fulls} f(D)(s) \\
& = & 
\sum_{s \in \fulls} \limit{k} (f(D_1)(s) + \ldots + f(D_k)(s)) \\
& = & 
\limit{k}{\sum_{s \in \fulls} (f(D_1)(s) + \ldots + f(D_k)(s))} \\
& = &
\limit{k}{(\sum_{s \in \fulls} 
f(D_1)(s) + \ldots + \sum_{s \in \fulls} f(D_k)(s))} \\
& = &
\limit{k}{(\sum_{s \in \fulls} D_1(s) + \ldots + \sum_{s \in \fulls} D_k(s))} \\
& = &
\limit{k}{\sum_{s \in \fulls} D'_k(s)} \\
& = &
\sum_{s \in \fulls}\limit{k}{D'_k(s)} = \sum_{s \in \fulls}D(s)
=
\sumd{D}.
\end{eqnarray*}
\end{proof}

\noindent
{\bf Lemma~\ref{lem:assignRirrel}}:
Assume that $\assignf{x}{E}(D) = D'$
and that $x \notin R$. Then
$\dagree{D}{D'}{R}$.

\begin{proof}
Given $s_0 \in \sto{R}$, we must show that
$D'(s_0) = D(s_0)$.
But this follows since
\begin{eqnarray*}
D'(s_0) & = & \displaystyle 
\sum_{s' \in \fulls\ \mid\ \sagree{s'}{s_0}{R}} D'(s')\ = 
\sum_{s' \in \fulls\ \mid\ \sagree{s'}{s_0}{R}} 
\left( \sum_{s \in \fulls\ \mid\ s' = \upd{s}{x}{\seme{E}s}} D(s)\right) \\[1mm]
& = & \displaystyle
\sum_{s,s' \in \fulls\ \mid\ \sagree{s'}{s_0}{R},\ s'= \upd{s}{x}{\seme{E}s}}D(s) 
\\[1mm] \mbox{(as $x \notin R$)} & = &
\sum_{s,s' \in \fulls\ \mid\ \sagree{s}{s_0}{R},\ s' = \upd{s}{x}{\seme{E}s}}D(s) 
\ =
\sum_{s \in \fulls\ \mid\ \sagree{s}{s_0}{R}} D(s)
=
D(s_0)
\end{eqnarray*}
\end{proof}

\noindent
{\bf Lemma~\ref{lem:assign-cont}}
$\assignf{x}{E}$ is continuous.

\begin{proof}
Obviously, $\assignf{x}{E}$ is monotone.
To show continuity,
let $\chain{D_k}{k}$ be a chain. With
$D'_k = \assignf{x}{E}(D_k)$, monotonicity implies that also
$\chain{D'_k}{k}$ is a chain;
let $D = \limit{k}{D_k}$ and $D' = \limit{k}{D'_k}$.
Our goal is to prove that $D' = \assignf{x}{E}(D)$.
But this follows since by Lemma~\ref{lem:sumlim-limsum}
for each $s'$ we have the calculation
\begin{eqnarray*}
D'(s') & = & \limit{k}{D'_k(s')} =
     \limit{k}{\sum_{s \in \fulls\ \mid\ s' = \upd{s}{x}{\seme{E}s}} D_k(s)} 
\\ & = & \sum_{s \in \fulls\ \mid\ s' = \upd{s}{x}{\seme{E}s}}\limit{k}{D_k(s)}
    = \sum_{s \in \fulls\ \mid\ s' = \upd{s}{x}{\seme{E}s}} D(s)
\end{eqnarray*}
\end{proof}

\noindent
{\bf Lemma~\ref{lem:rassignRirrel}}:
Assume that $\rassignf{x}{E}(D) = D'$
and that $x \notin R$. Then
$\dagree{D}{D'}{R}$.

\begin{proof}
Given $s_0 \in \sto{R}$, we must show that
$D'(s_0) = D(s_0)$.
But this follows since 
\begin{eqnarray*}
D'(s_0) & = & \displaystyle 
\sum_{s' \in \fulls\ \mid\ \sagree{s'}{s_0}{R}} D'(s')\ = 
\sum_{s' \in \fulls\ \mid\ \sagree{s'}{s_0}{R}} 
\left( \sum_{s \in \fulls\ \mid\ \sagree{s'}{s}{\uvar \setminus \mkset{x}}} \psi(s'(x))D(s)  \right) \\[1mm] \mbox{(as $x \notin R$)} 
& = & \displaystyle
\sum_{s,s' \in \fulls\ \mid\ \sagree{s}{s_0}{R},\ \sagree{s'}{s}{\uvar \setminus \mkset{x}}} \psi(s'(x))D(s) 
 = 
\sum_{s \in \fulls\ \mid\ \sagree{s}{s_0}{R}} \left(
\sum_{s' \in \fulls\ \mid\ \sagree{s'}{s}{\uvar \setminus \mkset{x}}}
\psi(s'(x))D(s) \right) \\[1mm] & = &
\sum_{s \in \fulls\ \mid\ \sagree{s}{s_0}{R}} \left( D(s) \left(
\sum_{s' \in \fulls\ \mid\ \sagree{s'}{s}{\uvar \setminus \mkset{x}}}
\psi(s'(x)) \right) \right) \ =
\sum_{s \in \fulls\ \mid\ \sagree{s}{s_0}{R}} \left( D(s) \left(
\sum_{z \in \Zz}
\psi(z) \right) \right) 
\\[1mm] & = &
\sum_{s \in \fulls\ \mid\ \sagree{s}{s_0}{R}} (D(s) \cdot 1)
=
D(s_0)
\end{eqnarray*}
\end{proof}

\noindent
{\bf Lemma~\ref{lem:rassign-cont}}
$\rassignf{x}{E}$ is continuous.

\begin{proof}
Obviously, $\rassignf{x}{E}$ is monotone.
To show continuity,
let $\chain{D_k}{k}$ be a chain. With
$D'_k = \rassignf{x}{E}(D_k)$, monotonicity implies that also
$\chain{D'_k}{k}$ is a chain;
let $D = \limit{k}{D_k}$ and $D' = \limit{k}{D'_k}$.
Our goal is to prove that $D' = \rassignf{x}{E}(D)$.
But this follows since by Lemma~\ref{lem:sumlim-limsum}
for each $s'$ we have the calculation
\begin{eqnarray*}
D'(s') & = & \limit{k}{D'_k(s')} =
     \limit{k}{\sum_{s \in \fulls\ \mid\ \sagree{s'}{s}{\uvar \setminus \mkset{x}}} 
   \psi(s'(x))D_k(s)}
\\ & = & \sum_{s \in \fulls\ \mid\ \sagree{s'}{s}{\uvar \setminus \mkset{x}}}\limit{k}{\psi(s'(x))D_k(s)}
\\ & = &  \sum_{s \in \fulls\ \mid\ \sagree{s'}{s}{\uvar \setminus \mkset{x}}}\psi(s'(x)) D(s)
\end{eqnarray*}
\end{proof}

\noindent
{\bf Lemma~\ref{lem:HHcont}}
The functional $\HH{X}$ is continuous
on $\PD \rightarrow (\Dist \contarrow \Dist)$.

\begin{proof}
Consider a chain $\chain{g_k}{k}$,
so as to prove that
$\HH{X}(\limit{k}{g_k}) = \limit{k}{\HH{X}(g_k)}$. For all $(v,v') \in \PD$
and all $D$ in $\Dist$, we must thus prove
\[
\HH{X}(\limit{k}{g_k})(v,v')(D) = \limit{k}{\HH{X}(g_k)(v,v')(D)}
\]
and shall do so by induction in $\LAP{v}{v'}$, with a case analysis 
in Definition~\ref{def:HHX}. We shall consider some sample cases:
\begin{itemize}
\item
If $v \in X$ with $\labv{v}$ of the form $\assignv{x}{E}$
then both sides evaluate to
$\assignf{x}{E}(D)$.
\item
If $v' \neq v''$ where $v'' = \fppd{v}$
then we have the calculation 
\begin{eqnarray*}
\HH{X}(\limit{k}{g_k})(v,v')(D) 
& = &
\HH{X}(\limit{k}{g_k})(v'',v')(\HH{X}(\limit{k}{g_k})(v,v'')(D))
 \\ & = &
(\limit{k}{\HH{X}(g_k)(v'',v')})(\limit{k}{\HH{X}(g_k)(v,v'')(D)})
 \\ & = &
\limit{k}{\HH{X}(g_k)(v'',v')(\HH{X}(g_k)(v,v'')(D))}
 \\ & = &
\limit{k}{\HH{X}(g_k)(v,v')(D)}
\end{eqnarray*}
where the second equality follows from the induction hypothesis,
and the third equality from continuity of $\HH{X}(g_k)(v'',v')$
(Lemma~\ref{lem:hcontH}).
\item
If $v$ is a branching node with condition $B$,
$\true$-successor $v_1$, and $\false$-successor $v_2$,
where $\LAP{v_1}{v'} \geq \LAP{v}{v'}$ and
$\LAP{v_2}{v'} < \LAP{v}{v'}$ (other cases are similar),
with $D_1 = \selectf{B}(D)$ and $D_2 = \selectf{\neg B}(D)$
we have the calculation (where the second equality follows
from the induction hypothesis):
\begin{eqnarray*}
\HH{X}(\limit{k}{g_k})(v,v')(D) 
 & = &
\limit{k}{g_k}(v_1,v')(D_1) + \HH{X}(\limit{k}{g_k})(v_2,v')(D_2)
 \\ & = &
\limit{k}{(g_k(v_1,v')(D_1) + \HH{X}(g_k)(v_2,v')(D_2))}
 \\ & = &
\limit{k}{\HH{X}(g_k)(v,v')(D)}
\end{eqnarray*}
\end{itemize}
\end{proof}

The following two lemmas are often convenient.
\begin{lemma}
\label{lem:assign-simp}
Assume that $\assignf{x}{E}(D) = D'$.
Assume that $R,R'$ are such that $x \in R'$, and that
$R'' \cup \fv{E} \subseteq R$
where $R'' = R' \setminus \mkset{x}$.
For $s' \in \sto{R'}$ we then have
\begin{displaymath}
D'(s') 
= \sum_{s \in \sto{R} \ \mid\ \sagree{s}{s'}{R''},\ s'(x) = \seme{E}s} D(s).
\end{displaymath}
\end{lemma}
\begin{prooff}
This follows from the calculation
\begin{eqnarray*} 
D'(s') & = &
\sum_{s'_0 \in \fulls \ \mid\ \sagree{s'_0}{s'}{R'}} D'(s'_0) \ 
= 
\sum_{s'_0 \in \fulls \ \mid\ \sagree{s'_0}{s'}{R'}} 
\left( \sum_{s_0 \in \fulls \ \mid\ s'_0 = \upd{s_0}{x}{\seme{E}s_0}} D(s_0) \right)
\\[2mm] & = &
\sum_{s'_0,s_0 \in \fulls \ \mid\ s'_0 = \upd{s_0}{x}{\seme{E}s_0},\ \sagree{s'_0}{s'}{R''},\ s'(x) = \seme{E}s_0} D(s_0)\ 
=
\sum_{s_0 \in \fulls \ \mid\ \sagree{s_0}{s'}{R''},\ s'(x) = \seme{E}s_0} D(s_0)
\\[2mm] & = &
\sum_{s_0 \in \fulls \ \mid\ \sagree{\restrs{s_0}{R}}{s'}{R''},\ s'(x) = \seme{E}(\restrs{s_0}{R})} D(s_0)
=
\sum_{s_0 \in \fulls}\sum_{s \in \sto{R}\ \mid\ s = \restrs{s_0}{R},\ \sagree{s}{s'}{R''},\ s'(x) = \seme{E}s} D(s_0)
\\[2mm] & = &
\sum_{s \in \sto{R} \ \mid\ \sagree{s}{s'}{R''},\ s'(x) = \seme{E}s} 
\left(\sum_{s_0 \in \fulls \ \mid\ \sagree{s_0}{s}{R}} D(s_0)\right) \ 
= 
\sum_{s \in \sto{R} \ \mid\ \sagree{s}{s'}{R''},\ s'(x) = \seme{E}s} D(s)
\end{eqnarray*}
\end{prooff}
\begin{lemma}
\label{lem:rassign-simp}
Assume that $\rassignf{x}{\psi}(D) = D'$.
Assume that $R,R'$ are such that $x \in R'$, and that 
$R'' \subseteq R$
where $R'' = R' \setminus \mkset{x}$.
For $s' \in \sto{R'}$ we then have
\begin{displaymath}
D'(s') 
= 
\psi(s'(x)) 
\sum_{s \in \sto{R} \ \mid\ \sagree{s}{s'}{R''}}
D(s)
\end{displaymath}
\end{lemma}
\begin{prooff}
This follows from the calculation
\begin{eqnarray*} 
D'(s') & = &
\sum_{s'_0 \in \fulls \ \mid\ \sagree{s'_0}{s'}{R'}} D'(s'_0) \ 
= 
\sum_{s'_0 \in \fulls \ \mid\ \sagree{s'_0}{s'}{R'}} 
\left(\sum_{s_0 \in \fulls\ \mid\ \sagree{s'_0}{s_0}{\uvar \setminus \mkset{x}}} 
\psi(s'_0(x))D(s_0) \right)
\\[2mm] & = &
\sum_{s_0,s'_0 \in \fulls\ \mid\ \sagree{s_0}{s'}{R''},\ s'_0 = \upd{s_0}{x}{s'(x)}} \psi(s'(x))D(s_0)
\\[2mm] & = &
\sum_{s_0 \in \fulls\ \mid\ \sagree{s_0}{s'}{R''}} \psi(s'(x))D(s_0)
=
\psi(s'(x)) 
\sum_{s \in \sto{R} \ \mid\ \sagree{s}{s'}{R''}} \ 
\sum_{s_0 \in \fulls\ \mid\ \sagree{s_0}{s}{R}} D(s_0)
\\[2mm] & = &
\psi(s'(x)) 
\sum_{s \in \sto{R} \ \mid\ \sagree{s}{s'}{R''}}
D(s)
\end{eqnarray*}
\end{prooff}

Next some results that prepare for the proof of
Lemma~\ref{lem:fixed-mult-nonincr-determ}.

\begin{lemma}
\label{lem:limit-add}
Let $\chain{f_k}{k}$ be a chain of additive functions.
Then $f = \limit{k}{f_k}$ is additive.
\end{lemma}
\begin{proof}
For all distributions $D_1$ and $D_2$,
%% such that $D_1 + D_2$ is a distribution, 
we have
\begin{eqnarray*}
f(D_1 + D_2) & = & \limit{k}{f_k(D_1 + D_2)} = \limit{k}{f_k(D_1) + f_k(D_2)} 
\\ & = & 
\limit{k}{f_k(D_1)} + \limit{k}{f_k(D_2)} = f(D_1) + f(D_2).
\end{eqnarray*}
\end{proof}

\begin{lemma}
\label{lem:limit-mult}
Let $\chain{f_k}{k}$ be a chain of multiplicative functions.
Then $f = \limit{k}{f_k}$ is multiplicative.
\end{lemma}
\begin{proof}
For all distributions $D$, and for all $c$ with $c \geq 0$, we have
\[
f(cD) = \limit{k}{f_k(cD)} = \limit{k}{c f_k(D)} = c\;\limit{k}{f_k(D)} = cf(D).
\]
\end{proof}

\begin{lemma}
\label{lem:limit-non-incr}
Let $\chain{f_k}{k}$ be a chain of non-increasing functions.
Then $f = \limit{k}{f_k}$ is a non-increasing function.
\end{lemma}
\begin{proof}
For all distributions $D$,
by assumption
$\sumd{f_k(D)} \leq \sumd{D}$ for all $k$ so
by Lemma~\ref{lem:sumlim-limsum} we get:
\[  
\sumd{f(D)} = \sumd{\limit{k}{f_k(D)}} 
= \limit{k}{\sumd{f_k(D)}}
\leq \limit{k}{\sumd{D}} = \sumd{D}.
\]
\end{proof}

%% \begin{lemma}
%% \label{lem:limit-determine}
%% Let $\chain{f_k}{k}$ be a chain of deterministic functions.
%% Then $f = \limit{k}{f_k}$ is deterministic.
%% \end{lemma}
%% \begin{proof}
%% With concentrated $D$ given, 
%% we must prove that $f(D)$ is concentrated.
%% If $f(D) = 0$, the claim is obvious.
%% Otherwise, there exists $m$ and $s_0 \in \fulls$ such 
%% that $f_m(D)(s_0) > 0$. As $\chain{f_k}{k}$ is a chain, we infer that
%% for all $n \geq m$ we have $f_n(D)(s_0) > 0$,
%% which as each $f_n$ is deterministic implies
%% that for all $s \in \fulls$ with $s \neq s_0$ we have
%% $f_n(D)(s) = 0$, and thus the desired $f(D)(s) = 0$.
%% \end{proof}

\begin{lemma}
\label{lem:HH-preserve-misc}
With $\HH{X}$ as defined in Def.~\ref{def:HHX}, we have:
\begin{itemize}
\item
if $\hpd{h_0}{v}{v'}$ is additive for all $(v,v') \in \PD$
then $\hpd{\HH{X}(h_0)}{v}{v'}$ is additive for all $(v,v') \in \PD$;
\item
if $\hpd{h_0}{v}{v'}$ is multiplicative for all $(v,v') \in \PD$
then $\hpd{\HH{X}(h_0)}{v}{v'}$ is multiplicative for all $(v,v') \in \PD$;
\item
if $\hpd{h_0}{v}{v'}$ is non-increasing for all $(v,v') \in \PD$
then $\hpd{\HH{X}(h_0)}{v}{v'}$ is non-increasing for all $(v,v') \in \PD$.
%% \item
%% for a deterministic CFG, 
%% if $\hpd{h_0}{v}{v'}$ is deterministic for all $(v,v') \in \PD$
%% then $\hpd{\HH{X}(h_0)}{v}{v'}$ is deterministic for all $(v,v') \in \PD$.
\end{itemize}
\end{lemma}
\begin{proof}
An easy induction in $\LAP{v}{v'}$,
using Lemmas~\ref{lem:selectB}, \ref{lem:assign-misc} and
\ref{lem:rassign-misc}.
\end{proof}

{\bf Lemma~\ref{lem:fixed-mult-nonincr-determ}}
Let $h = \fixed{\HH{X}}$. Then for each $(v,v') \in \PD$,
$\hpd{h}{v}{v'}$ is additive, multiplicative and non-increasing
(as is also $\hpd{\omega_k}{v}{v'}$ for each $k \geq 0$).
%% Also, if the CFG is deterministic, then $\hpd{h}{v}{v'}$ is deterministic.

\begin{proof}
The function $0$ is obviously additive, multiplicative and non-increasing,
%% and deterministic,
so by Lemma~\ref{lem:HH-preserve-misc} we infer that 
for each $k$, letting $h_k = \HH{X}^k(h_0)$, and for each $(v,v') \in \PD$,
$\hpd{h_k}{v}{v'}$ is additive, multiplicative and non-increasing.
%% and also deterministic if the CFG is. 
The claim now follows from Lemmas~\ref{lem:limit-add}, 
\ref{lem:limit-mult}, and \ref{lem:limit-non-incr}.
%% and~\ref{lem:limit-determine}.
\end{proof}

{\bf Lemma~\ref{lem:phi1-id}}
Given a pCFG, a slice set $Q$, and $(v,v') \in \PD$ such that
$v$ stays outside $Q$ until $v'$. 
We then have
$\hpd{\HH{Q}(h)}{v}{v'}(D) = D$ for all $D \in \Dist$
and all modification functions $h$.

\begin{proof}
We do induction in $\LAP{v}{v'}$. The claim is obvious if $v = v'$.
If with $v'' = \fppd{v}$ we have $v' \neq v''$, we can
(by Lemmas~\ref{lem:outside-decomp} and \ref{lem:LAP-add}) apply the
induction hypothesis to $(v,v'')$ and $(v'',v')$,
to get the desired
$\hpd{\HH{Q}(h)}{v}{v'}(D) = \hpd{\HH{Q}(h)}{v''}{v'}(\hpd{\HH{Q}(h)}{v}{v''}(D)) =
\hpd{\HH{Q}(h)}{v''}{v'}(D) = D$.

We are left with the case when $v' = \fppd{v}$, and since $v$ stays
outside $Q$ until $v'$ we see that $v \notin Q$ and thus
clause~\ref{defn:evalk:notinX} in Definition~\ref{def:HHX}
gives the desired $\hpd{\HH{Q}(h)}{v}{v'}(D) = D$.
\end{proof}

We now prepare for the proof of Lemma~\ref{lem:sum-preserve-not-affect-rv}.

\begin{definition}
A function $f: \Dist \rightarrow \Dist$ is non-increasing wrt.~$R$,
a set of variables, if $f(D)(s) \leq D(s)$ holds
for all $D \in \Dist$ and $s \in \sto{R}$.
\end{definition}
Observe that with $R = \emptyset$, this reduces to the previous notion
of non-increasing. 

\begin{lemma}
\label{lem:outside-nonincrease}
With $(v,v') \in \PD$, assume that
$Q$ is closed under data dependence and that $v$ stays outside $Q$ until $v'$.
Let $R = \rv{Q}{v} = \rv{Q}{v'}$ (well-defined by
Lemma~\ref{lem:outside-same-rv}),
and let $\omega_k = \HH{\unodes}^k(0)$ for each $k \geq 0$. Then
$\hpd{\omega_k}{v}{v'}$ is non-increasing wrt.~$R$ for all $k \geq 0$.
\end{lemma}

\begin{proof}
Induction in $k$, followed by induction in $\LAP{v}{v'}$.
The case where $k = 0$ is trivial as then $\omega_k = 0$.
Now let $k > 0$, in which case $\omega_k = \HH{\unodes}(\omega_{k-1})$.
If $v' = v$, then $\hpd{\omega_k}{v}{v'}(D) = D$
and the claim is trivial. 
Otherwise, let $v_0 = \fppd{v}$; if $v_0 \neq v'$
then Lemma~\ref{lem:outside-decomp} tells us that we can apply
the induction hypothesis twice to infer that
$\hpd{\omega_k}{v}{v_0}$ and
$\hpd{\omega_k}{v_0}{v'}$ are both non-increasing wrt.~$R$
which shows that
$\hpd{\omega_k}{v}{v'}$ is non-increasing wrt.~$R$ since for $s \in \sto{R}$
we have
\[
\hpd{\omega_k}{v}{v'}(D)(s) =
\hpd{\omega_k}{v_0}{v'}(\hpd{\omega_k}{v}{v_0}(D))(s)
\leq
\hpd{\omega_k}{v}{v_0}(D)(s)
\leq 
D(s).
\]
We are left with case where $v' = \fppd{v}$.
If $\labv{v}$ is of the form $\observev{B}$ the claim is trivial.

If $\labv{v}$ is of the form $\assignv{x}{E}$ or of the form
$\rassignv{x}{\psi}$ we first infer that $x \notin R$ because otherwise,
as $Q$ is closed under data dependence, we would have $v \in Q$
which contradicts that $v$ stays outside $Q$ until $v'$.
But then the claim follows from
Lemmas~\ref{lem:assignRirrel} and \ref{lem:rassignRirrel}.

The last case is if $v$ is a branching node with condition $B$,
with $v_1$ the $\true$-successor of $v$ and $v_2$ the $\false$-successor
of $v$.  
For each $D \in \Dist$, let $D_1 = \selectf{B}(D)$ and 
$D_2 = \selectf{\neg B}(D)$ and let (for $i \in \{1,2\}$)
$D'_i$ be defined as $\hpd{\omega_k}{v_i}{v'}(D_i)$
if $\LAP{v_i}{v'} < \LAP{v}{v'}$,
but otherwise as
$\hpd{\omega_{k-1}}{v_i}{v'}(D_i)$.
For each $i \in \{1,2\}$ we can apply,
as $v_i$ stays outside $Q$ until $v'$
(Lemma~\ref{lem:outside-branch}),
either the outer induction hypothesis, or the inner induction hypothesis,
to infer that for all $s \in \sto{R}$,
$D'_i(s) \leq D_i(s)$.
For each $s \in \sto{R}$
we thus infer the desired
\[
\hpd{\omega_k}{v}{v'}(s) = D'_1(s) + D'_2(s)
\leq D_1(s) + D_2(s) = D(s).
\]
\end{proof}

{\bf Lemma~\ref{lem:sum-preserve-not-affect-rv}}
With $(v,v') \in \PD$, assume that
$Q$ is closed under data dependence and that $v$ stays outside $Q$
until $v'$ (by Lemma~\ref{lem:outside-same-rv} it thus makes sense
to define $R = \rv{Q}{v} = \rv{Q}{v'}$).

For all distributions $D$,
if $\sumd{\hpd{\omega}{v}{v'}(D)} = \sumd{D}$ 
then $\dagree{\hpd{\omega}{v}{v'}(D)}{D}{R}$.

\begin{proof}
Given $D \in \Dist$ with $\sumd{\hpd{\omega}{v}{v'}(D)} = \sumd{D}$,
for all $s \in \sto{R}$,
Lemma~\ref{lem:outside-nonincrease} yields
$\hpd{\omega_k}{v}{v'}(D)(s) \leq D(s)$ for all $k \geq 0$,
and as $\omega = \limit{k}{\omega_k}$ 
this implies ---since $\omega$ is the \emph{least} upper bound--- 
$\hpd{\omega}{v}{v'}(D)(s) \leq D(s)$.
We thus get (using Lemma~\ref{lem:sumd})
\[
\sumd{\hpd{\omega}{v}{v'}(D)} =
\sum_{s \in \sto{R}} \hpd{\omega}{v}{v'}(D)(s)
\leq
\sum_{s \in \sto{R}} D(s)
= \sumd{D}.
\]

If 
$\sumd{\hpd{\omega}{v}{v'}(D)} = \sumd{D}$ we infer from the above that 
\[
\sum_{s \in \sto{R}} \hpd{\omega}{v}{v'}(D)(s)
=
\sum_{s \in \sto{R}} D(s) 
\]
which since $\hpd{\omega}{v}{v'}(D)(s) \leq D(s)$ for all $s \in \sto{R}$
is possible only if $\hpd{\omega}{v}{v'}(D)(s) = D(s)$ for all $s \in \sto{R}$,
that is $\dagree{\hpd{\omega}{v}{v'}(D)}{D}{R}$.
\end{proof}

\subsection{Proofs for Section~\ref{sec:conditions}}

{\bf Lemma~\ref{lem:weak-union}}
If $Q_1$ and $Q_2$ are weak slice sets,
also $Q_1 \cup Q_2$ is a weak slice set.
\begin{proof}
Let $Q = Q_1 \cup Q_2$. To see that $Q$ 
is closed under data dependence, assume that $\dd{v}{v'}$ with
$v' \in Q$; wlog.\ we can assume $v' \in Q_1$ which since
$Q_1$ is closed under data dependence implies $v \in Q_1$ 
and thus $v \in Q$.

We shall now look at a node $v$, and argue that
$\nextq{Q}{v}$ exists.
If $\nextq{Q_1}{v} = v$ or $\nextq{Q_2}{v} = v$,
then $v \in Q \cup \mkset{\finalv}$ and thus
$\nextq{Q}{v} = v$.
Otherwise, let $v_1 = \nextq{Q_1}{v}$
and $v_2 = \nextq{Q_2}{v}$; both $v_1$ and $v_2$
are proper postdominators of $v$ so by
Lemma~\ref{lem:prec-ordering} we can wlog.\ assume
that $v_1$ occurs before $v_2$ in all paths from $v$ to $\finalv$
(or that $v_1 = v_2$).

We shall now show that 
$v_1 = \nextq{Q}{v}$, where we first observe that
$v_1 \in Q \cup \mkset{\finalv}$.
Now consider a path $\pi$ from
$v$ to $Q \cup \mkset{\finalv}$. We must show that $\pi$ contains
$v_1$, which is obvious
if the path $\pi$ is to $Q_1 \cup \mkset{\finalv}$.
Otherwise, when $\pi$ is
to $Q_2$, we infer that $\pi$ contains $v_2$
(which yields the claim if $v_1 = v_2$).
Since all nodes have a path to $\finalv$,
$\pi$ is a prefix of a path $\pi'$ from $v$ to $\finalv$;
as $v_1$ occurs before $v_2$ in all paths from $v$ to $\finalv$, 
we see that $v_1$ occurs before $v_2$ in $\pi'$.
We infer that $v_1$ occurs also in $\pi$, as desired.
\end{proof}

{\bf Lemma~\ref{lem:sum-pres}}
Assume $Q'$ is a node set which 
contains all $\observevv$ nodes, and 
that for each cycle-inducing node $v_0$,
either $v_0 \in Q'$ or $\hpd{\omega}{v_0}{\fppd{v_0}}$ is sum-preserving.
If $v$ stays outside $Q'$ until $v'$
then $\sumd{\hpd{\omega}{v}{v'}(D)} = \sumd{D}$ for all $D$.

\begin{proof}
We do induction in $\LAP{v}{v'}$.
The claim is obvious if $v' = v$, so we can assume that $v'$ is a proper
postdominator of $v$.

With $v'' = \fppd{v}$, let us first assume that $v' \neq v''$.
By Lemma~\ref{lem:LAP-add} we have
$\LAP{v}{v''} < \LAP{v}{v'}$ 
and $\LAP{v''}{v'} < \LAP{v}{v'}$,
and by Lemma~\ref{lem:outside-decomp} we see that
$v$ stays outside $Q'$ until $v''$ and that $v''$ stays outside $Q'$ until $v'$.
We can thus apply the induction hypothesis on
$\hpd{\omega}{v}{v''}$ (the third equality) and on 
$\hpd{\omega}{v''}{v'}$ (the second equality) to infer that
for all $D$ we have
\[
\sumd{\hpd{\omega}{v}{v'}(D)} =
\sumd{\hpd{\omega}{v''}{v'}(\hpd{\omega}{v}{v''}(D))} =
\sumd{\hpd{\omega}{v}{v''}(D)} = \sumd{D}.
\]
Thus we can now assume that $v' = \fppd{v}$.
If $v$ is labeled $\skipv$ the claim is trivial;
if $v$ is labeled $\assignv{x}{E}$ (or
$\rassignv{x}{\psi}$) then the claim follows from
Lemma~\ref{lem:assign-misc}
(or Lemma~\ref{lem:rassign-misc}).
Note that $v \notin Q'$ (as $v$ stays outside $Q'$ until $v'$),
so our assumptions entail that
$v$ cannot be an $\observevv$ node, and that if $v$ is cycle-inducing
then $\hpd{\omega}{v}{v'}$ is sum-preserving and thus the claim.

We are thus left with the case that $v$ is a branching node
which is not cycle-inducing.
With $v_1$ the $\true$-successor and $v_2$ the $\false$-successor of $v$, 
we thus have
$\LAP{v_1}{v'} < \LAP{v}{v'}$ and
$\LAP{v_2}{v'} < \LAP{v}{v'}$,
and by Lemma~\ref{lem:outside-branch} also
that $v_1$ and $v_2$ both stay outside $Q'$ until $v'$.
Hence we can apply the induction hypothesis to $\hpd{\omega}{v_1}{v'}$
and $\hpd{\omega}{v_2}{v'}$, to get the desired result:
\begin{eqnarray*}
\sumd{\hpd{\omega}{v}{v'}(D)} & = &
\sumd{(\hpd{\omega}{v_1}{v'}(\selectf{B}(D)) + \hpd{\omega}{v_2}{v'}(\selectf{\neg B}(D)))}
\\ & = &
\sumd{\hpd{\omega}{v_1}{v'}(\selectf{B}(D))} + 
\sumd{\hpd{\omega}{v_2}{v'}(\selectf{\neg B}(D))}
\\ & = & \sumd{\selectf{B}(D)} + \sumd{\selectf{\neg B}(D)} \ = \sumd{D}.
\end{eqnarray*}
\end{proof}

\subsection{Proofs for Section~\ref{sec:correct}}

\begin{lemma}
\label{lem:gamma-mult-nonincr}
For each $(v,v') \in \PD$, and each $k \geq 0$,
$\hpd{\gamma_k}{v}{v'}$ 
is additive, multiplicative and non-increasing,
\end{lemma}
\begin{proof}
We know from Lemma~\ref{lem:fixed-mult-nonincr-determ} that
$\omega$ is additive, multiplicative and non-increasing
(and so is the function 0);
the result thus follows from Lemma~\ref{lem:HH-preserve-misc}.
\end{proof}
Similarly, we have (with $\Phi_k$ defined in Def.~\ref{def:phipk}):
\begin{lemma}
\label{lem:Phi-mult-nonincr}
For each $(v,v') \in \PD$, and each $k \geq 0$,
$\hpd{\Phi_k}{v}{v'}$ 
is additive, multiplicative and non-increasing,
\end{lemma}
{\bf Lemma~\ref{lem:gamma-as-omega}}
Assume that $v$ stays outside $Q \cup Q_0$ until $v'$.
Then $\hpd{\gamma_k}{v}{v'} = \hpd{\omega}{v}{v'}$ holds for all $k \geq 0$.

\begin{proof}
We do induction in $k$, where the base case $k = 0$ follows from
the definition of $\gamma_0$.

For the inductive case, where 
$\gamma_{k} = \HH{\unodes}(\gamma_{k-1})$ with $k > 0$,
we do induction in $\LAP{v}{v'}$, with a case analysis on the 
definition of $\HH{\unodes}$:
\begin{itemize}
\item
If $v' = v$ then $\hpd{\gamma_k}{v}{v'}(D) = D = \hpd{\omega}{v}{v'}(D)$;
\item
If with $v'' = \fppd{v}$ we have $v' \neq v''$ then
we know from Lemma~\ref{lem:outside-decomp} that
$v$ stays outside $Q \cup Q_0$ until $v''$,
and that $v''$ stays outside $Q \cup Q_0$ until $v'$;
by Lemma~\ref{lem:LAP-add} we know that
$\LAP{v}{v''} < \LAP{v}{v'}$ 
and $\LAP{v''}{v'} < \LAP{v}{v'}$.
Hence we can apply the inner induction hypothesis to infer that
$\hpd{\gamma_k}{v}{v''} = \hpd{\omega}{v}{v''}$ and
$\hpd{\gamma_k}{v''}{v'} = \hpd{\omega}{v''}{v'}$.
But then we get the desired
\[
\hpd{\gamma_k}{v}{v'} = \hpd{\gamma_k}{v}{v''}\; ; \; \hpd{\gamma_k}{v''}{v'}
= \hpd{\omega}{v}{v''}\; ; \; \hpd{\omega}{v''}{v'}
= \hpd{\omega}{v}{v'}.
\]
\item
Otherwise, when  $v' = \fppd{v}$, the claim is trivial except when $v$ 
is a branching node. So consider such a $v$, and let $B$ be its condition,
$v_1$ its $\true$-successor, and $v_2$ its $\false$-successor.
Our goal is to prove, for a given $D$, that
$\hpd{\gamma_k}{v}{v'}(D) = \hpd{\omega}{v}{v'}(D)$ which amounts to
\[
\hpd{\HH{\unodes}(\gamma_{k-1})}{v}{v'}(D) =
\hpd{\HH{\unodes}(\omega)}{v}{v'}(D).
\]
With $D_1 = \selectf{B}(D)$ and $D_2 = \selectf{\neg B}(D)$,
examining the definition of $\HH{\unodes}$ shows that
it suffices if for $i \in \{1,2\}$ we can prove:
\begin{itemize}
\item
if $\LAP{v_i}{v'} < \LAP{v}{v'}$ then
\[
\hpd{\gamma_k}{v_i}{v'}(D_i) =
\hpd{\omega}{v_i}{v'}(D_i)
\]
which follows by the inner induction hypothesis;
\item
if $\LAP{v_i}{v'} \geq \LAP{v}{v'}$ then
\[
\hpd{\gamma_{k-1}}{v_i}{v'}(D_i) =
\hpd{\omega}{v_i}{v'}(D_i)
\]
which follows by the outer induction hypothesis.
\end{itemize}
\end{itemize}
\end{proof}

To prepare for the proof of Lemma~\ref{lem:indpd2indpd},
we state a result about branching nodes:
\begin{lemma}
\label{lem:indpd-branch}
Assume that $v$ is a branching node,
with $B$ its condition,
and $v_1$ its $\true$-successor and $v_2$ its $\false$-successor.
Let $D \in \Dist$
with $D_1 = \selectf{B}(D)$ and $D_2 = \selectf{\neg B}(D)$. 
Assume that $R,R_0,R_1,R_2$ are such that
$\fv{B} \cup R_1 \cup R_2 \subseteq R$
and $R \cap R_0 = \emptyset$.
Finally assume that
$R$ and $R_0$ are independent in $D$.
Then for $i = 1,2$ we have
\begin{enumerate}
\item 
\label{branch-p1}
$R_i$ and $R_0$ are independent in $D_i$.
\item 
\label{branch-p2}
$\forall s_0 \in \sto{R_0}:
D(s_0)\sumd{D_i} = D_i(s_0)\sumd{D}$, and
\end{enumerate}
\end{lemma}
\begin{proof}
We shall consider only the case $i = 1$ (as the case $i = 2$ is symmetric).
For (\ref{branch-p2}), we have the calculation
(where the 3rd equality follows 
from $R$ and $R_0$ being independent in $D$)
\begin{eqnarray*}
D(s_0)\sumd{D_1} 
& = &
D(s_0) \sum_{s \in \sto{R}\ \mid\ \seme{B}s}D(s) \
=\
\sum_{s \in \sto{R}\ \mid\ \seme{B}s}D(s_0)D(s)
\\[1mm] & = &
\sum_{s \in \sto{R}\ \mid\ \seme{B}s}\left(D(\adds{s_0}{s})\sumd{D}\right)\
 =\
\left(\sum_{s \in \sto{R}\ \mid\ \seme{B}s}D(\adds{s_0}{s})\right)\sumd{D}\ 
\\[1mm] & = &
\left(\sum_{s' \in \sto{R \cup R_0}\ \mid\ \sagree{s'}{s_0}{R_0},\ \seme{B}s'}D(s')\right)\sumd{D} \\[1mm] 
& = &
\left(\sum_{s \in \fulls\ \mid\ \sagree{s}{s_0}{R_0},\ \seme{B}s}D(s)\right)\sumd{D}\ 
  =\
\left(\sum_{s \in \fulls\ \mid\ \sagree{s}{s_0}{R_0}}D_1(s)\right)\sumd{D}\ =\
D_1(s_0)\sumd{D}
\end{eqnarray*}
For (\ref{branch-p1})
we have with $s_1 \in \sto{R_1}$ and $s_0 \in \sto{R_0}$
the calculation
(which uses (\ref{branch-p2}) and the fact that if $D = 0$ 
then the claim is trivial)
\begin{eqnarray*}
D_1(\adds{s_1}{s_0})\sumd{D_1}
& = & 
\left(\sum_{s \in \sto{R}\ \mid\ \sagree{s}{s_1}{R_1},\ \seme{B}s}D(\adds{s}{s_0})\right)\sumd{D_1}
\\[1mm] & = &
\left(\sum_{s \in \sto{R}\ \mid\ \sagree{s}{s_1}{R_1},\ \seme{B}s}\frac{D(s)D(s_0)}{\sumd{D}}\right)\sumd{D_1}
\\[1mm] & = & 
\hspace*{-1mm}
\left(\sum_{s \in \sto{R}\ \mid\ \sagree{s}{s_1}{R_1},\ \seme{B}s}D(s)\right)\frac{D(s_0)\sumd{D_1}}{\sumd{D}}
\\[1mm]
& = &
D_1(s_1) D_1(s_0).
\end{eqnarray*}
\end{proof}
To facilitate the proof of Lemma~\ref{lem:indpd2indpd},
we introduce some notation:
\begin{definition}
\label{def:indpd-pres}
We say that $h \in \PD \rightarrow \Dist \contarrow \Dist$ \dt{preserves
probabilistic independence} iff for all slicing pairs $(Q,Q_0)$,
the following holds for all $(v,v') \in \PD$, 
with $R = \rv{Q}{v}$, $R' = \rv{Q}{v'}$,
$R_0 = \rv{Q_0}{v}$, and $R'_0 = \rv{Q_0}{v'}$:
for all $D$ such that
$R$ and $R_0$ are independent in $D$,
with $D' = \hpd{h}{v}{v'}(D)$ it is the case that
\begin{enumerate}
\item 
\label{indpdlem2:indp}
$R'$ and $R'_0$ are independent in $D'$
\item
\label{indpdlem2:outq}
 if $v$ stays outside $Q$ until $v'$ (and thus $R' = R$) then for all 
$s \in \sto{R}$ we have
\begin{displaymath}
D(s)\sumd{D'} = D'(s)\sumd{D}
\end{displaymath}
\item
\label{indpdlem2:outq0}
if $v$ stays outside $Q_0$ until $v'$
(and thus $R'_0 = R_0$)
then for all $s_0 \in \sto{R_0}$ we have
\begin{displaymath}
D(s_0)\sumd{D'} = D'(s_0)\sumd{D}.
\end{displaymath}
\end{enumerate}
\end{definition}

\begin{lemma}
\label{lem:indpd2indpdk}
For each $k \geq 0$, $\gamma_k$ 
preserves probabilistic independence.
\end{lemma}
\begin{proof}
We shall proceed by induction in $k$.
We shall first consider the base case $k = 0$.
For a given $(v,v') \in \PD$,
the claims are trivial if $\hpd{\gamma_0}{v}{v'} = 0$, so 
assume that
$v$ stays outside $Q \cup Q_0$ until $v'$
(implying $R' = R$ and $R'_0 = R_0$)
and thus $\hpd{\gamma_0}{v}{v'} = \hpd{\omega}{v}{v'}$.
We can apply Lemma~\ref{lem:outsideQQ0irrelevant}
(and part 1 of Lemma~\ref{lem:rv-cup})
to infer that for all $D$,
with $D' = \hpd{\gamma_0}{v}{v'}(D)$ we have 
$\dagree{D'}{D}{R \cup R_0}$,
and by Lemma~\ref{lem:dist-partial}
thus also $\dagree{D'}{D}{R}$ and $\dagree{D'}{D}{R_0}$ and 
$\dagree{D'}{D}{\emptyset}$.
That is, for $s \in \sto{R}$ and $s_0 \in \sto{R_0}$ 
we have
$D'(\adds{s}{s_0}) = D(\adds{s}{s_0})$ and
$D'(s) = D(s)$ and $D'(s_0) = D(s_0)$, 
and also $D'(\emptyset) = D(\emptyset)$ which
amounts to $\sumd{D'} = \sumd{D}$.
This clearly implies claims
\ref{indpdlem2:outq} and \ref{indpdlem2:outq0}
in Definition~\ref{def:indpd-pres},
and also claim \ref{indpdlem2:indp}
since for $s \in \sto{R}$ and $s_0 \in \sto{R_0}$ 
we have, by our assumption that $R$ and $R_0$ are independent in $D$:
\[
D'(\adds{s}{s_0})\sumd{D'} =
D(\adds{s}{s_0})\sumd{D} =
D(s)D(s_0) = D'(s)D'(s_0).
\]
We shall next consider the case $k > 0$,
where we assume that 
$\gamma_{k-1}$ preserves probabilistic independence and
with $\gamma_k = \HH{\unodes}(\gamma_{k-1})$ we must then prove that 
$\gamma_k$ preserves probabilistic independence,
that is: given $(v,v') \in \PD$
with $R = \rv{Q}{v}$, $R' = \rv{Q}{v'}$,
$R_0 = \rv{Q_0}{v}$, and $R'_0 = \rv{Q_0}{v'}$,
and given $D \in \Dist$ such that $R$ and $R_0$ are independent in $D$,
with $D' = \hpd{\gamma_k}{v}{v'}(D)$
we must show that:
\begin{enumerate}
\item
$R'$ and $R'_0$ are independent in $D'$
\item
if $v$ stays outside $Q$ until $v'$ (and thus $R' = R$) then for all 
$s \in \sto{R}$ we have
\begin{displaymath}
D(s)\sumd{D'} = D'(s)\sumd{D}
\end{displaymath}
\item
if $v$ stays outside $Q_0$ until $v'$
(and thus $R'_0 = R_0$)
then for all $s_0 \in \sto{R_0}$ we have
\begin{displaymath}
D(s_0)\sumd{D'} = D'(s_0)\sumd{D}.
\end{displaymath}
\end{enumerate}
We shall establish the required claims by induction in 
$\LAP{v}{v'}$.
First observe that if $D' = 0$ the claims are trivial.
We can thus assume that $\sumd{D'} > 0$, which
by Lemma~\ref{lem:gamma-mult-nonincr} entails that 
$\sumd{D} > 0$. 

If $v' = v$, then $D' = D$ and $R' = R$ and $R'_0 = R_0$
and again the claims are trivial.

Otherwise, let $v'' = \fppd{v}$,
and first assume that $v' \neq v''$ in which case the situation
is that there exists $D''$ such that
$\hpd{\gamma_k}{v}{v''}(D) = D''$ and
$\hpd{\gamma_k}{v''}{v'}(D'') = D'$;
by Lemma~\ref{lem:gamma-mult-nonincr}
we can assume that $D'' \neq 0$ (as otherwise $D' = 0$
which we have already considered).
By Lemma~\ref{lem:LAP-add} we have
$\LAP{v}{v''} < \LAP{v}{v'}$ 
and $\LAP{v''}{v'} < \LAP{v}{v'}$,
so we can apply the induction hypothesis
on $\hpd{\gamma_k}{v}{v''}$ and on $\hpd{\gamma_k}{v''}{v'}$.
With $R'' = \rv{Q}{v''}$ and $R''_0 = \rv{Q_0}{v''}$,
the induction hypothesis now first gives us 
that $R''$ and $R''_0$ are independent in $D''$,
and next that
$R'$ and $R'_0$ are independent in $D'$.

Concerning claim \ref{indpdlem2:outq}
(claim \ref{indpdlem2:outq0} is symmetric),
assume that $v$ stays outside $Q$ until $v'$;
by Lemma~\ref{lem:outside-decomp} we see that
$v$ stays outside $Q$ until $v''$ and
$v''$ stays outside $Q$ until $v'$.
Inductively, we can thus assume that for $s \in \sto{R}$ we have
\begin{displaymath}
D(s)\sumd{D''} = D''(s)\sumd{D}
\mbox{ and }
D''(s)\sumd{D'} = D'(s)\sumd{D''}.
\end{displaymath}
which since $\sumd{D''} > 0$ gives us the desired
\begin{eqnarray*}
D(s)\sumd{D'} & = & \displaystyle
\frac{D''(s)\sumd{D}}{\sumd{D''}} \sumd{D'} \ =
D''(s)\sumd{D'} \frac{\sumd{D}}{\sumd{D''}} \\[1mm]
& = & \displaystyle
D'(s)\sumd{D''}\frac{\sumd{D}}{\sumd{D''}} \ =
D'(s)\sumd{D}.
\end{eqnarray*}

We are left with the case $v' = \fppd{v}$,
and split into several cases, depending on $\labv{v}$, where
the case for $\skipv$ is trivial.

{\bf Case 1: $\mathbf{v}$ is an observe node.}
With $B$ the condition, for $s'$ with $\fv{B} \subseteq \dom{s'}$
we thus have $D'(s') = D(s')$ if $\seme{B}s'$, and
$D'(s') = 0$ otherwise.
As $(Q,Q_0)$ is a slicing pair, either $v \in Q$ or $v \in Q_0$.
Let us assume that $v \in Q$; the other case is symmetric.
Thus $\fv{B} \subseteq R$, and also $R' \subseteq R$ 
as $\defv{v} = \emptyset$;
as $v$ stays outside $Q_0$ until $v'$, 
by Lemma~\ref{lem:outside-same-rv} we also have
$R'_0 = R_0$. 

The following calculation, where the 3rd equality is due to the assumption
that $R$ and $R_0$ are independent in $D$, shows that
for $s_0 \in \sto{R_0}$ we have
$D'(s_0)\sumd{D} = D(s_0)\sumd{D'}$:
\begin{eqnarray*}
D'(s_0)\sumd{D} & = &
\left(\sum_{s \in \sto{R}}D'(\adds{s}{s_0})\right) \sumd{D} =
\left(\sum_{s \in \sto{R}\ \mid\ \seme{B}s}D(\adds{s}{s_0})\right) \sumd{D}
\\[1mm] 
& = &
\sum_{s \in \sto{R}\ \mid\ \seme{B}s}\left(D(s)D(s_0)\right)\ =
D(s_0) \left(\sum_{s \in \sto{R}\ \mid\ \seme{B}s}D(s)\right)
\\[1mm] 
& = &
D(s_0) \left(\sum_{s \in \sto{R}}D'(s)\right) =
D(s_0)\sumd{D'}.
\end{eqnarray*}
This yields claim \ref{indpdlem2:outq0}
(while claim \ref{indpdlem2:outq} vacuously holds)
and also gives the last equality in the following 
derivation that establishes 
(again using the assumption that $R$ and $R_0$ are independent in $D$)
claim~\ref{indpdlem2:indp} by considering $s' \in \sto{R'}$ and $s_0 \in \sto{R_0}$:
\begin{eqnarray*}
D'(\adds{s'}{s_0})\sumd{D'}
& = &
\left(\sum_{s \in \sto{R}\ \mid\ \sagree{s}{s'}{R'}}
D'(\adds{s}{s_0})\right)\sumd{D'}
\\[1mm] & = &
\left(\left(\sum_{s \in \sto{R}\ \mid\ \sagree{s}{s'}{R'},\ \seme{B}s}
D(\adds{s}{s_0})\right)\sumd{D}\right)\frac{\sumd{D'}}{\sumd{D}}
\\[1mm] & = &
\left(\sum_{s \in \sto{R}\ \mid\ \sagree{s}{s'}{R'},\ \seme{B}s}
D(s) D(s_0)\right)\frac{\sumd{D'}}{\sumd{D}}
\\[1mm] & = &
\left(\sum_{s \in \sto{R}\ \mid\ \sagree{s}{s'}{R'}}D'(s)\right)
\frac{D(s_0)\sumd{D'}}{\sumd{D}}\\[1mm]
& = &
D'(s')D'(s_0).
\end{eqnarray*}

{\bf Case 2: $\mathbf{v}$ has exactly one successor and 
is not an $\mathbf{\observevv}$ node.}
Then $v$ is labeled with an assignment or with a random assignment;
in both cases, the transfer function is sum-preserving
(Lemmas~\ref{lem:assign-misc} and \ref{lem:rassign-misc})
so we get $\sumd{D'} = \sumd{D}$. 
We further infer by Lemma~\ref{lem:sum-preserve-not-affect-rv}
that if $v \notin Q$ then $R' = R$ and
$\dagree{D'}{D}{R}$
(as then $v$ stays outside $Q$ until $v'$);
similarly, if $v \notin Q_0$
then $R'_0 = R_0$ and $\dagree{D'_0}{D_0}{R_0}$.

The above observations obviously establish
the claims~\ref{indpdlem2:outq} and \ref{indpdlem2:outq0}.
We shall now address claim~\ref{indpdlem2:indp},
that is show that
$D'(\adds{s}{s_0})\sumd{D'}  = D'(s)D'(s_0)$
for $s' \in \sto{R'}$ and $s_0 \in \sto{R_0}$.
To do so, we shall do a case analysis
on whether $v \in Q \cup Q_0$.

First consider the case where $v \notin Q \cup Q_0$.
Then (by Lemma~\ref{lem:sum-preserve-not-affect-rv})
we get $R' = R$ and $R'_0 = R_0$ and $\dagree{D'}{D}{R \cup R_0}$.
This establishes 
claim~\ref{indpdlem2:indp} since
for $s \in \sto{R}$ 
and $s_0 \in \sto{R_0}$ we have (using the assumption that $R$ and $R_0$
are independent in $D$)
$D'(\adds{s}{s_0})\sumd{D'} = 
D(\adds{s}{s_0})\sumd{D} =
D(s)D(s_0) =
D'(s)D'(s_0)$.

Next consider the case where $v \in Q \cup Q_0$.
Without loss of generality, we may assume that $v \in Q$.
Thus $v \notin Q_0$ so $R'_0 = R_0$ and
$R_0 \cap \defv{v} = \emptyset$ and
$D'(s_0) = D(s_0)$ for all $s_0 \in \sto{R_0}$.
We split into two cases, depending on whether $R' \cap \defv{v}$ is empty or not.

First assume that
$R' \cap \defv{v} = \emptyset$. 
Then $R' \subseteq R$,
and by Lemmas~\ref{lem:assignRirrel} and \ref{lem:rassignRirrel}
(as $(R' \cup R_0) \cap \defv{v} = \emptyset$) we get
$\dagree{D}{D'}{R' \cup R_0}$.
For $s' \in \sto{R'}$ and $s_0 \in \sto{R_0}$ we thus have
$D'(\adds{s'}{s_o}) = D(\adds{s'}{s_0})$ and $D'(s') = D(s')$ and
$D'(s_0) = D(s_0)$ and 
$\sumd{D'} = \sumd{D}$, which 
(using the assumption that $R$ and $R_0$
are independent in $D$)
gives us the desired result
\begin{eqnarray*}
D'(\adds{s'}{s_0})\sumd{D'} 
& = & 
D(\adds{s'}{s_0})\sumd{D} \\
& = &
\left(\sum_{s \in \sto{R} \ \mid\ \sagree{s}{s'}{R'}}D(\adds{s}{s_0})\right)\sumd{D}
\\ &  = &
\sum_{s \in \sto{R} \ \mid\ \sagree{s}{s'}{R'}} D(s)D(s_0)
  =
D(s')D(s_0)
  =
D'(s')D'(s_0).
\end{eqnarray*}
Next assume that
$R' \cap \defv{v} \neq \emptyset$.
We now (finally) need to do a case analysis on the kind of assignment.

If $\labv{v} = \assignv{x}{E}$, we have 
(by Lemma~\ref{lem:rv-assign})
$R = (R'\setminus\mkset{x}) \cup \fv{E}$ 
and from our case assumptions also $x \in R'$ and $x \notin R_0$.
Let us now consider $s' \in \sto{R'}$ and $s_0 \in \sto{R_0}$;
the claim follows from the below calculation where the third equality
uses the assumption that $R$ and $R_0$ are independent in $D$,
and the first and last equality both uses
Lemma~\ref{lem:assign-simp}:
\begin{eqnarray*}
D'(\adds{s'}{s_0}) \sumd{D'}
& = &
\left(\sum_{s_1 \in \sto{R\cup R_0}\ \mid\ \sagree{s_1}{\adds{s'}{s_0}}{R' \setminus\mkset{x} \cup R_0},\ (\adds{s'}{s_0})(x) = \seme{E}s_1} 
\hspace*{-7mm}
D(s_1)\right) \sumd{D}
\\[1mm] & = &
\left(\sum_{s \in \sto{R}\ \mid\ \sagree{s}{s'}{R' \setminus\mkset{x}},\ s'(x) = \seme{E}s} 
\hspace*{-7mm}
D(\adds{s}{s_0})\right) \sumd{D}
\\[1mm] & = &
\sum_{s \in \sto{R}\ \mid\ \sagree{s}{s'}{R' \setminus\mkset{x}},\ s'(x) = \seme{E}s} 
\hspace*{-7mm}
D(s)D(s_0)\\[1mm]
& = &
D'(s') D'(s_0).
\end{eqnarray*}
Finally, if $\labv{v} = \rassignv{x}{\psi}$, we have 
$R = (R'\setminus\mkset{x})$
and from our case assumptions also $x \in R'$ and $x \notin R_0$.
Let us now consider $s' \in \sto{R'}$ and $s_0 \in \sto{R_0}$:
the claim follows from the below calculation where the third equality
uses the assumption that $R$ and $R_0$ are independent in $D$,
and the first and last equality both uses
Lemma~\ref{lem:rassign-simp}:
\begin{eqnarray*}
D'(\adds{s'}{s_0}) \sumd{D'}
& = &  \psi((\adds{s'}{s_0})(x))
\left(\sum_{s_1 \in \sto{R\cup R_0}\ \mid\ \sagree{s_1}{\adds{s'}{s_0}}{R' \setminus\mkset{x} \cup R_0}} 
\hspace*{-5mm}
D(s_1)\right) \sumd{D}
\\[1mm] & = & \psi(s'(x))
\left(\sum_{s \in \sto{R}\ \mid\ \sagree{s}{s'}{R' \setminus\mkset{x}}} 
\hspace*{-6mm}
D(\adds{s}{s_0})\right) \sumd{D}
\\[1mm] & = & \psi(s'(x))
\left(\sum_{s \in \sto{R}\ \mid\ \sagree{s}{s'}{R' \setminus\mkset{x}}} 
\hspace*{-5mm}
D(s)D(s_0)\right)\\[1mm]
& = &
D'(s') D'(s_0).
\end{eqnarray*}

{\bf Case 3: $\mathbf{v}$ is a branching node.}
First assume that $v \notin Q \cup Q_0$.
Here $Q \cup Q_0$ is a weak slice set (Lemma~\ref{lem:weak-union}),
so from Lemma~\ref{lem:notQ-outside} we see that
$v$ stays outside $Q \cup Q_0$ until $v'$;
thus $R' = R$ and $R'_0 = R_0$ (by Lemma~\ref{lem:outside-same-rv}).
By Lemma~\ref{lem:gamma-as-omega} we see that
$D' = \hpd{\gamma_k}{v}{v'}(D) = \hpd{\omega}{v}{v'}(D)$,
and Lemma~\ref{lem:outsideQQ0irrelevant}
thus tells us that
$\dagree{D'}{D}{R \cup R_0}$.
In particular for all $s \in \sto{R}$ and $s_0 \in \sto{R_0}$
we have
$D'(\adds{s}{s_0}) = D(\adds{s}{s_0})$,
$D'(s) = D(s)$, $D'(s_0) = D(s_0)$ and $\sumd{D'} = \sumd{D}$.
But this clearly entails all the 3 claims.

In the following, we can thus assume that $v \in Q \cup Q_0$,
and shall only look at the case $v \in Q$ as the case $v \in Q_0$
is symmetric. 
Claim \ref{indpdlem2:outq} thus holds vacuously;
we shall embark on the other two claims.
As $v \in Q$ we have  (by Lemma~\ref{lem:rv-branch})
$\fv{B} \subseteq R = \rv{Q}{v}$, and as $v \notin Q_0$ we see
(by Lemma~\ref{lem:notQ-outside}) that
$v$ stays outside $Q_0$ until $v'$ 
so that (by Lemma~\ref{lem:outside-same-rv}) $R'_0 = R_0$.
With $v_1$ the $\true$-successor of $v$ and $v_2$ the $\false$-successor
of $v$, and with $D_1 = \selectf{B}(D)$ and $D_2 = \selectf{\neg B}(D)$,
the situation is that
$D' = D'_1 + D'_2$ where for each
$i \in \{1,2\}$, $D'_i$ is computed as
\begin{itemize}
\item
if $\LAP{v_i}{v'} < \LAP{v}{v'}$ then
$D'_i =\hpd{\gamma_{k}}{v_i}{v'}(D_i)$;
\item
if $\LAP{v_i}{v'} \geq \LAP{v}{v'}$ then
$D'_i = \hpd{\gamma_{k-1}}{v_i}{v'}(D_i)$.
\end{itemize}
Let $R_1 = \rv{Q}{v_1}$ and $R_2 = \rv{Q}{v_2}$; 
thus (by Lemma~\ref{lem:rv-branch}) $R_1 \subseteq R$ and $R_2 \subseteq R$.

By Lemma~\ref{lem:indpd-branch}, we see that
\begin{equation}
\label{indpd:cond1}
\forall i \in \{1,2\}: R_i \mbox{ and } R_0 \mbox{ are independent in } D_i
\end{equation}
\begin{equation}
\label{indpd:cond2}
\forall i \in \{1,2\}: 
D(s_0)\sumd{D_i} = D_i(s_0)\sumd{D} \mbox{ for all } s_0 \in \sto{R_0}.
\end{equation}
Given (\ref{indpd:cond1}), and the fact
that $v_i$ stays outside $Q_0$ until $v'$, we can infer that
\begin{equation}
\label{indpd:cond3}
\forall i \in \{1,2\}: R' \mbox{ and } R_0 \mbox{ are independent in } D'_i
\end{equation}
\begin{equation}
\label{indpd:cond4}
\forall i \in \{1,2\}: D_i(s_0)\sumd{D'_i} = D'_i(s_0)\sumd{D_i}
\mbox{ for all } s_0 \in \sto{R_0}
\end{equation}
since when $\LAP{v_i}{v'} < \LAP{v}{v'}$ this follows from the 
(inner) induction hypothesis,
and otherwise it follows from the assumption (the outer induction)
about $\gamma_{k-1}$.
We also have
\begin{equation}
\label{indpd:cond5}
D(s_0)\sumd{D'} =
D'(s_0) \sumd{D} 
\mbox{ for all } s_0 \in \sto{R_0}
\end{equation}
since  for $s_0 \in \sto{R_0}$ we have
\begin{eqnarray*}
D'(s_0)  \sumd{D}
 & = &
(D'_1(s_0) + D'_2(s_0)) \sumd{D}\ 
\\[1mm] 
\mbox{by } (\ref{indpd:cond4}) & = &
\left(D_1(s_0)\frac{\sumd{D'_1}}{\sumd{D_1}} 
+ D_2(s_0)\frac{\sumd{D'_2}}{\sumd{D_2}}\right) \sumd{D} 
\\[1mm]
\mbox{by } (\ref{indpd:cond2}) & = &
D(s_0)\sumd{D'_1} + D(s_0)\sumd{D'_2}\ =
D(s_0)\sumd{D'}
\end{eqnarray*}
where we have assumed that $D_1 \neq 0$ and $D_2 \neq 0$; if say
$D_1 = 0$ then 
$D'_1 = 0$ 
(as each $\gamma_k$ is non-increasing by Lemma~\ref{lem:gamma-mult-nonincr})
and $D = D_2$ and $D' = D'_2$
in which case the claim follows directly from (\ref{indpd:cond4}).

From (\ref{indpd:cond5}) we get claim \ref{indpdlem2:outq0},
and are thus left with showing claim \ref{indpdlem2:indp} which is
that $R'$ and $R_0$ are independent in $D'$.
If $D'_1 = 0$ then $D' = D'_2$ and it follows from (\ref{indpd:cond3});
similarly if $D'_2 = 0$.
Otherwise, in which case also $\sumd{D_1} > 0$ and
$\sumd{D_2} > 0$, for $s' \in \sto{R'}$ and $s_0 \in \sto{R_0}$ we have
\begin{eqnarray*}
D'(\adds{s'}{s_0})\sumd{D'}
& = &
(D'_1(\adds{s'}{s_0}) + D'_2(\adds{s'}{s_0}))\sumd{D'}
\\[1mm]
\mbox{by } (\ref{indpd:cond3})
 & = &
D'_1(s')D'_1(s_0)\frac{\sumd{D'}}{\sumd{D'_1}}
  +
D'_2(s')D'_2(s_0)\frac{\sumd{D'}}{\sumd{D'_2}}
\\[1mm] 
\mbox{by } (\ref{indpd:cond4})
& = &
D'_1(s')D_1(s_0)\frac{\sumd{D'}}{\sumd{D_1}}
  +
D'_2(s')D_2(s_0)\frac{\sumd{D'}}{\sumd{D_2}}
\\[1mm] 
\mbox{by } (\ref{indpd:cond2})
& = &
D'_1(s')D(s_0)\frac{\sumd{D'}}{\sumd{D}}
  +
D'_2(s')D(s_0)\frac{\sumd{D'}}{\sumd{D}}
\\[1mm]
& = &
D'(s')D(s_0)\frac{\sumd{D'}}{\sumd{D}}
\\[1mm]
\mbox{by } (\ref{indpd:cond5})
& = &
D'(s')D'(s_0)
\end{eqnarray*}
\end{proof}

{\bf Lemma~\ref{lem:indpd2indpd}}
[rephrased using Definition~\ref{def:indpd-pres}]:
$\omega$ preserves probabilistic independence.

\begin{proof}
Let a slicing pair $(Q,Q_0)$ be given, and 
let $\chain{\gamma_k}{k}$ be defined as in Definition~\ref{def:gamma-k}.
By Lemma~\ref{lem:indpd2indpdk},
each element in the chain $\chain{\gamma_k}{k}$ 
preserves probabilistic independence;
by Proposition~\ref{prop:gamma-limit-omega}, 
it is sufficient to prove 
that also $\limit{k}{\gamma_k}$ preserves probabilistic independence.

With $(v,v')$ and $D$ given, let $D_k = \hpd{\gamma_k}{v}{v'}(D)$ for
each $k \geq 0$, and let $D' = \hpd{(\limit{k}{\gamma_k})}{v}{v'}(D)$.
Then we can 
%% use Lemma~\ref{lem:sumlim-limsum} 
establish each of the 3 claims about $D'$;
%% given our assumption that they hold for each $D_k$;
claim~\ref{indpdlem2:indp} follows from
the calculation
\begin{eqnarray*}
D'(\adds{s_1}{s_2})\sumd{D'} 
& = &
(\limit{k}{D_k(\adds{s_1}{s_2})})\,\sumd{\limit{k}{D_k}}
\\[1mm] 
\mbox{(Lemma~\ref{lem:sumlim-limsum})}
 & = &
(\limit{k}{D_k(\adds{s_1}{s_2})})\,(\limit{k}{\sumd{D_k}})
=
\limit{k}{(D_k(\adds{s_1}{s_2})\sumd{D_k})}
\\[1mm]
\mbox{(Lemma~\ref{lem:indpd2indpdk})} 
 & = &
\limit{k}{(D_k(s_1) D_k(s_2))}
= D'(s_1)D'(s_2)
\end{eqnarray*}
whereas claim~\ref{indpdlem2:outq} (claim~\ref{indpdlem2:outq0} is symmetric)
follows from the calculation
\begin{eqnarray*}
D(s)\sumd{D'} 
& = & D(s) \sumd{\limit{k}{D_k}}
\\[1mm]
\mbox{(Lemma~\ref{lem:sumlim-limsum})}
& = &
D(s) \limit{k}{\sumd{D_k}}
=
\limit{k}{(D(s)\sumd{D_k})}
\\[1mm]
\mbox{(Lemma~\ref{lem:indpd2indpdk})}
& = &
\limit{k}{(D_k(s) \sumd{D})}
= D'(s)\sumd{D}.
\end{eqnarray*}
\end{proof}

{\bf Lemma~\ref{lem:slicing-correct}}
For a given pCFG,
let $(Q,Q_0)$ be a slicing pair.
For all $k \geq 0$, all $(v,v') \in \PD$
with $R = \rv{Q}{v}$ and $R' = \rv{Q}{v'}$ and $R_0 = \rv{Q_0}{v}$,
all $D \in \Dist$
such that $R$ and $R_0$ are independent in $D$,
and all $\dlt \in \Dist$ such
that $\dagree{D}{\dlt}{R}$, we have
\[
\dagree{\hpd{\gamma_k}{v}{v'}(D)}{\cnst{k}{v}{v'}{D} \cdot 
\hpd{\Phi_k}{v}{v'}(\dlt)}{R'}.
\]
\begin{proof}
The proof is by induction in $k$, with an inner induction on
$\LAP{v}{v'}$.

Let us first (for all $k$) consider the case where
$\mathbf{v \notin Q \cup Q_0}$.
Thus, by Lemma~\ref{lem:notQ-outside}, $v$ stays outside $Q \cup Q_0$ until $v'$.
By Lemma~\ref{lem:gamma-as-omega} we see that
$\hpd{\gamma_k}{v}{v'} = \hpd{\omega}{v}{v'}$,
and we also see 
that $\hpd{\Phi_k}{v}{v'}(\dlt) = \dlt$
(by Definition~\ref{def:phipk} if $k = 0$, and
by Lemma~\ref{lem:phi1-id} otherwise).

Our proof obligation is thus, since $R' = R$, that
\[
\dagree{\hpd{\omega}{v}{v'}(D)}{1 \cdot 
\dlt}{R}.
\]
But from Lemma~\ref{lem:outsideQQ0irrelevant} we get
$\dagree{\hpd{\omega}{v}{v'}(D)}{D}{R}$,
and by assumption we have
$\dagree{D}{\dlt}{R}$, so the claim follows since
$\dagree{}{}{R}$ is obviously transitive.

If {\bf $\mathbf{k = 0}$ but $\mathbf{v}$ does not stay outside $\mathbf{Q \cup Q_0}$
until $\mathbf{v'}$} then
our proof obligation is
\[
\dagree{0}{\cnst{k}{v}{v'}{D} \cdot 0}{R'}
\]
which obviously holds (no matter what $\cnst{k}{v}{v'}{D}$ is).

We now consider $\mathbf{k > 0}$, in which case 
$\gamma_k = \HH{\unodes}(\gamma_{k-1})$ and
$\Phi_k = \HH{Q}(\Phi_{k-1})$,
and again consider several cases.

First assume that $\mathbf{v' = v}$, and thus $R' = R$. Then our obligation is
\[
\dagree{D}{1 \cdot \dlt}{R'}
\]
which follows directly from our assumptions.

Next assume that {\bf $\mathbf{v' \neq v''}$ with $\mathbf{v'' = \fppd{v}}$}.
Then, by Definition~\ref{def:HHX},
 $\hpd{\gamma_k}{v}{v'} = \hpd{\gamma_k}{v}{v''}\,;\,\hpd{\gamma_k}{v''}{v'}$
and with $D'' = \hpd{\gamma_k}{v}{v''}(D)$ we thus
have $\hpd{\gamma_k}{v}{v'}(D) = \hpd{\gamma_k}{v''}{v'}(D'')$; 
similarly, with
$\dlt'' = \hpd{\Phi_k}{v}{v''}(\dlt)$ we 
have $\hpd{\Phi_k}{v}{v'}(\dlt) = \hpd{\Phi_k}{v''}{v'}(\dlt'')$.
Since $\LAP{v}{v''} < \LAP{v}{v'}$ we can
apply the inner induction hypothesis to $(v,v'')$ and get
\[
\dagree{D''}{\cnst{k}{v}{v''}{D} \cdot \dlt''}{R''}
\]
where $R'' = \rv{Q}{v''}$. With $R''_0 = \rv{Q_0}{v''}$,
by Lemma~\ref{lem:indpd2indpdk} we moreover see that
$R''$ and $R''_0$ are independent in $D''$.
Hence we can apply the inner induction hypothesis to $(v'',v')$
to get
\[
\dagree{\hpd{\gamma_k}{v''}{v'}(D'')}{\cnst{k}{v''}{v'}{D''} \cdot 
\hpd{\Phi_k}{v''}{v'}(\cnst{k}{v}{v''}{D} \cdot \dlt'')}{R'}
\]
which since $\Phi_k$ is multiplicative (Lemma~\ref{lem:Phi-mult-nonincr})
amounts to
\[
\dagree{\hpd{\gamma_k}{v}{v'}(D)}{\cnst{k}{v''}{v'}{D''} \cdot \cnst{k}{v}{v''}{D} \cdot
\hpd{\Phi_k}{v}{v'}(\dlt)}{R'}
\]
which is as desired since
$\cnst{k}{v}{v'}{D} = \cnst{k}{v}{v''}{D} \cdot 
\cnst{k}{v''}{v'}{\hpd{\gamma_k}{v}{v''}(D)}$.

We are left with the situation that $v' = \fppd{v}$ (and $k > 0$),
with either $v \in Q$ or $v \in Q_0$; we now consider each of these
possibilities.

Assume {\bf $\mathbf{v \in Q_0}$ (and $\mathbf{k > 0}$ and $\mathbf{v' = \fppd{v}}$).}
Thus $v \notin Q$, and hence (by
Lemma~\ref{lem:notQ-outside}) $v$ stays outside $Q$ until $v'$ and thus
$R' = R$.
With $D' = \hpd{\gamma_k}{v}{v'}(D)$, 
by Lemma~\ref{lem:indpd2indpdk}
we see that
\begin{equation}
\label{correct:1}
D(s)\sumd{D'} = D'(s)\sumd{D}
\mbox{ for all } s \in \sto{R}.
\end{equation}
Since $v \notin Q$, and
$\Phi_{k} = \HH{Q}(\Phi_{k-1})$ as $k > 0$,
we see from clause~\ref{defn:evalk:notinX}
in Definition~\ref{def:HHX} that
$\hpd{\Phi_k}{v}{v'}(\dlt) = \dlt$.
Thus our proof obligation is 
\[
\dagree{D'}{\cnst{k}{v}{v'}{D} \cdot \dlt}{R}
\]
which (as we assume $\dagree{D}{\dlt}{R}$)
amounts to proving that for all $s \in \sto{R}$:
\[
D'(s) = \cnst{k}{v}{v'}{D} \cdot D(s)
\]
If $D = 0$ and hence $D' = 0$ then this is obvious.
Otherwise, Definition~\ref{def:slicing-const} stipulates
\[
\cnst{k}{v}{v'}{D} = \frac{\sumd{D'}}{\sumd{D}}
\]
and (\ref{correct:1}) yields the claim.

We shall finally consider the case
{\bf $\mathbf{v \in Q}$ (and $\mathbf{k > 0}$ and $\mathbf{v' = \fppd{v}}$).}
Thus $v \notin Q_0$, and hence (by
Lemma~\ref{lem:notQ-outside}) $v$ stays outside $Q_0$ until $v'$;
thus $\cnst{k}{v}{v'}{D} = 1$ 
so that our proof obligation,
with $D' = \hpd{\gamma_k}{v}{v'}(D)$ and $\dlt' = \hpd{\Phi_k}{v}{v'}(\dlt)$,
is to establish $\dagree{D'}{\dlt'}{R'}$,
that is $D'(s') = \dlt'(s')$ for all $s' \in \sto{R'}$.
We need a case analysis on the label of $v$,
where the case $\skipv$ is trivial.

If $\mathbf{\labv{v} = \observev{B}}$, we have 
$\fv{B} \subseteq R = \rv{Q}{v}$ and as $\defv{v} = \emptyset$ 
also $R' \subseteq R$; 
for $s' \in \sto{R'}$ this gives us the desired
\begin{eqnarray*}
\dlt'(s') & = & \sum_{s \in \sto{R}\ \mid\ \sagree{s}{s'}{R'}} \dlt'(s)\
  =\
\sum_{s \in \sto{R}\ \mid\ \sagree{s}{s'}{R'},\ \seme{B}s} \dlt(s)\
  =\
\sum_{s \in \sto{R}\ \mid\ \sagree{s}{s'}{R'},\ \seme{B}s} 
\hspace*{-7mm} D(s)\\[1mm]
& = &
\sum_{s \in \sto{R}\ \mid\ \sagree{s}{s'}{R'}} 
\hspace*{-5mm} D'(s)\
=\
D'(s').
\end{eqnarray*}

If {\bf $\mathbf{v}$ is a branching node},
with $B$ its condition,
and $v_1$ its $\true$-successor and $v_2$ its $\false$-successor,
with $D_1 = \selectf{B}(D)$ and $D_2 = \selectf{\neg B}(D)$ 
and $\dlt_1 = \selectf{B}(\dlt)$ and $\dlt_2 = \selectf{\neg B}(\dlt)$ 
the situation is that
$D' = D'_1 + D'_2$ and $\dlt' = \dlt'_1 + \dlt'_2$ where for each
$i \in \{1,2\}$, $D'_i$ and $\dlt'_i$ is computed as
\begin{itemize}
\item
if $\LAP{v_i}{v'} < \LAP{v}{v'}$ then
$D'_i =\hpd{\gamma_{k}}{v_i}{v'}(D_i)$
and
$\dlt'_i =\hpd{\gamma_{k}}{v_i}{v'}(\dlt_i)$;
\item
if $\LAP{v_i}{v'} \geq \LAP{v}{v'}$ then
$D'_i = \hpd{\gamma_{k-1}}{v_i}{v'}(D_i)$ and
$\dlt'_i = \hpd{\gamma_{k-1}}{v_i}{v'}(\dlt_i)$.
\end{itemize}
Let $R_1 = \rv{Q}{v_1}$ and $R_2 = \rv{Q}{v_2}$; 
thus (by Lemma~\ref{lem:rv-branch})
$R_1 \subseteq R$ and $R_2 \subseteq R$,
and also $\fv{B} \subseteq R$.
For each $i = 1,2$, we see
\begin{itemize}
\item
that $R_i$ is independent of $R_0$ in $D_i$
(by Lemma~\ref{lem:indpd-branch}),
\item
that $\dagree{D_i}{\dlt_i}{R_i}$
(by Lemma~\ref{lem:dist-partial} from $\dagree{D_i}{\dlt_i}{R}$ 
which holds as
for $s \in \sto{R}$ we have
$D_1(s) = \dlt_1(s)$ and $D_2(s) = \dlt_2(s)$
since if say $\seme{B}{s}$ is false
then the first equation amounts to $0 = 0$
and the second to $D(s) = \dlt(s)$),
\item
that $v_i$ stays outside $Q_0$ until $v'$,
and thus $\cnst{k}{v_i}{v'}{D_i} = 1$.
\end{itemize}
We now infer that for each $i = 1,2$ we have
$\dagree{D'_i}{\dlt'_i}{R'}$
(when $\LAP{v_i}{v'} < \LAP{v}{v'}$ this follows from the inner
induction hypothesis,
and otherwise it follows from the outer induction hypothesis on $k$).
For $s \in \sto{R'}$ we thus have
$D'_1(s') = \dlt'_1(s')$ and
$D'_2(s') = \dlt'_2(s')$, which implies
$D'(s') = \dlt'(s')$. This amounts to the desired
$\dagree{D'}{\dlt}{R'}$.

If {\bf $\mathbf{v}$ is a (random) assignment},
let $x = \defv{v}$ and first assume that
$x \notin R'$. Thus $R' \subseteq R$
so that $\dagree{D}{\dlt}{R'}$, and 
by Lemma~\ref{lem:assignRirrel} (\ref{lem:rassignRirrel}) we get
$\dagree{D}{D'}{R'}$ and 
$\dagree{\dlt}{\dlt'}{R'}$.
But this implies the desired $\dagree{D'}{\dlt'}{R'}$ since
for $s' \in \sto{R'}$ we have
$D'(s') = D(s') = \dlt(s') = \dlt'(s')$.

We can thus assume $x \in R'$
and first consider when $\labv{v} = \assignv{x}{E}$.
Then (by Lemma~\ref{lem:rv-assign}) $R = R'' \cup \fv{E}$
with $R'' = R' \setminus \mkset{x}$, so given $s' \in \sto{R'}$ we can use
Lemma~\ref{lem:assign-simp} twice to give us the desired
\begin{displaymath}
D'(s')\ =\ 
\hspace*{-12mm}
\sum_{s \in \sto{R} \ \mid\ \sagree{s}{s'}{R''},\ s'(x) = \seme{E}s} 
\hspace*{-7mm}
D(s)\ =\ \sum_{s \in \sto{R} \ \mid\ \sagree{s}{s'}{R''},\ s'(x) = \seme{E}s} 
\hspace*{-10mm}\dlt(s)\ =\ \dlt'(s').
\end{displaymath}
We next consider the case when $\labv{v} = \rassignv{x}{\psi}$.
Then $R = R' \setminus \mkset{x}$,
so given $s' \in \sto{R'}$ we can use
Lemma~\ref{lem:rassign-simp} twice to give us the desired
\[
D'(s') =  \psi(s'(x)) 
\left(\sum_{s \in \sto{R} \ \mid\ \sagree{s}{s'}{R}} D(s)\right) 
=
\psi(s'(x)) 
\left(\sum_{s \in \sto{R} \ \mid\ \sagree{s}{s'}{R}} \dlt(s)\right) 
=
\dlt'(s')
\]

\mbox{ }

\end{proof}

\subsection{Proofs for Section~\ref{sec:alg-least}}

{\bf Lemma~\ref{lem:DDclose}}
There exists an algorithm $\DDclose$
which given a node set $Q$ that is closed under data dependence,
and a node set $Q_1$,
returns the least set containing $Q$ and $Q_1$
that is closed under data dependence.
Moreover, assuming $\DDS$ is given,
$\DDclose$ runs in time $O(n \cdot |Q_1|)$.

\begin{proof}
We incrementally augment $Q$ as follows:
for each $v_1 \in Q_1$, and each $v \notin Q$,
we add $v$ to $Q$ iff $\DDS(v,v_1)$ holds.
This is necessary since any set containing $Q_1$ that is closed under 
data dependence must contain $v$; observe that
$Q$ will end up containing $Q_1$ since for all $v_1 \in Q_1$
we have $\DDS(v_1,v_1)$.

Thus the only non-trivial claim is that the resulting $Q$ will
be closed under data dependence. With $v \in Q$ we must show that if 
$\dd{v'}{v}$ then $v' \in Q$. If $v$ was in $Q$ initially, this follows
since $Q$ was assumed to be closed under data dependence.
Otherwise, assume 
that $v$ was added to $Q$ because for some $v_1 \in Q_1$
we have $\DDS(v,v_1)$. By correctness of $\DDS$ this means
that $\dds{v}{v_1}$ which implies $\dds{v'}{v_1}$ and thus
$\DDS(v',v_1)$ holds. Hence also $v'$ will be added to $Q$.
\end{proof}

{\bf Lemma~\ref{lem:PN}}
The function $\PN$ runs in time $O(n)$ and, given $Q$,
returns $C$ such that $C \cap Q = \emptyset$ and
\begin{itemize}
\item
if $C$ is empty then $Q$ provides next visibles
\item
if $C$ is non-empty then
all supersets of $Q$ that provide next
visibles will contain $C$.
\end{itemize}

\begin{proof}
It is convenient to introduce some terminology:
we say that $q \in Q \cup \{\finalv\}$ is $m$-next from $v$ iff there exists
a path $v = v_1 \ldots v_k = q$ with $k \leq m$ and
$v_j \notin Q$ for all $j$ with $1 \leq j < k$.
Also, we use superscript $m$ to denote the value of a variable when the
guard of the while loop is evaluated for the $m$'th time; note that
if $q = N^m(v)$ and $m' > m$ then also $q = N^{m'}(v)$.
A key part of the proof is to establish 3 facts, for each $m \geq 1$:
\begin{enumerate}
\item
\label{least-1}
if $q = N^m(v)$ then $q$ is $m$-next from $v$
\item
\label{least-2}
if $v \in F^m$ then $N^m(v) \neq \bot$  and no 
$q \in Q$ is $(m-1)$-next from $v$
\item
\label{least-3}
if $C^m = \emptyset$ and $q$ is $m$-next from $v$ then
\begin{enumerate}
\item
\label{least-3a}
$q = N^m(v)$, and
\item
\label{least-3b}
if $q$ is not $(m-1)$-next from $v$ then $v \in F^m$.
\end{enumerate}
\end{enumerate}
We shall prove the above facts simultaneously, by induction in $m$. 
The base case is when $m = 1$ and the facts follow from inspecting the
preamble of the while loop: for (\ref{least-1}),
if $q = N^1(v)$ then $q = v \in Q \cup \{\finalv\}$
and the trivial path $v$ shows that $q$ is 1-next from $v$;
for (\ref{least-2}), if $v \in F^1$ then $N^1(v) \neq \bot$ and no
$q$ can be 0-next from $v$;
for (\ref{least-3}), if $q$ is 1-next from $v$ then 
$q = v \in Q \cup \{\finalv\}$
in which case $q = N^1(v)$ and $v \in F^1$.

We now do the inductive case where $m > 1$.
Note that $C^{m-1} = \emptyset$ (as otherwise the loop would
have exited already).
For (\ref{least-1}), we assume that $q = N^{m}(v)$, and split into two 
cases: if $q = N^{m-1}(v)$ then we inductively infer that $q$
is $(m-1)$-next of $v$ and thus also $m$-next of $v$.
Otherwise, $v \notin Q$ and there exists an edge from $v$ to some
$v' \in F^{m-1}$ with $q = N^{m-1}(v')$.
Inductively, $q$ is $(m-1)$-next from $v'$.
But then, as $v \notin Q$, $q$ is $m$-next from $v$.

For (\ref{least-2}), we assume that $v \in F^m$ in which case
code inspection yields that $N^{m-1}(v) = \bot$
and that $N^m(v) = N^{m-1}(v')$ for some $v' \in F^{m-1}$
so inductively 
$N^m(v) \neq \bot$.
Also, no $q$ can be $m-1$-next from $v$, for if so then
we could inductively use (\ref{least-3a}) to infer 
$N^{m-1}(v) = q$.

For (\ref{least-3}), we assume that $C^m = \emptyset$ and $q$ is $m$-next
from $v$. We have two cases:
\begin{itemize}
\item
if $q$ is $(m-1)$-next from $v$ then (\ref{least-3b}) holds vacuously,
and as $C^{m-1} = \emptyset$ we infer inductively that
$q = N^{m-1}(v)$ and thus $q = N^m(v)$ which is (\ref{least-3a}).
\item
if $q$ is not $(m-1)$-next from $v$,
we can inductively use (\ref{least-1}) to get
$q \neq N^{m-1}(v)$, and also infer that there is a path
$v v' \ldots q$ where $v \notin Q$ and $q$ is $(m-1)$-next from $v'$
but not $(m-2)$-next from $v'$.
As $C^{m-1} = \emptyset$ we inductively infer that
$q = N^{m-1}(v')$ and $v' \in F^{m-1}$.
Thus the edge from $v$ to $v'$ has been considered in the recent iteration,
and since $C_m = \emptyset$ it must be the case that $N^{m-1}(v) = \bot$
so we get the desired $q = N^{m}(v)$ and $v \in F^m$.
\end{itemize}
We are now ready to address the claims in the lemma.
From (\ref{least-2}) we see that each node gets into $F$
at most once and hence the running time is in $O(n)$.
That $C \cap Q = \emptyset$ follows since only nodes not in $Q$
get added to $C$.
Next we shall prove that
\begin{quote}
if $C$ is empty then $Q$ provides next visibles
\end{quote}
and thus consider the situation where
for some $m$, $C^m = \emptyset$ and $F^m = \emptyset$.

We shall first prove that for all 
$v \in \unodes$, all $q \in Q \cup \{\finalv\}$,
and all $k \geq 1$ we have that if $q$ is $k$-next from $v$ then
$N^{m}(v) = q$. To see this, we may wlog.\ assume that $k$ is chosen
as small as possible, that is, $q$ is not $(k-1)$-next from $v$.
It is impossible that $k > m$ since then there would be a path
$v \ldots v' \ldots q$ where $q$ is $m$-next from $v'$
but not $(m-1)$-next from $v'$ which by (\ref{least-3b})
entails $v' \in F^m$ which is a contradiction.
Thus $k \leq m$ and $q$ is $m$-next from $v$ so
(\ref{least-3a}) yields the claim.

Now let $v \in \unodes$ be given, to show that $v$ has a next visible
in $Q$. Since there is a path from $v$ to $\finalv$
there will be a node $q \in Q \cup \{\finalv\}$ such that (for some $k$)
$q$ is $k$-next from $v$. By what we just proved,
$N^{m}(v) = q$ and we shall show that $q$ is a next visible
in $Q$ of $v$. Thus assume, to get a contradiction,
that we have a path not containing $q$ from $v$ to a node in 
$Q \cup \{\finalv\}$. Then there exists
$q' \neq q$ and $k'$ such that $q'$ is $k'$-next from $v$.
Again applying what we just proved,
$N^{m}(v) = q'$ which is a contradiction.

Finally, we shall prove that
\begin{quote}
if $C$ is non-empty \\
then all supersets of $Q$ that provide next
visibles will contain $C$
\end{quote}
and thus consider the situation where
for some $m$, $C^m \neq \emptyset$.
It is sufficient to consider $v \in C^m$ (and thus $v \notin Q$)
and prove that 
if $Q \subseteq Q_1$ where $Q_1$ provides next visibles then
$v \in Q_1$.

Since $v \in C^m$, the situation is that there is an edge
from $v$ to some $v' \in F^{m-1}$ with $q \neq q'$
where $q = N^m(v)$ and $q' = N^{m-1}(v')$.
From (\ref{least-1}) we see that 
$q$ is $m$-next from $v$, and that $q'$ is $(m-1)$-next from
$v'$.
%%  and thus $m$-next from $v$.
That is, there exists a path $\pi$ from $v$ to $q$ and a path
$\pi'$ from $v$ to $q'$. 
Since $q,q' \in Q_1 \cup \{\finalv\}$ and $Q_1$ provides next visibles,
there exists $v_0 \in Q_1$ that occurs on both paths.
If $v_0 \neq v$ 
then $q$ and $q'$ are both $(m-1)$-next from $v_0$
which since $C^{m-1} = \emptyset$ implies $q = N^{m-1}(v_0) = q'$
which is a contradiction.
Hence $v_0 = v$ which amounts to the desired $v \in Q_1$.
\end{proof}

{\bf Lemma~\ref{lem:lws-correct}}
The function $\LWS$, given $\hat{Q}$,
returns $Q$ such that 
\begin{itemize}
\item
$Q$ is a weak slice set
\item
$\hat{Q} \subseteq Q$
\item
if $Q'$ is a weak slice set with $\hat{Q} \subseteq Q'$
then $Q \subseteq Q'$.
\end{itemize}
Moreover, assuming $\DDS$ is given, 
$\LWS$ runs in time $O(n^2)$.

\begin{proof}
We shall establish the following loop invariant:
\begin{itemize}
\item
$Q$ is closed under data dependence;
\item
$Q$ includes $\hat{Q}$ and 
is a subset of any weak slice set that includes $\hat{Q}$.
\end{itemize}
This holds before the first iteration,
by the properties of $\DDclose$.

We shall now argue that each iteration preserves the invariant.
This is obvious for the part about $Q$ being closed under data dependence
and including $\hat{Q}$. Now assume that $Q'$ is a weak slice
set that includes $\hat{Q}$; we must prove that $Q \subseteq Q'$ holds after
the iteration. We know that before the iteration we have $Q \subseteq Q'$,
and also $C = \PN(Q) \neq \emptyset$ so we know from Lemma~\ref{lem:PN} 
(since $Q'$ provides next visibles) that
$C \subseteq Q'$; hence we can apply Lemma~\ref{lem:DDclose} to
infer (since $Q'$ is closed under data dependence)
$\DDclose(Q,C) \subseteq Q'$ which yields the claim.

When the loop exits, with $C = \emptyset$,
Lemma~\ref{lem:PN} tells us that
$Q$ provides next visibles.
Together with the invariant, this yields the desired
correctness property.

The loop will terminate, as $Q$ cannot keep increasing;
the total number of calls to $\PN$ is in $O(n)$.
By Lemma~\ref{lem:PN} we see that the total time spent
in $\PN$ is in $O(n^2)$.
And by Lemma~\ref{lem:DDclose} we infer that the total time spent
in $\DDclose$ is in $O(n^2)$.
Hence the running time of $\LWS$ is in $O(n^2)$.
\end{proof}

{\bf Theorem~\ref{thm:BSP-correct}}
The algorithm $\BSP$ returns,
given a pCFG and a set of nodes $\ESS$,
sets $Q$ and $Q_0$
such that
\begin{itemize}
\item
$(Q,Q_0)$ is a slicing pair wrt.~$\ESS$
\item
if $(Q',Q'_0)$ is a slicing pair wrt.$\ESS$
then $Q \subseteq Q'$.
\end{itemize}
Moreover, $\BSP$ runs in time $O(n^3)$
(where $n$ is the number of nodes in the pCFG).

\begin{proof}
As stated in Figure~\ref{fig:best-slicing-pair},
we shall use the following invariants for the while loop:
\begin{enumerate}
\item
\label{best-1}
$Q$ is a weak slice set
\item
\label{best-2}
$F$ is a weak slice set
\item
\label{best-3}
$\finalv \in Q \cup F$
\item
\label{best-4}
$W \subseteq \ESS$
\item
\label{best-5}
if $v \in W$ then $Q_v \cap Q = \emptyset$
\item
\label{best-6}
if $v \in \ESS$ but $v \notin W$ then $v \in Q \cup F$
\item
\label{best-7}
if $(Q',Q'_0)$ is a slicing pair wrt.~$\ESS$
then $Q \cup F \subseteq Q'$.
\end{enumerate}
We shall first show that the invariants are established by the
loop preamble, which is mostly trivial;
for (\ref{best-2},\ref{best-3},\ref{best-7}) we use
Lemma~\ref{lem:lws-correct}.

Let us next show that the invariants are preserved by each iteration
of the while loop.
For (\ref{best-1}) this follows from Lemma~\ref{lem:weak-union},
the repeated application of which and Lemma~\ref{lem:lws-correct}
gives (\ref{best-2}).
Code inspection easily gives (\ref{best-3},\ref{best-4},\ref{best-5}).
To show (\ref{best-6}) we do a case analysis: either $v$ was not in $W$
before the iteration so that (by the invariant) $v$ belonged to
$Q \cup F$ and thus $v \in Q$ by the end of the iteration,
or $v$ was removed from $W$ during the iteration in which
case $Q_v \subseteq F$ and thus (Lemma~\ref{lem:lws-correct})
$v \in F$.

For (\ref{best-7}), let $(Q',Q'_0)$ be a slicing pair wrt.~$\ESS$;
we know that $Q \cup F \subseteq Q'$ holds before the iteration
and thus $Q \subseteq Q'$ holds before the members of $W$ are processed.
It is sufficient to prove that 
if $v \in W$ with $Q_v \cap Q \neq \emptyset$
then $Q_v \subseteq Q'$. The invariant tells us that $v \in \ESS$,
so as $(Q',Q'_0)$ is a slicing pair wrt.~$\ESS$
we infer that $v \in Q'$ or $v \in Q'_0$;
by Lemma~\ref{lem:lws-correct} this
shows $Q_v \subseteq Q'$ (as desired) or $Q_v \subseteq Q'_0$ which
we can rule out: for then we would have 
$Q_v \cap Q' = \emptyset$ and 
thus $Q_v \cap Q = \emptyset$ before $W$ is processed
which contradicts our assumption.

The while loop will terminate since $W$ keeps getting smaller which
cannot go on infinitely, and if an iteration does not make $W$ smaller
then it will have $F = \emptyset$ at the end and the loop exits.

When the loop exits, with $F = \emptyset$, we have:
\begin{itemize}
\item
$Q$ is a weak slice set with $\finalv \in Q$,
by invariants (\ref{best-1}) and (\ref{best-3});
\item
$Q_0$ is a weak slice set, by Lemmas~\ref{lem:lws-correct} and
\ref{lem:weak-union};
\item
$Q \cap Q_0 = \emptyset$, by invariant (\ref{best-5});
\item
if $v \in \ESS$ then $v \in Q \cup Q_0$ since
\begin{itemize}
\item
if $v \notin W$ then $v \in Q$ by invariant (\ref{best-6})
(and $F = \emptyset$),
\item
if $v \in W$ then $v \in Q_0$ by construction of $Q_0$.
\end{itemize}
\end{itemize}
Thus $(Q,Q_0)$ is a slicing pair.
If $(Q',Q'_0)$ is another slicing pair 
we see from invariant (\ref{best-7}) that
$Q \cup F \subseteq Q'$ and thus $Q \subseteq Q'$.

Finally, we can address the running time.
By Lemma~\ref{lem:DDS-n3}, we can compute $\DDS$ in time
$O(n^3)$. Then Lemma~\ref{lem:lws-correct} tells us that
each call to $\LWS$ takes time in $O(n^2)$.
As there are $O(n)$ such calls, this shows the the code
in $\BSP$ before the while loop runs in time $O(n^3)$.
The while loop iterates $O(n)$ times, with each iteration
processing $O(n)$ members of $W$; as each such processing
(taking intersection and union) can be done in time $O(n)$
this shows that the while loop runs in time $O(n^3)$.
The total running time is thus in $O(n^3)$.
\end{proof}

\end{document}